\documentclass[11pt]{article}
\pdfoutput=1

\usepackage{epsfig}
\usepackage{latexsym}
\usepackage{enumerate}
\usepackage{url}
\usepackage{float}
\usepackage{texnansi}
\usepackage{xcolor}
\definecolor{LinkBlue}{rgb}{0,0,1}
\usepackage[pdfborder={0 0 0}, colorlinks=true, linkcolor=LinkBlue, urlcolor=LinkBlue, citecolor=LinkBlue]{hyperref}
\usepackage{tikz}
\usepackage[margin=10pt,font=small,labelfont=bf]{caption}
\usepackage[subrefformat=parens,labelformat=parens]{subcaption}
\usepackage{afterpage}
\usepackage{enumitem}
\usepackage[boxed]{algorithm}
\usepackage{algpseudocode}
\usepackage[normalem]{ulem}
\usepackage{lmodern}
\usepackage{booktabs}
\usepackage{sectsty}
\usepackage{ifthen}
\usepackage{amsmath}
\usepackage{amsthm}
\usepackage{amssymb}
\usepackage{amsfonts}
\usepackage{thmtools}
\usepackage{thm-restate}
\usepackage{mathtools}
\usepackage{xspace}
\usepackage{titling}
\usepackage{natbib}
\usepackage{xfrac}
\usepackage{multirow}
\usepackage{bigdelim}
\usepackage{bm}
\usepackage{comment}
\usepackage[capitalise]{cleveref}
\usepackage[textsize=tiny]{todonotes}
\newcommand{\field}[1]{\ensuremath{\mathbb{#1}}}
\newcommand{\R}{\ensuremath{\field{R}}} 
\newcommand{\I}[1]{\ensuremath{\mathbb{I}_{\left\{#1\right\}}}} 
\newcommand{\PR}{\ensuremath{\mathsf{P}}} 
\newcommand{\E}{\ensuremath{\mathsf{E}}} 
\newcommand{\defeq}{\ensuremath{\triangleq}}
\newcommand{\subjectto}{\text{\rm subject to}} 

\newcommand{\Lscr}{\ensuremath{\mathcal L}}

\newcommand{\Xscr}{\ensuremath{\mathcal X}}
\newcommand{\Yscr}{\ensuremath{\mathcal Y}}

\DeclarePairedDelimiter\norm{\lVert}{\rVert}

\DeclareMathOperator{\supp}{supp}

\DeclareMathOperator*{\argmin}{\mathrm{argmin}}

\newcommand{\minimize}{\ensuremath{\mathop{\mathrm{minimize}}\limits}}
\newcommand{\maximize}{\ensuremath{\mathop{\mathrm{maximize}}\limits}}
\declaretheoremstyle[headfont=\sffamily\bfseries,bodyfont=\itshape]{thm-sf}
\declaretheoremstyle[headfont=\sffamily\bfseries,bodyfont=\normalfont]{thm-sf-remark}

\declaretheorem[style=thm-sf,numberwithin=section]{theorem}

\declaretheorem[style=thm-sf-remark,sibling=theorem]{remark}
\declaretheorem[style=thm-sf,sibling=theorem]{assumption}
\crefname{assumption}{assumption}{assumptions}
\declaretheorem[style=thm-sf,sibling=theorem]{definition}
\declaretheorem[style=thm-sf-remark,sibling=theorem]{example}

\declaretheorem[style=thm-sf,sibling=theorem]{corollary}
\declaretheorem[style=thm-sf,sibling=theorem]{lemma}

\declaretheorem[style=thm-sf,sibling=theorem]{proposition}

\renewcommand{\thmcontinues}[1]{\hyperref[#1]{continued}}

\usetikzlibrary{arrows,patterns,plotmarks,pgfplots.groupplots}
\tikzstyle{every picture} += [>=stealth]
\tikzset{axis/.style={semithick, line join=miter}}
\allsectionsfont{\sffamily}
\makeatletter
\def\@seccntformat#1{\csname the#1\endcsname.\quad}
\makeatother

\floatname{algorithm}{\normalfont\sffamily\bfseries Algorithm}
\floatstyle{ruled}
\newcommand{\emailhref}[1]{\href{mailto:#1}{\tt #1}} 
\provideboolean{fastcompile}
\newcommand{\hidefastcompile}[1]{\ifthenelse{\boolean{fastcompile}}{}{#1}}
\usepackage{pgfplots}
\usepackage{pgfplotstable}
\usetikzlibrary{calc}
\usepackage{mathtools}
\definecolor{orange}{rgb}{0.85,0.33,0.13} 
\definecolor{green}{rgb}{0.13,0.85,0.33}
\definecolor{purple}{rgb}{0.33,0.13,0.85}
\definecolor{lime}{rgb}{0.65,0.85,0.13}
\definecolor{blue}{rgb}{0.13,0.65,0.85}
\pgfplotscreateplotcyclelist{tricolor}{%
  orange,every mark/.append style={fill=orange!80!black},mark=*\\%
  green,every mark/.append style={fill=green!80!black},mark=square*\\%
  purple,every mark/.append style={fill=purple!80!black},mark=otimes*\\%
  black,mark=star\\%
  orange,every mark/.append style={fill=orange!80!black},mark=diamond*\\%
  green,densely dashed,every mark/.append style={solid,fill=green!80!black},mark=*\\%
  purple,densely dashed,every mark/.append style={solid,fill=purple!80!black},mark=square*\\%
  black,densely dashed,every mark/.append style={solid,fill=gray},mark=otimes*\\%
  orange,densely dashed,mark=star,every mark/.append style=solid\\%
  green,densely dashed,every mark/.append style={solid,fill=green!80!black},mark=diamond*\\%
}
\pgfplotsset{colormap={tricolormap}{color=(orange) color=(green) color=(purple)},
  colormap={quadcolormap}{color=(orange) color=(lime) color=(blue) color=(purple)}}
\pgfplotstableset{%
  font=\small,
  every head row/.style={before row=\toprule[1pt], after row=\midrule},
  every last row/.style={after row=\bottomrule[1pt]}}
\pgfplotsset{compat=1.15}

\provideboolean{submissionversion}
\setboolean{submissionversion}{true}

\ifthenelse{\boolean{submissionversion}}{
  \renewcommand{\todo}[2][1]{}

  \newcommand{\CCM}[1]{#1}
}{

  \newcommand{\CCM}[1]{{\color{orange}{[CCM: #1]}}}
}

\usepackage{setspace}
\usepackage[margin=1in]{geometry}

\title{\vspace{-1.0em}\bfseries\sffamily
  Risk-Based Auto-Deleveraging\thanks{NH and CCM are supported by the Briger Family Digital
    Finance Lab at Columbia University. CCM is a research advisor for Paradigm and Uniswap Labs. SC and MN are supported by a  Columbia University CDFT Research Grant. SC is further supported by an NSERC Postdoctoral Fellowship (PDF‑599675-2025). MN is further supported by NSF Grants DMS-2407074 and DMS-2106056.}
}
\author{
  Steven Campbell\thanks{These authors contributed equally to this work.}  \\
  Columbia University \\
  \emailhref{sc5314@columbia.edu}
  \and
  Natascha Hey\protect\footnotemark[2] \\
  Columbia University \\
  \emailhref{nyh2112@columbia.edu}
  \and
  Ciamac C.~Moallemi \\
  Columbia University \\
  \emailhref{ciamac@gsb.columbia.edu}
  \and
  Marcel Nutz \\
  Columbia University \\
  \emailhref{mnutz@columbia.edu}
}
\date{\today\vspace{-3em}}

\begin{document}
\maketitle

\ifthenelse{\boolean{submissionversion}}{}{

}

\vspace{3em}
\begin{abstract}
Auto-deleveraging (ADL) mechanisms are a critical yet understudied component of risk management
on cryptocurrency futures exchanges. When available margin and other loss-absorbing
resources are insufficient to cover losses following large price moves, exchanges reduce
positions and socialize losses among solvent participants via rule-based ADL protocols.

We formulate ADL as an optimization problem that minimizes the exchange's risk of loss arising
from future equity shortfalls. In a single-asset, isolated-margin setting, the \emph{minimax
  leverage} policy --- minimizing the maximum leverage among participants --- is optimal for all
monotone risk measures. This policy has a transparent structure: positions are reduced first for
the most highly levered accounts, and leverage is progressively equalized via a water-filling (or
``leverage-draining'') rule. The policy is distribution-free, wash-trade resistant, Sybil
resistant, and path-independent.  It provides a canonical and implementable benchmark for ADL
design and clarifies the economic logic underlying queue-based mechanisms used in practice.

We further study the multi-asset, cross-margin setting, where the ADL problem becomes genuinely
multi-dimensional: the exchange must allocate reductions across accounts with portfolios exposed
to correlated price moves.  Under the expected loss objective, asset-level shadow prices separate
the problem across accounts, yielding a scalable numerical method.  Naive gross leverage misleads
here, ignoring within-portfolio hedging. When prices are driven by a single risk factor, the
optimal policy is again water-filling, but in a factor-adjusted leverage, so better-hedged
portfolios are deleveraged less.

We apply the framework to the October~10, 2025 Hyperliquid ADL event. Relative to the exchange's
realized allocation, our risk-minimizing allocations achieve lower expected shortfall.
\end{abstract}


\newpage

\begingroup
\makeatletter
\setcounter{tocdepth}{2}
\renewcommand{\baselinestretch}{0.98}\selectfont
\renewcommand*\l@section[2]{%
  \ifnum \c@tocdepth >\z@
    \addpenalty\@secpenalty
    \addvspace{.9em \@plus .9\p@}%
    \setlength\@tempdima{1.5em}%
    \begingroup
      \parindent \z@
      \rightskip \@pnumwidth
      \parfillskip -\@pnumwidth
      \leavevmode \bfseries
      \advance\leftskip\@tempdima
      \hskip -\leftskip
      #1\nobreak\hfil\nobreak\hb@xt@\@pnumwidth{\hss #2}\par
    \endgroup
  \fi}
\tableofcontents
\makeatother
\endgroup
\newpage

\section{Introduction}\label{sec:intro}

Perpetual futures exchanges in cryptocurrency markets allow for extreme leverage and continuous
trading, making them vulnerable to large price shocks.  When rapid price moves trigger liquidation
cascades, insurance funds can be exhausted, and exchanges are forced to reduce the exposure of
otherwise solvent traders through auto-deleveraging (ADL).  Such events are not hypothetical:
during the market-wide shock of October~10--11,~2025, multiple major venues activated ADL
mechanisms to preserve solvency, as reported by
\citet{chitra2025adl}.  Months later, on January 30, 2026,
over \$2.56 billion in leveraged positions were liquidated within a single trading day \citep{lang2026_crypto_volatility_yahoo}.

In these regimes, when a participant has insufficient account equity, their positions must be
unwound. Ideally, this is done through liquidation mechanisms that sell the positions to other
market participants at mutually agreeable prices, thereby unwinding distressed positions and containing losses within
defaulting accounts.
However, in
times of extreme duress, there may be no willing buyers in the market, or the prices that are bid
may leave the exchange insolvent. In these cases, ADL mechanisms unwind the positions of a
distressed participant by forcibly closing the positions of solvent participants on the other side
of the market. These forced position closures occur at prices that leave the exchange solvent. In
this way, ADL is a mechanism both to reduce positions and to socialize potential losses across market
participants.

Similar position reduction and loss-allocation problems arise in traditional financial
infrastructure. Central clearinghouses rely on default waterfalls to mutualize losses, but extreme
stress can overwhelm these mechanisms. During the March~2022 nickel crisis at the London Metal
Exchange (LME), rapidly rising prices threatened the exhaustion of default resources, see \citep{heilbron2025lme}. Rather than
invoking an explicit end-of-waterfall mechanism, the exchange suspended trading and retroactively
canceled executed trades --- an opaque intervention followed by legal challenges and lasting
damage to market confidence.

On perpetual futures exchanges, where forcible position reduction is implemented through ADL,
margining may be applied on an isolated basis, where collateral is posted separately for positions
in different assets, or under cross-margining, where collateral is pooled across positions in
different assets and liquidation decisions become coupled. In practice, most exchanges, including
BitMEX \citep{bitmexADL}, Hyperliquid \citep{hyperliquidADL}, and Binance \citep{binanceADL} use a
``queue-based'' methodology, first proposed by BitMEX, that ranks accounts according to a priority
metric which is the product of percentage profit-and-loss (P\&L) and leverage of the account.
Accounts then have their positions unwound sequentially
in decreasing order of priority until the requisite quantity has been
unwound. \citet{chitra2025adl} estimates that 95\% of
perpetual trading volume occurs on exchanges that use this queue-based ADL policy.  Despite
their central role in market stability, such ADL mechanisms are largely heuristic and have
rarely been studied as explicit quantitative design objects.

In this paper, we formalize auto-deleveraging as an explicit risk-minimization problem faced by an
exchange once the equity of an account has been exhausted.  In our formulation, the exchange
defines a loss function, which is the total equity shortfall aggregated across all accounts.  The
exchange then seeks to deleverage in a manner that minimizes the risk of future loss, where risk
is quantified according to some risk measure.  Risk-based objectives of this form are standard in
financial risk management, where coherent risk measures provide a principled framework for
evaluating tail exposure under uncertainty \citep{artzner1999coherent}. Among these, conditional
value-at-risk (CVaR) \citep{rockafellar2000cvar} and, more generally, spectral risk measures
\citep{acerbi2002spectral}, offer tractable and economically interpretable criteria that capture
extreme losses while preserving convexity.

Our contributions are as follows:
\begin{itemize}
\item We begin in the single-asset, isolated-margining setting. Here, when the exchange is
  risk-neutral and minimizes expected loss (i.e., expected total future account equity shortfall),
  we show that the unique optimal deleveraging policy seeks to minimize the maximum leverage
  across all participants in the system. This \emph{minimax leverage} policy is equivalent to a
  \emph{leverage water-filling} (more accurately ``leverage-draining'') or threshold rule that equalizes post-intervention leverage across
  affected accounts.  This characterization is analytical and identifies leverage as the natural
  variable governing optimal position reduction in this setting.

\item We next consider the case where the exchange is risk-averse and minimizes the conditional
  value-at-risk (CVaR) of the terminal aggregate account equity shortfall. Here, we show that the
  minimax leverage policy remains optimal, but is no longer unique. Instead, the set of optimal
  policies is characterized by a cutoff in leverage: if there exist accounts with leverage above
  the cutoff, positions of those accounts can be reduced in an arbitrary way until all are at the cutoff, at which point the minimax leverage policy must be applied. This
  leverage cutoff depends on the confidence level $\beta$ of the CVaR objective. We further establish that the minimax leverage policy remains optimal for the more general
  class of monotone risk measures.

\item We show that the minimax leverage policy has a number of
  desirable properties: it is distribution-free, wash-trade resistant, Sybil resistant, and path-independent. In these ways, the minimax leverage policy is a natural and principled
  deleveraging policy.
  
\item We then turn to the multi-asset, cross-margining setting, where the simple
  one-dimensional leverage ordering breaks down because default risk depends on the joint
  distribution of returns and on cross-asset hedges within each portfolio. We show that, under an
  expected-loss objective, the exchange's problem is still tractable. Introducing an asset-by-asset
  vector of shadow prices separates the optimization across accounts and yields an efficient
  numerical approach for the general case.
  Moreover, we identify a sharp analytic benchmark: when prices are driven
  by a single dominant market factor, the optimal deleveraging rule again has a water-filling
  structure, but in a modified factor-adjusted leverage rather than gross leverage, so that hedged portfolios
  are penalized less. 

\item We give an efficient numerical method for the general cross-margin problem. Exploiting the separable structure of the expected-loss objective, we propose a sample-average ADMM scheme \citep{boyd2011distributed} that decomposes across accounts at each iteration, allowing fully parallel per-account updates. On a stylized multi-asset example, the method recovers the factor water-filling structure predicted by the theory. We then adapt the scheme to exchange scale, with refinements for the heterogeneous-asset, sparse-portfolio regime.

\item We instantiate the method on the October~10, 2025 Hyperliquid ADL event, comparing the exchange's realized allocation against our risk-minimizing allocations. The model-based allocations achieve lower expected shortfall, and the difference is structural: the realized policy tends to fully close out affected accounts with little regard to leverage, whereas the risk-based allocations reduce each account only partially and in order of leverage.
\end{itemize}

At a practical level, the minimax leverage policy has important similarities to and differences from
the BitMEX queue-based ADL policy that is currently used in practice. The policies are
similar in that, in both cases, accounts are ranked according to a priority metric: in the minimax
leverage policy, this priority is the account leverage, while in the queue-based policy, it is the
product of percentage P\&L and leverage. Further, in the minimax leverage policy, unwinds occur
across accounts in a water-filling fashion according to the priority metric, while in the
queue-based policies, unwinds occur in sequence\footnote{I.e., the second-ranked account will be delevered only after the first-ranked is completely depleted, and so on.} according to the priority metric. The similarities
highlight that the minimax leverage policy is practically implementable, while the differences
highlight suboptimality in the status quo. Indeed, beyond the fact that the BitMEX policy cannot
directly be motivated by principled considerations, it satisfies none of the wash-trade
resistance, Sybil resistance, and path independence properties enjoyed by the minimax leverage
policy.

The choice of exchange shortfall risk as the ADL objective also deserves some discussion. In the
single-asset, isolated-margin setting, one possible criticism is that this criterion may overly
focus on the exchange's exposure to the side being delevered. If the market has moved down
sharply, for example, and as a result short positions are delevered, one may argue that these
short accounts already have significant profits and therefore pose little additional risk to the
exchange. However, by definition, deleveraging is constrained to the shorts in this
situation. Further, minimizing risk is one principled way of determining allocations, as opposed
to the ad hoc methods currently employed in practice. Moreover, the
criticism is weaker in the multi-asset, cross-margin setting, where accounts hold portfolios with
offsetting exposures across assets and there is no simple notion of a uniformly ``winning'' side
of the market. There, unrealized gains in one component of the portfolio need not imply low future
shortfall risk once the joint distribution of price moves and the account's residual exposures are
taken into account. More broadly, while alternative criteria are conceivable, it is not clear what
objective besides risk the exchange should optimize once ADL has been triggered.  For this reason,
we view exchange shortfall risk as a natural and tractable starting point for the design of ADL
rules. To the best of our knowledge, multi-asset ADL has not been discussed in the prior
literature.

\paragraph{Literature review.}
The design of default management and loss allocation mechanisms for central clearing
counterparties (CCPs) such as futures exchanges has received sustained attention in the
literature, but almost exclusively in a modular fashion, with different strands focusing on
distinct layers of the default waterfall. While margining, default resource sizing, and
auction-based close-out procedures have been studied in detail, the final stage of the waterfall
--- loss allocation once prefunded resources are exhausted --- has not been analyzed
systematically from an optimality or mechanism-design perspective.

The broader CCP literature examines several of these upstream layers. \citet{duffie2011centralclearing} provide a foundational analysis of when central clearing reduces counterparty risk, emphasizing the role of multilateral netting and the trade-offs involved in concentrating exposures at a CCP. \citet{menkveld2015crowded} identifies crowded positions among clearing members as an overlooked systemic risk and introduces tail-risk measures that account for correlated losses across members; \citet{huangmenkveldyu2021ccpexposure} complement this with an empirical decomposition of CCP exposure dynamics in stressed markets. From a capital-structure perspective, \citet{huang2019ccpcapital} studies how CCP capitalization interacts with for-profit incentives in the design of default resources. These strands address counterparty risk reduction, exposure measurement, and capital adequacy, but stop short of analyzing the rule by which residual losses are allocated once prefunded resources are exhausted.

Indeed, the final stage of the default waterfall has been discussed primarily in qualitative and
policy-oriented work. The end-of-waterfall procedure advocated by the International Swaps and
Derivatives Association is variation margin gains haircutting (VMGH), which reallocates residual
losses to clearing members whose positions accrued gains during the liquidation period
\citep{isda2013waterfall}. Surveys by \citet{domanski2015central} and \citet{armakola2015default}
document current practices and argue that this ``winner-pays'' mechanism is robust and typically
sufficient. \citet{Cont2015} provides a detailed discussion of end-of-waterfall tools, including
VMGH, contract tear-ups, and assessments, emphasizing incentive effects, legal feasibility, and
market confidence. \citet{Duffie2014} similarly analyzes CCP resolution from a systemic risk
perspective, framing the trade-off between continuity of clearing and contagion. However, neither
work proposes a formal model or optimization problem for the allocation of residual losses.

A recent exception is the work of \citet{chitra2025adl},
who provides a formal analysis of ADL mechanisms in perpetual futures
markets. Importantly, Chitra studies ADL under a different definition than considered here: Chitra
models ADL as an ex-post loss-socialization rule that seeks to recover an exchange equity deficit
by applying equity or profit
haircuts
to accounts. This is different from the formulation considered here, where ADL is defined as the allocation of forced position reductions to manage post-event risk.\footnote{In the recent (v3) version of \citep{chitra2025adl}, this difference is now highlighted as follows (Section~2.4 \emph{ibid.}, see also ``corrections'' on p.\,9): ``Production ADL is executed in \emph{contract space}: the engine
selects positions by a ranking score and forces contract-level reductions. The theoretical
analysis in this paper \cite[i.e.,][]{chitra2025adl} is
written in \emph{wealth space} (equity haircuts and haircutable endowment).''} Within his framework, \citet{chitra2025adl} casts haircut-based ADL as a mechanism-design problem and proves impossibility results showing that no ADL rule can
simultaneously satisfy a full set of natural desiderata, including revenue,
fairness, and solvency preservation. However,
\citet{chitra2025adl} does not
directly describe how positions should be reduced,
which is the focus of this paper. On the other hand, in our framework, if the positions are transferred at a fair market price, there is no change in equity and hence no loss socialization. In this way, the two papers consider different variations of ADL and are not directly comparable.  The variation we consider here is consistent with the way that ADL is currently implemented by perpetual futures exchanges in practice.

Complementing this formal literature, a recent discussion by \citet{jia2026hyperliquidthread}
provides an on-chain forensic analysis of Hyperliquid's October~10,~2025 liquidation event. They
analyze the overall economics of the full liquidation pipeline, from initial liquidations to
backstop takeovers and subsequent ADL unwinds.  They observe that delevered short positions were ex
post profitable since they were bought in at relative market lows. This challenges the notion of
ADL as a loss socialization mechanism. However, individual ADL outcomes are heterogeneous because
there were several waves of ADL and queue position mattered. 
While not a formal model of
ADL, \citep{jia2026hyperliquidthread} is useful institutional evidence. In particular, it reinforces the importance of modeling
the contract space allocation rule itself and of distinguishing forced position reduction from
ex post wealth transfers, which can be hard to determine during the time of a crisis.

The water-filling rule is also related to the axiomatic
literature on claims, bankruptcy, and taxation problems, surveyed by
\citet{Thomson2003,Thomson2015}. In the base claims model, agents have
claims \(c_i\) on an insufficient endowment \(E\), and a rule selects awards
\(a_i\) satisfying \(0 \leq a_i \leq c_i\) and \(\sum_i a_i = E\). Canonical
rules include the proportional rule, the constrained equal awards rule, the
constrained equal losses rule, and the Talmud rule, together with weighted
and priority versions. This literature studies such rules through axioms
including order preservation, composition, consistency, duality, and no
advantageous splitting or merging. As detailed in Appendix~\ref{app:claims_correspondence}, there is an exact
connection between this framework and the water-filling rule in our
single-asset isolated-margin model: under a change of variables from
buybacks to residual positions, the water-filling allocation corresponds to
the weighted constrained equal awards rule with weights proportional to
account equity. The interpretation, however, differs from the classical claims
model. In the claims literature, weights are typically primitive priority or
entitlement parameters, whereas in our setting they arise endogenously from
risk minimization considerations.
Thus our contribution in the isolated-margin case is not the abstract
water-filling formula itself, but the derivation of the particular weighted
claims rule from an exchange shortfall-risk minimization problem.

The multi-asset cross-margin setting is not covered by the claims literature.
Although the latter contains extensions with uncertainty,
multi-dimensional claims, network structure, and non-homogeneous endowments
\citep{Thomson2015}, our cross-margin ADL problem has different primitives:
signed vector positions, a scalar collateral pool for each account, vector
reduction constraints, and stochastic losses
driven by the joint distribution of asset prices. In particular, hedging and
factor exposure are central to the cross-margin problem, whereas the base
claims model is a scalar rationing problem. Thus, the single-asset
water-filling rule can be identified with a weighted claims rule, but the
general cross-margin ADL problem is not a special case of the claims framework.


\section{Single-Asset Isolated Margining}\label{sec:isolated.margining}
We consider a single traded asset under isolated margining and focus on a stress scenario following
a large price move. Without loss of generality, as a result of this move, a collection of long accounts becomes
insolvent and is liquidated, generating an aggregate exposure of size $Q>0$ units that must be absorbed
by the rest of the system. ADL reallocates this exposure to short accounts
by forcing them to reduce their positions.

We model this reallocation at a fixed execution price $p_\tau>0$ at time $\tau$.\footnote{In practice, this price is often chosen so that
  the exchange remains solvent. It is not necessarily a fair market price, and as a consequence,
  ADL may also transfer profits or losses to the short accounts. The present work focuses on the
  risk exposure of the exchange, not on how losses are socialized
  \citep[cf.][]{chitra2025adl}.} The exchange considers
all short accounts, indexed by $i=1,\dots,n$, and forces each of them to buy back a quantity
$x_i\ge 0$, thereby reducing their short positions. Each short account $i$ is characterized by its
position size $q_i>0$, entry price $p_i^{(e)}$, and posted margin $m_i\ge 0$. Clearly, feasibility
of $x\defeq (x_1,\dots,x_n)$ requires
\begin{align*}
\sum_{i=1}^n x_i = Q \qquad \text{and} \qquad
0 \le x_i \le q_i \quad \text{for} \quad i=1,\dots,n, 
\end{align*}
and we define the feasible set as
\begin{equation*}
\mathcal X \defeq \Big\{x\in\R^n:\ \sum_{i=1}^n x_i = Q,\ \ 0\le x_i\le q_i\ \text{for all }i\Big\}.
\end{equation*}

\begin{assumption}[Feasibility]\label{ass:feasible}
$Q < \sum_{i=1}^n q_i$.
\end{assumption}
Note that we can extend to the case $Q = \sum_{i=1}^n q_i$, but in this case the problem is trivial as there is only one feasible allocation.

\subsection{Equity, Losses, and Leverage}
For any price level $p$, the post-ADL equity of account $i$ is defined as
\begin{equation*}
e_i(x_i,p)
\;\defeq\;
q_i\big(p_i^{(e)} - p\big)
\;-\;
x_i\big(p_\tau - p\big)
\;+\;
m_i.
\end{equation*}
This expression equals the equity the account would have had at price $p$ absent ADL, plus the
realized P\&L from buying back $x_i$ units at $p_\tau$ rather than at $p$. Evaluated at $p=p_\tau$, the equity is
\begin{equation}\label{eq:def_Ei}
E_i \defeq e_i(x_i,p_\tau)=q_i\big(p_i^{(e)}-p_\tau\big)+m_i.
\end{equation}
In particular, at the time of intervention, equity is unaffected by the allocation $x_i$, whereas ADL reshapes the sensitivity of future equity to subsequent price movements. We will assume throughout that each candidate account for deleveraging is solvent.
\begin{assumption}[Solvency at the ADL time]\label{ass:solvent}
$E_i>0$ for all $i$.
\end{assumption}

The exchange incurs losses whenever account equity becomes negative. Aggregating across all short accounts,
the total exchange loss due to short accounts that would be incurred at a future time $T > \tau$ when the price is given by $p_T=p$ is
\begin{equation}\label{eq:exchange_loss}
\Lscr(x,p)
\;\defeq\;
\sum_{i=1}^n \big(-e_i(x_i,p)\big)_+
\end{equation}
where $(u)_+\defeq \max\{u,0\}$. We write the corresponding per-account shortfall as
\begin{equation}\label{eq:per.account.shortfall}
\sigma_i(x_i,p)\defeq \big(-e_i(x_i,p)\big)_+.
\end{equation}

A convenient state variable at the ADL time $\tau$ is the post-ADL leverage of each account, defined as the ratio of notional exposure to equity,
\begin{equation}\label{eq:def_leverage}
\ell_i(x_i)
\;\defeq\;
\frac{p_\tau(q_i-x_i)}{e_i(x_i,p_\tau)}
=
\frac{p_\tau(q_i-x_i)}{q_i(p_i^{(e)}-p_\tau)+m_i}=\frac{p_\tau(q_i-x_i)}{E_i}.
\end{equation}
Under Assumption~\ref{ass:solvent}, $\ell_i$ is affine and strictly decreasing in $x_i$ on $[0,q_i]$, with
\begin{equation*}
\ell_i(0)=\frac{p_\tau q_i}{E_i},\qquad \ell_i(q_i)=0.
\end{equation*}
While equity determines solvency,
leverage captures the sensitivity of future losses to price movements and will be the natural
variable in which optimal deleveraging policies are expressed.

\subsection{Risk-Based ADL}

\begin{figure}[tbh]
  \centering
  \begin{tikzpicture}[>=stealth, thick]
      \draw[->] (0,0) -- (12cm,0) node[right] {Time $t$};

      \coordinate (tau-tick) at (2cm,0);
      \draw ($(tau-tick) + (0,2mm)$) -- ($(tau-tick) - (0,2mm)$);
      \node[above=1mm] (tau) at (tau-tick) {$\tau$};
      \node[below=3mm, align=center, font=\small] (tau-desc) at (tau-tick) {ADL trigger price $p_\tau$};

      \coordinate (T-tick) at (9cm,0);
      \draw ($(T-tick) + (0,2mm)$) -- ($(T-tick) - (0,2mm)$);
      \node[above=1mm] (T) at (T-tick) {$T$};
      \node[below=3mm, align=center, font=\small] (T-desc) at (T-tick) {Close-out price $p_T$};

      \node[above=5mm, align=center, orange, font=\small] (tau-label) at (tau.north) {Exchange allocates\\buybacks $x_i$};
      \draw[->, orange] (tau-label.south) -- (tau.north);

      \node[above=5mm, align=center, font=\small] (T-label) at (T.north) {Exchange realizes\\loss $\Lscr(x, p_T)$};
      \draw[->] (T-label.south) -- (T.north);

    \end{tikzpicture}
\caption{Timeline of the ADL allocation problem. At the trigger time~$\tau$, the exchange observes the price $p_\tau$ and chooses the buyback allocation $x=(x_1,\dots,x_n)$. At the terminal time~$T>\tau$, the close-out price $p_T$ is realized and the exchange incurs the equity shortfall $\Lscr(x,p_T)$, whose risk the allocation $x$ is chosen to minimize.}
\label{fig:adl_timeline}
\end{figure}

We propose a risk-based design principle for ADL, namely to choose the buyback allocation $x$ which minimizes a
risk measure of the exchange's equity shortfall at the time horizon~$T>\tau$.  Given a model for the (random) price~$p_T$ at the terminal time~$T$, the exchange evaluates the risk of the terminal loss $\Lscr(x,p_T)$ through a risk measure $\rho(\cdot)$ and chooses an allocation $x\in\mathcal{X}$ that solves
\begin{equation}\label{eq:risk_based_adl}
\begin{array}{ll}
\minimize_{x\in\R^n} & \rho\big(\Lscr(x,p_T)\big) \\
\subjectto           & x\in\mathcal X.
\end{array}
\end{equation}
The two-stage structure of this problem is summarized in Figure~\ref{fig:adl_timeline}: at the ADL trigger time~$\tau$ the exchange observes the price $p_\tau$ and commits to the buyback allocation $x$, while the resulting equity shortfall is only realized at the terminal time~$T$ once $p_T$ has materialized.

We pause to discuss several modeling choices that frame the analysis.
\begin{itemize}
\item The formulation in~\eqref{eq:risk_based_adl} is static: the allocation $x$ is chosen once at the trigger time $\tau$ and held fixed until $T$, with no trading by other market participants in between, no intermediate intervention by the exchange, and no re-liquidation of accounts whose equity may turn negative before $T$. In practice, ADL is one step in a dynamic risk-management process and may need to be repeated as market conditions evolve. We adopt the static formulation because it isolates the central allocation question --- given a target quantity $Q$ that must be unwound at $\tau$, how should it be distributed across accounts? --- and admits tractable convex analysis. We view the resulting policies as a principled benchmark against which dynamic refinements can be compared.

\item Alternative objectives are conceivable: once ADL has been triggered, the exchange could in principle aim to maximize social welfare, minimize wealth transfers between accounts \citep[in the spirit of][]{chitra2025adl}, or maximize its own revenue. We focus on risk minimization because ADL is invoked precisely in tail scenarios, where the exchange's solvency is the dominant concern, and risk measures such as expected loss and CVaR offer a principled and economically interpretable framework for quantifying that tail exposure. Risk-based objectives also have the practical advantage of yielding tractable convex problems with structurally transparent optima, as we develop in the remainder of the paper.

\item 
  The loss $\mathcal{L}$ accounts only
  for the equity shortfall associated with short positions, and ignores (solvent) long
  accounts. This restriction is
  without loss of generality for the allocation problem~\eqref{eq:risk_based_adl}. That is easiest to see in the risk-neutral case ($\rho=\E$), because any loss contribution from long positions is independent of the allocation variable $x$ and therefore enters the objective only as an additive constant. The proposed allocation remains optimal even for the more general risk-based objectives considered below; see Proposition~\ref{prop:wf.optimal.monotone}.
\end{itemize}

\subsection{Minimizing Expected Loss}

In this subsection, we specialize the risk-based ADL problem~\eqref{eq:risk_based_adl} to the case
where $\rho=\E$, i.e., the exchange minimizes its expected loss:
\begin{equation}\label{eq:opt_rn}
\begin{array}{ll}
\minimize_{x\in\R^n} & \displaystyle V(x)\defeq \E\!\left[\Lscr(x,p_T)\right] \\
\subjectto & x\in\mathcal X.
\end{array}
\end{equation}
We will see that this principled approach leads to a particular rule $x^\star\in\mathcal{X}$ that is straightforward to implement and has numerous desirable properties.

Under Assumption~\ref{ass:p_T_regular}, the objective $V(x)$ is finite and convex in $x$ (see
Lemma~\ref{lem:shortfall_prop} in Appendix~\ref{app:isolated.margining}).  To ensure that
problem~\eqref{eq:opt_rn} is well posed, we assume that the terminal price~$p_T$ is
integrable. For simplicity of presentation, we further assume that the distribution of $p_T$
admits a density with sufficiently large support:

\begin{assumption}[Regularity of $p_T$]\label{ass:p_T_regular}
  $p_T$ is an integrable random variable that  admits a density $f_T$
  with $f_T(p)>0$ for $ p\ge p_\tau$.
\end{assumption}

The following theorem summarizes several key insights. First, minimizing expected losses for the exchange is equivalent to minimizing the maximal post-ADL leverage across accounts; cf.~(b). Solving this minimax leverage problem is further equivalent to using a threshold-rule on leverage: for a certain leverage threshold $t$, all accounts with initial leverage exceeding~$t$ are deleveraged down to $t$, whereas accounts with initial leverage below~$t$ remain untouched; cf.~(c). The threshold $t$ is set so that the total allocation across accounts equals the target buyback amount~$Q$. In fact, this policy not only minimizes the expected loss, but even minimizes the loss in any realization of $p_T$; cf.~(d). Any of these formulations leads to the same solution $x^\star\in\mathcal{X}$, which we call the \emph{water-filling rule} or \emph{minimax leverage policy}. 

\begin{theorem}[Minimax
leverage policy]\label{thm:risk_neutral}
Let Assumptions~\ref{ass:feasible}--\ref{ass:p_T_regular} hold. For a feasible allocation $x\in\mathcal{X}$, the following are equivalent:
\begin{enumerate}
    \item[(a)] $x$ minimizes the expected loss, i.e., solves~\eqref{eq:opt_rn}.
    \item[(b)] $x$ minimizes the maximal post-ADL leverage, i.e., solves
        \begin{equation}\label{eq:minmax}
        \minimize_{x\in\mathcal X} \displaystyle \maximize_{1\le i\le n}\ \ell_i(x_i).
        \end{equation}
    \item[(c)] $x$ is a threshold rule on leverage, i.e., there exists a threshold $t\in[0,\max_i\ell_i(0)]$ such that
        \begin{equation*}
        \begin{aligned}
        &x_i>0 &&\Longrightarrow\quad \ell_i(x_i)=t,\\
        &x_i=0 &&\Longrightarrow\quad \ell_i(0)\le t.
        \end{aligned}
        \end{equation*}
    \item[(d)] $x$ is a simultaneous pointwise minimizer of realized loss, i.e.,
    \[
    \mathcal{L}(x,p) \le \mathcal{L}(y,p) \qquad \text{for all}\quad y\in\mathcal X,\ p\in\mathbb R.
    \]
    \end{enumerate}
The unique optimal allocation for (a)--(d) is the \emph{water-filling rule} $x^\star=(x^\star_1,\dots,x^\star_n)\in\mathcal{X}$ defined by
\[
x^\star_i\;\defeq\;\Big(q_i-\frac{E_i}{p_\tau}\,t^\star\Big)_+,\qquad i=1,\dots,n,
\]
where the leverage threshold $t^\star>0$ is the unique root of the equation
\begin{equation}\label{eq:root_main_thm}
    \sum_{i=1}^n \Big(q_i-\frac{E_i}{p_\tau}\,t\Big)_+=Q. 
\end{equation}
\end{theorem}
As detailed in Lemma~\ref{lem:G_properties_unique_level} in \Cref{sec:minimax}, the function $G(t)
\defeq \sum_{i=1}^n (q_i-\frac{E_i}{p_\tau}\,t)_+$ appearing in~\eqref{eq:root_main_thm} is continuous and strictly decreasing on $[0,\max_i \ell_i(0)]$. Thus, it is straightforward to compute the root~$t^\star$. All proofs for this section are reported in Appendix~\ref{app:isolated.margining}.

Theorem~\ref{thm:risk_neutral} underscores that the optimal ADL allocation for minimizing expected losses induces a \emph{leverage equalization principle}. It reallocates exposure by leveling down the highest post-ADL leverages first, until they reach the common leverage threshold. As illustrated in Figure~\ref{fig:water_filling_two_thresholds}, this corresponds to a water-filling (or leverage-draining) rule: Enumerate the accounts in decreasing initial leverage, so that $\ell_{1}(0) \ge \ell_{2}(0) \ge\dots\ge\ell_{n}(0)$. Then as the leverage cap decreases from $\ell_{1}(0)$ downward, initially only the most levered account is forced to buy back.
Once its leverage reaches $\ell_{2}(0)$, the top two accounts are deleveraged jointly so as to keep their leverages equal, and so on. The procedure stops when the target buyback amount~$Q$ is reached.

\begin{figure}[tbh]
\centering

\begin{minipage}[t]{0.48\textwidth}
\centering
\begin{tikzpicture}
\begin{axis}[
    width=\textwidth,
    height=0.88\linewidth,
    ybar stacked,
    bar width=18pt,
    ymin=0, ymax=8.8,
    xmin=0.6, xmax=3.8,
    xtick={1.0,1.8,2.6,3.4},
    xticklabels={1,2,3,4},
    xlabel={Account $i$},
    ylabel={Leverage},
    axis lines=left,
    tick align=outside,
    ymajorgrids,
    grid style={dotted, gray!35},
    tick label style={font=\small},
    yticklabel style={opacity=0}, 
    label style={font=\small},
    clip=false,
    legend cell align={left},
    legend style={
        at={(0.5,1.05)},
        anchor=south,
        legend columns=2,
        draw=none,
        fill=none,
        font=\small,
    },
    legend image code/.code={%
        \draw[draw=black,#1] (0cm,-0.09cm) rectangle (0.34cm,0.09cm);
    },
]

\addlegendimage{fill=blue!25}
\addlegendentry{Leverage removed by ADL}
\addlegendimage{fill=blue!80}
\addlegendentry{Leverage post-ADL}


\addplot+[draw=black, fill=blue!80] coordinates
    {(1.0,6.0) (1.8,5.0) (2.6,4.0) (3.4,2.0)};
\addplot+[draw=black, fill=blue!25] coordinates
    {(1.0,2.0) (1.8,0.0) (2.6,0.0) (3.4,0.0)};

\draw[black, dashed, thick] (axis cs:0.7,6.0) -- (axis cs:3.7,6.0);
\node[font=\small, anchor=south west] at (axis cs:3.52,6.0) {$t^\star$};

\node[font=\small, anchor=south] at (axis cs:1.0,8.08) {$\ell_1(0)$};
\node[font=\small, anchor=south] at (axis cs:1.8,5.08) {$\ell_2(0)$};
\node[font=\small, anchor=south] at (axis cs:2.6,4.08) {$\ell_3(0)$};
\node[font=\small, anchor=south] at (axis cs:3.4,2.08) {$\ell_4(0)$};

\end{axis}
\end{tikzpicture}

\end{minipage}
\hfill
\begin{minipage}[t]{0.48\textwidth}
\centering
\begin{tikzpicture}
\begin{axis}[
    width=\textwidth,
    height=0.88\linewidth,
    ybar stacked,
    bar width=18pt,
    ymin=0, ymax=8.8,
    xmin=0.6, xmax=3.8,
    xtick={1.0,1.8,2.6,3.4},
    xticklabels={1,2,3,4},
    xlabel={Account $i$},
    axis lines=left,
    tick align=outside,
    ymajorgrids,
    grid style={dotted, gray!35},
    tick label style={font=\small},
    yticklabel style={opacity=0}, 
    label style={font=\small},
    clip=false,
]


\addplot+[draw=black, fill=blue!80, forget plot] coordinates
    {(1.0,3.5) (1.8,3.5) (2.6,3.5) (3.4,2.0)};
\addplot+[draw=black, fill=blue!25, forget plot] coordinates
    {(1.0,4.5) (1.8,1.5) (2.6,0.5) (3.4,0.0)};

\draw[black, dashed, thick] (axis cs:0.7,3.5) -- (axis cs:3.7,3.5);
\node[font=\small, anchor=south west] at (axis cs:3.52,3.5) {$t^\star$};

\node[font=\small, anchor=south] at (axis cs:1.0,8.08) {$\ell_1(0)$};
\node[font=\small, anchor=south] at (axis cs:1.8,5.08) {$\ell_2(0)$};
\node[font=\small, anchor=south] at (axis cs:2.6,4.08) {$\ell_3(0)$};
\node[font=\small, anchor=south] at (axis cs:3.4,2.08) {$\ell_4(0)$};

\end{axis}
\end{tikzpicture}

\end{minipage}

\caption{The minimax leverage policy equalizes leverage by water-filling (``leverage-draining'') from the initial leverage $\ell_i(0)$ down to the threshold $t^\star$. A larger total buyback quantity $Q$ leads to a lower threshold $t^\star$ and more accounts being affected (right).}
\label{fig:water_filling_two_thresholds}
\end{figure}

The intuition connecting minimization of expected losses with minimization of maximal leverage is as follows. A buyback at the execution price $p_\tau$ does not change the equity $E_i$ at time~$\tau$, but it does reduce the residual short position $q_i-x_i$ and hence increases the account's ``distance to default.'' Indeed, as derived in \eqref{eq:ze_price}--\eqref{eq:ze_price_leverage}, for $x_i<q_i$ the account incurs a shortfall only when the terminal price exceeds its zero-equity (bankruptcy) level
\[
p_\tau+\frac{E_i}{q_i-x_i} = p_\tau\big(1+\ell_i(x_i)^{-1}\big),
\]
so higher post-ADL leverage $\ell_i(x_i)$ corresponds to
a larger region of the right tail of $p_T$ over which the exchange bears losses from that account. Consequently, the marginal reduction in expected loss from increasing $x_i$ is governed by how much the buyback shifts the tail threshold. 
This marginal benefit is strictly larger (in absolute value) for accounts with higher leverage. It follows that whenever two accounts can be adjusted at the margin, any allocation that buys back from a relatively low-leverage account while leaving a strictly higher-leverage account unaddressed can be improved by reallocating an infinitesimal unit of buyback toward the more levered account. 
As a result, an expected-loss minimizer must first remove exposure from the most levered accounts, which is exactly the minimax formulation~\eqref{eq:minmax}.

\subsection{Properties of the Minimax Leverage Policy}\label{sec:minimax-props}
This section studies the properties of the water-filling allocation obtained in Theorem~\ref{thm:risk_neutral}. An obvious but important feature is that the policy is \emph{distribution-free}, i.e., does not require a model for the terminal price~$p_T$ --- the allocation depends only on the state
$\{(q_i,E_i)\}_{i=1}^n$ at the ADL time~$\tau$ (as well as $Q,p_\tau$),
and is therefore independent of both the horizon~$T$ and the distribution of~$p_T$.

A related property is that the \emph{leverage ordering of accounts is invariant to the execution
  price}. That is, the ranking of accounts by initial leverage, which determines the
water-filling priority, does not depend on the price~$p_\tau$ at which ADL is executed.

\begin{proposition}[Price-independence of leverage ordering]\label{prop:leverage_ordering}
For any two accounts $i$ and $j$, define the leverage of account~$i$ at price~$p$ (before ADL) as
\[
\ell_i(p) \defeq \frac{p\, q_i}{E_i(p)}, \qquad E_i(p) \defeq q_i(p_i^{(e)} - p) + m_i.
\]
Then, for any prices $p,p'>0$ at which both accounts are solvent, $\ell_i(p) \geq \ell_j(p)$ if and only if $\ell_i(p') \geq \ell_j(p')$.
\end{proposition}

As a consequence, the water-filling priority order used in Theorem~\ref{thm:risk_neutral} is invariant to the choice of execution price~$p_\tau$.

Additionally, unlike ADL rules that sort accounts using realized or unrealized P\&L measures, the water-filling rule is \emph{wash-trade resistant}. The reason is that the relevant state variable for minimax leverage is the pair $(q_i,E_i)$, where $E_i$ already includes the mark-to-market P\&L at the execution price~$p_\tau$. Consequently, any trade prior to $\tau$ that leaves the position size $q_i$ unchanged merely reallocates equity between unrealized P\&L at the entry price and margin. It therefore leaves $E_i$ unchanged, and hence does not affect the initial leverage $\ell_i(0)=p_\tau q_i/E_i$ that determines the water-filling priority.

The next two sections derive two further important properties of the minimax leverage policy. First, Sybil resistance, meaning that traders cannot evade deleveraging by operating through multiple, smaller accounts. Second, path-independence, meaning that two subsequent ADL events with quantities $Q_1$ and $Q_2$ lead to the same outcome as a single ADL event with quantity $Q_1+Q_2$. We also show that path-independence, together with leverage-priority, characterizes the minimax leverage policy among all ADL policies.

\begin{remark}\label{rem:queue-is-inconsistent}
The BitMEX-style queue-based ADL policy (see Section~\ref{sec:intro}) fails to have some of these desirable properties:
    \begin{itemize}
    \item Because the BitMEX ranking is based on unrealized P\&L, the priority ordering of accounts can change depending on the execution price, unlike the leverage ordering used by the minimax leverage policy (\Cref{prop:leverage_ordering}).

    \item Because the BitMEX policy deleverages the highest-ranked account before targeting any lower-ranked account, it is not Sybil resistant. For instance, the most levered position in the system could be split across two accounts: one with even higher leverage and one with sufficiently low leverage to appear third (or lower) in the queue. An ADL event of a given size (e.g., equal to the size of the original position) would then deleverage the first two accounts in the queue but leave the third untouched, whereas without splitting the same event could have fully deleveraged the original account.

        \item Prioritizing the (initially) top-ranked account also induces path dependence. Under two successive ADL events, the first event may deleverage the initially top-ranked account enough to lower its rank, so that the second event targets a different account first; if instead the same aggregate quantity were executed in a single event, the entire quantity would have been applied to the initially top-ranked account.

        \item Finally, because queue priority depends explicitly on unrealized P\&L, an account could lower its queue priority via a wash trade.

    \end{itemize}
\end{remark}

\subsubsection{Sybil Resistance}
When accounts are anonymous, a participant can split their aggregate position and collateral across multiple accounts. If deleveraging allocations depend on per-account quantities, this creates a ``Sybil'' manipulation channel that can potentially reduce the participant's total forced buyback without changing the underlying economic exposure. Clearly, it is preferable to use ADL mechanisms that are \emph{Sybil resistant}: splitting an account into several accounts with the same combined position and equity should not decrease the participant's total buyback. Theorem~\ref{thm:sybil_waterfilling} below states that the minimax leverage policy is Sybil resistant.

To formalize this result, we define account splits and Sybil resistance mathematically. Fix a set of
\emph{non-attacker} short accounts indexed by $\mathcal N$, with position and equity parameters $\{(q_j,E_j)\}_{j\in\mathcal N}$. We view this as the relevant account state at time $\tau$ since, up to scaling by the price, $\ell_i(0)=p_\tau q_i/E_i$ depends only on these two quantities. An \emph{attacker} controls an aggregate pair $(q^A,E^A)$ with $q^A\ge 0$ and $E^A>0$. For a split into $K$ accounts, denote the attacker index set by
$\mathcal A_K=\{1,\dots,K\}$ and let the full economy index set be $\mathcal I_K\defeq\mathcal N\cup\mathcal A_K$. Then $\mathcal A_1$ represents the attacker in the \emph{unsplit} (single-account) configuration, holding
$(q^A,E^A)$. The next definition formalizes that the attacker can split her holdings $q^A$ and equity $E^A$ arbitrarily over her $K$ Sybil accounts.

\begin{definition}[Sybil split]\label{def:sybil-split}
A \emph{Sybil split} of $(q^A,E^A)$ is a collection $\{(q_k,E_k)\}_{k=1}^K\in (\R_+\times\R_{++})^K$ of attacker-controlled accounts such that $\sum_{k=1}^K q_k=q^A$ and $\sum_{k=1}^K E_k=E^A$.
\end{definition}

Next, we formally define Sybil resistance. For a given ADL mechanism, let $x^{K}\in\R_+^{\mathcal I_K}$ denote the resulting allocation
for the scenario with Sybil split, i.e., $x_i^{K}$ is the buyback assigned to account~$i$ and $\sum_{i\in\mathcal I_K} x_i^{K}=Q$. Thus, the attacker’s \emph{total} buyback is
\[
X_A^{K}\defeq \sum_{k\in\mathcal A_K} x_k^{K}.
\]

\begin{definition}[Sybil resistance]\label{def:sybil-resistance}
An ADL mechanism is \emph{Sybil resistant} if splitting an account cannot reduce the attacker’s total buyback. That is, $X_A^{K} \ge X_A^{1}$ for every attacker aggregate state $(q^A,E^A)\in\R_+\times\R_{++}$ and every Sybil split
$\{(q_k,E_k)\}_{k=1}^K$ of $(q^A,E^A)$.
\end{definition}

We can now state the formal result.

\begin{theorem}\label{thm:sybil_waterfilling}
The minimax leverage policy is Sybil resistant.
\end{theorem}

While the proof in Appendix~\ref{app:pf.sybil.waterfilling} proceeds via the subadditivity of the function defining the leverage threshold~\eqref{eq:root_main_thm}, the key intuition for Theorem~\ref{thm:sybil_waterfilling} is the very essence of the minimax leverage policy: it focuses on the maximal leverage, which cannot be reduced by splitting. To see this, let
\[
\ell^A \defeq \frac{p_\tau q^A}{E^A}
\qquad\text{and}\qquad
\ell_k \defeq \frac{p_\tau q_k}{E_k}
\]
be the pre-ADL leverages of the unsplit attacker and sub-account $k$, respectively.
The following shows that $\ell^A$ is a convex combination of $\{\ell_k\}_{k=1}^K$ with weights
$w_k\defeq E_k/E^A$:
\[
\ell^A
= \frac{p_\tau q^A}{E^A}
= \frac{p_\tau\sum_{k=1}^K q_k}{\sum_{k=1}^K E_k}
= \sum_{k=1}^K \frac{E_k}{\sum_{j=1}^K E_j}\,\frac{p_\tau q_k}{E_k}
= \sum_{k=1}^K w_k\,\ell_k.
\]
In particular, $\max_{1\le k\le K}\ell_k \ge \ell^A$, meaning that the \emph{maximal leverage among the Sybil accounts is never lower than the original leverage.}

We observe that while splitting accounts is not beneficial under the minimax leverage policy, the total amount deleveraged is not necessarily invariant --- it may increase.

\begin{remark}
Splitting accounts can lead to a strictly larger amount bought back from the attacker. For instance, this would occur if an account with leverage just below the threshold $t^\star$ is split so that one sub-account becomes the highest leveraged account in the system. Equivalently, merging several accounts can lead to a strict decrease in the amount deleveraged. As a result, the water-filling rule indirectly incentivizes account aggregation.
\end{remark}

\subsubsection{Path-Independence, Leverage Priority, and Axiomatic Characterization}\label{sec:path_indep}

In this section, we first observe that the minimax leverage policy has two properties that we call path-independence and leverage-priority. Then, we show that these two properties can serve as an axiomatic characterization of the minimax leverage policy among all possible ADL mechanisms. For brevity, we limit ourselves to informal statements in this section and report the mathematical details in Section~\ref{sec:Formalization_TC_and_proof_of_axiamatioc_char}. In particular, we refer to Assumptions~\ref{ass:time-consist-F-isolated} and~\ref{ass:leverage-priority-F-isolated} for precise versions of the following two definitions.

\begin{definition}[Path-independence]\label{def:path_indep_informal}
    An ADL mechanism is path-independent if two successive applications with buy-back quantities~$Q_1$ and~$Q_2$ lead to the same outcome as a single application with buy-back quantity $Q_1+Q_2$.
\end{definition}

\begin{definition}[Leverage-priority]\label{def:leverage_priority_informal}
    An ADL mechanism satisfies leverage-priority if, for sufficiently small buy-back quantity $Q>0$, it only affects the account(s) with maximal initial leverage.
\end{definition}

Using the water-filling representation of the minimax leverage policy, it is not hard to see that it satisfies both properties, which is one implication of Theorem~\ref{thm:tc-waterfilling-isolated} below. Intuitively, path-independence holds because water-filling to the target level $t^\star$ and then continuing to the target level~$t^{\star\star}$ yields the same as directly water-filling to~$t^{\star\star}$ (and one verifies that deleveraging~$Q_1$ and then~$Q_2$ yields the same eventual target $t^{\star\star}$ as directly deleveraging $Q_1+Q_2$). Of course, leverage-priority is immediate from the water-filling representation.

The main result of the theorem is the reverse implication: if an ADL mechanism satisfies path-independence and leverage-priority, it must be the minimax leverage policy. To state such a result rigorously, Section~\ref{sec:Formalization_TC_and_proof_of_axiamatioc_char} formalizes a general ADL mechanism mathematically as follows. Consider the state space $\mathcal S \defeq (\R_+\times\R_{++})^n$ of all possible account states $s=(q_i,E_i)_{i=1}^n$. Each state encodes a short position size $q_i\ge 0$
and an equity level $E_i>0$ for each account~$i$. Then an ADL mechanism is a family
\[
\{F_Q:\mathcal S\to\mathcal S\}_{Q\ge 0}
\]
of maps on $\mathcal S$, where $F_Q(s)$ is the post-ADL state after a total buyback quantity $Q$ has been allocated across accounts. Naturally, we only consider mechanisms that reduce (but never increase) existing positions (see Assumption~\ref{ass:market-clearing-F-isolated} for details).
We then have the following axiomatic characterization of the minimax leverage policy.

\begin{theorem}[Characterization by path-independence and leverage-priority]\label{thm:tc-waterfilling-isolated} The minimax leverage policy satisfies path-independence and leverage-priority. Conversely, if any ADL mechanism satisfies these two properties, then it must be the minimax leverage policy.
\end{theorem}

The intuition behind the axiomatic characterization is as follows. Split the amount $Q$ into many (infinitesimally) small bits. By path-independence, the mechanism for $Q$ yields the same result as consecutively buying back the small bits. By leverage-priority, the first bit is allocated to the account with the highest initial leverage, and consecutive bits will be allocated to the same account until its leverage is equalized with the second-highest initial leverage. At that point, the next small bit could be allocated to either of the two accounts (or split). However, the choice does not matter when quantities are infinitesimal. Intuitively, even if one bit was allocated to account 1 (rather than split between 1 and 2 as in the water-filling rule), the next bit would then go to account 2, realigning the leverage levels and approximately resulting in the water-filling outcome. In the limit of infinitesimally small bits, the ambiguity disappears and we recover exactly the water-filling of Theorem~\ref{thm:risk_neutral}.

For the mathematical proof, the key is that path-independence amounts to the semigroup property  $F_{Q_2}\circ F_{Q_1}=F_{Q_1+Q_2}$ where the deleveraging quantity $Q$ plays the role of the time parameter. Moreover, the monotonicity of any ADL mechanism implies that this semigroup is Lipschitz-continuous (Lemma~\ref{lem:tc-implies-reg-isolated}), and once this regularity is recognized, the above intuition can be converted into a proof.

\begin{remark}
Taken individually, path-independence and leverage-priority do not imply the minimax leverage policy. For instance, consider the ADL mechanism which allocates the entire quantity $Q$ to the account(s) with the highest leverage at the ADL time (see also Remark~\ref{rem:queue-is-inconsistent}). This mechanism clearly satisfies leverage-priority, but is not path-independent: After being deleveraged with $Q_1$, the targeted account $i$ may no longer have the highest leverage, and then in a subsequent application of the mechanism, the quantity $Q_2$ would be allocated to a different account. Whereas when applied in one shot with $Q_1+Q_2$, the entire quantity is allocated to account~$i$. On the other hand, consider the pro-rata policy, which allocates amounts $x_i$ proportional to positions $q_i$, or equivalently, decreases the leverage of all accounts by the same percentage. This policy is path-independent, but clearly does not prioritize highest leverage.
\end{remark}


\subsection{Conditional Value-at-Risk and Monotone Risk Measures}\label{subsec:cvar}

We now replace the risk-neutral objective $\E\big[\Lscr(x,p_T)\big]$ by Conditional Value-at-Risk
$\mathrm{CVaR}_\beta$ at confidence level $\beta\in(0,1)$ as introduced by \citet{rockafellar2000cvar}, while keeping the same
buyback constraints. 

\begin{definition}For an integrable random variable $Z$ and confidence level $\beta\in(0,1)$, define the Value-at-Risk
\[
\mathrm{VaR}_\beta(Z)\;\defeq\;\inf\{z\in\R:\ \PR(Z\le z)\ge \beta\}.
\]
The Conditional Value-at-Risk at level $\beta$ is the tail mean\footnote{We use the convention $\mathrm{CVaR}_0(Z)=\E[Z]$.} beyond $\mathrm{VaR}_\beta(Z)$,
\[
  \mathrm{CVaR}_\beta(Z)\;\defeq\;\frac{1}{1-\beta}\int_\beta^1 \mathrm{VaR}_u(Z)\,du.
\]
If $Z$ has no atom at $\mathrm{VaR}_\beta(Z)$, this can also be written as
\[
\mathrm{CVaR}_\beta(Z) = \E\!\left[ Z\ \big|\ Z \ge \mathrm{VaR}_\beta(Z) \right].
\]
\end{definition}

Unlike the expectation, $\mathrm{CVaR}_\beta$ emphasizes tail losses and therefore changes
the marginal incentives driving auto-deleveraging. In the following, we formulate the CVaR-based ADL problem and characterize its separable structure and optimality conditions.  We will see that the
water-filling rule remains optimal, but (in contrast to the expected loss) the optimizer need not be
unique. 

The CVaR-based ADL problem reads
\begin{equation}\label{eq:adl_primal_cvar}
\begin{array}{ll}
\minimize_{x\in\R^n} & \displaystyle \mathrm{CVaR}_\beta\big(\Lscr(x,p_T)\big) \\
\subjectto & x\in\mathcal X.
\end{array}
\end{equation}
It is straightforward to check that \eqref{eq:adl_primal_cvar} is convex (see Appendix~\ref{app:cvar_degeneracy}). A key structural observation in our setting is that the individual account losses are monotone functions of a \emph{single} scalar risk factor, the terminal price $p_T$. Specifically,  $\sigma_i(x_i,p)$ is nondecreasing in $p$ for each fixed $x_i$ (see Lemma~\ref{lem:shortfall_prop} in Appendix~\ref{app:isolated.margining}).
Hence for any feasible $x$, the random variables $\{\sigma_i(x_i,p_T)\}_{i=1}^n$ are comonotone: they move in the same
direction as the price $p_T$.  Since $\mathrm{CVaR}_\beta$ is comonotone additive \citep[cf.~Propositions 2.3 \& 3.4,][]{tasche2002expected}, it follows that
\[
\mathrm{CVaR}_\beta\big(\Lscr(x,p_T)\big)=\mathrm{CVaR}_\beta\!\left(\sum_{i=1}^n \sigma_i(x_i,p_T)\right)
=
\sum_{i=1}^n \mathrm{CVaR}_\beta\big(\sigma_i(x_i,p_T)\big).
\]
This account-by-account decomposition is a crucial simplification and yields a clean characterization of optimality. Our main result is as follows.

\begin{theorem}\label{thm:waterfilling_nonunique}
Let $x^{\mathrm{WF}}$ denote the (unique) water-filling allocation of Theorem~\ref{thm:risk_neutral}. Then $x^{\mathrm{WF}}$ is an optimizer of the CVaR-ADL problem \eqref{eq:adl_primal_cvar} for every $\beta\in(0,1)$. However, the solution may be non-unique; in general, \eqref{eq:adl_primal_cvar} may admit a continuum of optimizers.
\end{theorem}
Let
$ p_\beta \defeq \mathrm{VaR}_\beta(p_T)$ denote the $\beta$-quantile (stress level) of the terminal price.
Since each shortfall $\sigma_i(x_i,p_T)$ is nondecreasing in the single risk factor $p_T$,
$\mathrm{CVaR}_\beta$ is determined by losses in the tail event $\{p_T\ge p_\beta\}$. This motivates the definition of the \emph{leverage cutoff} $\ell_\beta$ via
\[
p_\tau\bigl(1+\ell_\beta^{-1}\bigr)=p_\beta,
\qquad\text{equivalently}\qquad
\ell_\beta=\frac{p_\tau}{p_\beta-p_\tau}, \qquad\text{if}\qquad p_\beta > p_\tau,
\]
whereas $\ell_\beta=+\infty$ if $p_\beta\leq p_\tau$. Thus $\ell_\beta$ is precisely the leverage level whose bankruptcy price (cf.~\eqref{eq:ze_price_leverage}) coincides
with $p_\beta$.

The cutoff separates two regimes. If an account’s post-ADL leverage satisfies $\ell_i(x_i)<\ell_\beta$,
then its bankruptcy threshold lies \emph{within} the CVaR tail,
and additional buyback changes tail losses in an account-specific way. By contrast, if an account remains
highly levered in the sense that $\ell_i(x_i)\ge \ell_\beta$, 
then the
tail event relevant to $\mathrm{CVaR}_\beta$ is governed by $p_\beta$ rather than by the account-specific
threshold. In this stressed regime the marginal $\mathrm{CVaR}_\beta$ benefit of further buyback becomes
identical across all such accounts. Consequently, the water-filling rule remains optimal, but the objective
can develop flat directions: redistributing buyback volume among stressed accounts (while preserving
feasibility) leaves $\mathrm{CVaR}_\beta$ unchanged.

A continuum of optimal allocations can arise when the deleveraging budget $Q$ is insufficient to bring
\emph{every} account below the cutoff $\ell_\beta$, so that multiple accounts remain in the
stressed set $\{\ell_i(x_i)\ge \ell_\beta\}$. In that case, $\mathrm{CVaR}_\beta$ is indifferent to how buybacks are
redistributed among stressed accounts, producing non-uniqueness; see Appendix~\ref{app:cvar_degeneracy}
for a formal construction.

\bigskip


The optimality of water-filling can also be extended beyond the expected loss and $\mathrm{CVaR}$ cases, to all \emph{monotone} risk measures.

\begin{definition}[Monotone risk measure]
A risk measure $\rho$ is monotone if
\[
Z_1\le Z_2\ \text{a.s.}
\quad\Longrightarrow\quad
\rho(Z_1)\le \rho(Z_2).
\]
\end{definition}

The CVaR criterion corresponds to focusing on a \emph{single} tail level, while this class includes all \emph{spectral risk measures}, which assign weights to different
quantile levels of the loss distribution. Specifically, we say that $\rho(\cdot)$ is a spectral risk measure \citep[cf.][]{acerbi2002spectral,shapiro2013kusuoka} if it has the form\footnote{We interpret $\rho(\cdot)$ as extended real valued on integrable random variables and tacitly assume $\E[Z]\geq0$ so that $\rho(Z)\geq0$. A sufficient condition ensuring that $\rho(Z)<\infty$ for all integrable $Z$ is $\int_{[0,1)} (1-\beta)^{-1}d\mu(\beta)<\infty$.}
\begin{equation*}
\rho(Z)\;\defeq\;\int_{[0,1)} \mathrm{CVaR}_\beta(Z)\,\mu(d\beta)
\end{equation*}
for all integrable random variables $Z$, where $\mu$ is a probability measure on $[0,1)$.

Mirroring \eqref{eq:risk_based_adl}, we can write the more general optimization over monotone risk measures as
\begin{equation}\label{eq:adl_primal_monotone}
\begin{array}{ll}
\minimize_{x\in\R^n} & \displaystyle \rho\big(\Lscr(x,p_T)\big)\\
\subjectto & x\in\mathcal X.
\end{array}
\end{equation}
The following is a straightforward consequence of the pointwise optimality in Theorem~\ref{thm:risk_neutral}(d).

\begin{proposition}\label{prop:wf.optimal.monotone}
     Under Assumptions~\ref{ass:feasible}--\ref{ass:p_T_regular} the water-filling allocation is optimal for \eqref{eq:adl_primal_monotone} whenever $\rho(\cdot)$ is monotone.
\end{proposition}


Note that Proposition~\ref{prop:wf.optimal.monotone} states optimality, not uniqueness. For
general monotone risk measures the optimal
allocations need not be unique, as shown in Theorem~\ref{thm:waterfilling_nonunique} for $\mathrm{CVaR}_\beta$.

\subsection{Numerical Example}
\label{sec:numerical_example}

We illustrate the structure of CVaR-based auto-deleveraging in a one-asset setting with $n=4$ short accounts and total short volume $\sum_{i=1}^n q_i=33$. The reference price is fixed at $p_\tau\approx\$67{,}000$ (BTC spot), and deleveraging outcomes are studied as a function of the available ADL budget $Q$.

The terminal price $p_T$ is modeled as a geometric Brownian motion over a horizon $\Delta t=10/365$, with zero drift and annual volatility 60\%. For the CVaR criterion we take $\beta=0.98$. Closed-form expressions for the resulting CVaR objective under GBM are given in Appendix~\ref{app:closeformGBM}.

Each account $i$ is characterized by a short position $q_i$, entry price $p_i^{(e)}$, and margin $m_i$, with initial leverage
\[
\ell_i(0)=\frac{p_\tau q_i}{E_{i}}.
\]
Under the GBM model, the $\beta$-quantile $p_\beta$ of the terminal price admits a closed form, yielding the stress cutoff leverage \begin{equation*}
\ell_\beta=\frac{p_\tau}{p_\beta-p_\tau}\approx4.55.
\end{equation*} 
Accounts with $\ell_i(0)>\ell_\beta$ are initially stressed.

Table~\ref{tab:comparison_cvar4} reports the static characteristics of four levered accounts. All four satisfy $\ell_i(0)>\ell_\beta$ and therefore lie in the stressed region at the trigger price. From these static quantities alone, there is no clear ordering of accounts in terms of which should be deleveraged first under a tail-risk criterion.

\begin{table}[htbp]
\centering
\caption{Parameters of the four most levered accounts, sorted in increasing initial leverage.}
\label{tab:comparison_cvar4}
\begin{tabular}{ccccc}
\toprule
Account
& $q_i$
& $p_i^{(e)}$ in \$1000
& $m_i$ in \$1000
& $\ell_i(0)$ \\
\midrule
1 & 8  & 71 & 82 & 4.7 \\
2 & 10 & 72 & 78.8 & 5.2 \\
3 & 8  & 70 & 59.8 & 6.4 \\
4 & 7  & 69.5 & 48.6 & 7.1 \\
\bottomrule
\end{tabular}
\end{table}

\medskip

Figure~\ref{fig:cvar_four_accounts} reports the CVaR-optimal deleveraging outcome as a function of the budget $Q$, computed using the water-filling construction that equalizes post-ADL leverage across active accounts whenever feasible. Black curves show individual post-ADL leverage trajectories $\ell_i(x_i)$, while the solid blue curve reports the mean leverage among at-risk accounts. The horizontal line indicates the stress threshold $\ell_\beta$.

For small budgets $Q$, deleveraging is concentrated on the most levered accounts, which are progressively brought down toward $\ell_\beta$. In this fully stressed regime, the CVaR objective decreases linearly in $Q$, reflecting identical marginal risk reduction across all stressed accounts.

Let $Q_\beta$ denote the smallest budget for which all accounts can be brought below the stress threshold $\ell_\beta$. At $Q=Q_\beta$, the system exits the stressed regime. For $Q>Q_\beta$, marginal risk reduction diminishes and the objective decreases at a strictly smaller rate, as shown by the change in slope in the inset of Figure~\ref{fig:cvar_four_accounts}.

This transition illustrates the degeneracy and non-uniqueness of CVaR-optimal allocations in the
stressed regime. The water-filling rule selects a canonical representative from the set of
minimizers, yielding smooth leverage paths and a transparent dependence on the available
deleveraging budget.

\begin{figure}[htbp]
\centering
\includegraphics[width=0.65\textwidth]{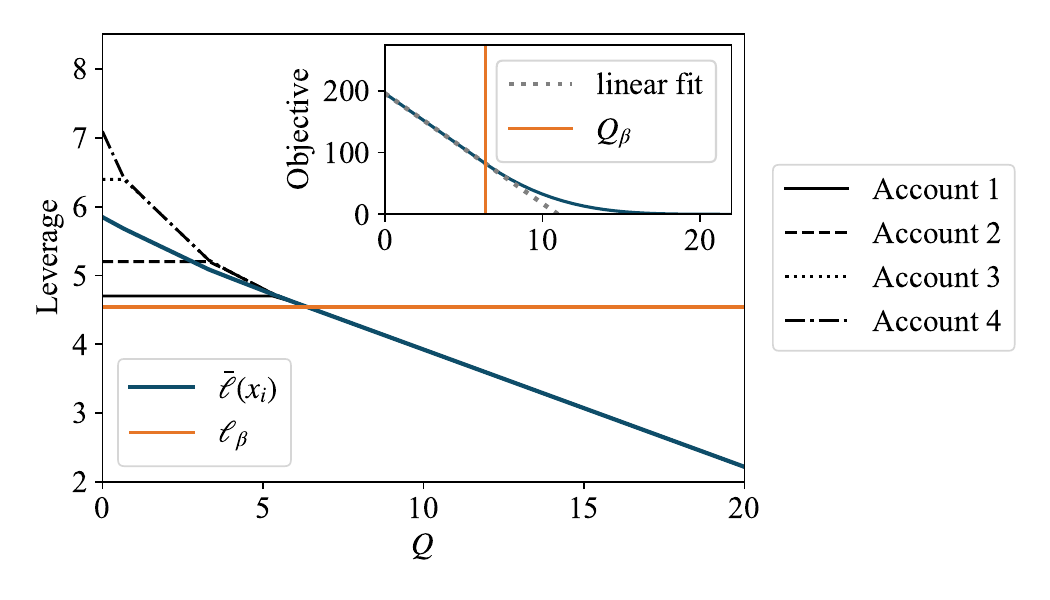}
\caption{
CVaR-optimal deleveraging under the water-filling rule.
Black line-style curves correspond to individual post-ADL leverage levels $\ell_i(x_i)$, the solid blue curve shows the mean leverage among at-risk accounts, and the orange horizontal line indicates the stress threshold $\ell_\beta$.
The inset displays the CVaR objective value as a function of the total deleveraging budget $Q$, where the gray dotted line is a linear fit and the orange vertical line illustrates the deleveraging volume $Q_\beta$ necessary to reach $\ell_\beta$.}
\label{fig:cvar_four_accounts}
\end{figure}


\section{Multi-Asset Cross-Margining}\label{subsec:rn_multiasset}

This section extends the preceding single-asset analysis to a multi-asset, cross-margin setting.
Under cross-margining, a single pool of collateral supports all open positions, so that gains in one position can offset losses in another and reduce the total margin required. This netting improves capital efficiency, but it also couples positions through shared collateral and correlated price moves. As a result, the ADL allocation decision is inherently multi-dimensional and must account for the joint distribution of returns across account positions. We formulate the resulting risk-based allocation problem faced by the exchange, focusing on expected loss as a risk measure. In \cref{sec:sep_decomp_for_exp_loss} we then show that the ADL problem becomes separable \emph{across accounts} once their coupling through the target reduction amounts~$Q$ is replaced by suitable shadow prices, opening the door to an efficient algorithm. Of course, this does not remove the possibly complicated coupling of the assets' returns. If assets are primarily driven by a common factor, this factor provides a natural dimension reduction. \Cref{sec:single_factor} formulates an idealized model where prices are driven by a single factor and  shows that ADL then admits close parallels to the isolated-margining case of \cref{sec:isolated.margining}. Indeed, a generalized water-filling rule is seen to be optimal for reducing the expected loss. However, water-filling occurs in a quantity that we call factor leverage, rather than the naive (gross) leverage which does not take into account hedging effects across portfolio components.

\subsection{Problem Formulation}

In what follows, we consider $d$ assets, with prices described by a vector in~$\R^d$.\footnote{Negative prices are allowed to accommodate derivative contracts.} Each account $i\in\{1,\dots,n\}$ holds a signed position vector $q_i=(q_i^1,\dots,q_i^d)\in\R^d$, where $q_i^k>0$ denotes a short position and $q_i^k<0$ a long position in asset $k$. The entry-price vector is $p_i^{(e)}\in\R^d$, and the posted margin is $m_i\ge 0$. Auto-deleveraging is triggered at time $\tau$ with reference prices $p_\tau\in\R^d$ and the exchange must reduce aggregate exposure by a predetermined amount $Q\in\R^d$.\footnote{We treat the trigger price $p_\tau$ and the aggregate reduction vector $Q$ as exogenously given inputs to the ADL mechanism. In practice, their determination is itself a separate design problem for the exchange, involving choices of marking rules, liquidity considerations, and stress metrics. The present analysis focuses on the allocation problem conditional on these quantities.}
That is, each component $Q^k$ of the vector $Q$ specifies the aggregate position reduction required in asset $k$, to be allocated across the $n$ accounts.

The ADL decision consists of vectors $x_1,\dots,x_n\in\R^d$, where $x_i^k$ is the (signed) position reduction applied to asset $k$ in account $i$, so that the post-ADL position vector is $q_i-x_i$.
We assume monotone deleveraging, in the sense that positions may be reduced toward zero but never increased or have their sign changed. Moreover, reductions are restricted to the direction of the target allocation $Q$: if $Q^k=0$ then $x_i^k=0$ for all $i$, and if $Q^k\neq 0$ then only accounts with
$\operatorname{sgn}(q_i^k)=\operatorname{sgn}(Q^k)$ may be reduced in asset $k$. Equivalently, define the directional bounds
\begin{equation}\label{eq:directional.bds}
l_i^k \defeq
\begin{cases}
0, & Q^k\geq 0,\\
\min(0,q_i^k), & Q^k<0,
\end{cases}
\qquad
u_i^k \defeq
\begin{cases}
\max(0,q_i^k), & Q^k>0,\\
0, & Q^k\leq 0,
\end{cases}
\end{equation}
so that, together with the aggregate reduction constraint $\sum_{i=1}^n x_i = Q$, the feasible set is
\[
\mathcal X
\defeq
\Big\{x=(x_1,\dots,x_n)\in(\R^{d})^n:\
\sum_{i=1}^n x_i = Q,\ \ l_i^k \le x_i^k \le u_i^k\ \ \forall\, i,k\Big\}.
\]

\begin{assumption}[Feasibility]\label{ass:cross_margin_feasible}
$Q$ satisfies
\[
Q^k \in \left[\sum_{i=1}^n \min(0,q_i^k),\ \sum_{i=1}^n \max(0,q_i^k)\right],
\qquad k=1,\dots,d,
\]
so that the feasible set $\mathcal X$ is nonempty.
\end{assumption}

For each account $i$, equity remains a scalar quantity but depends on the full price vector.
If account $i$ is reduced by $x_i$ and the price vector at time $T>\tau$ is $p_T\in\R^d$, define
\[
e_i(x_i,p_T)
\defeq
q_i^\top(p_i^{(e)}-p_T)
-
x_i^\top(p_\tau-p_T)
+
m_i .
\]
Writing $E_i\defeq e_i(0,p_\tau)=q_i^\top(p_i^{(e)}-p_\tau)+m_i$ to be the equity at the ADL time, we equivalently have
\begin{equation}\label{eq:equity.simplified}
e_i(x_i,p_T)=E_i+(q_i-x_i)^\top(p_\tau-p_T).
\end{equation}
The second term captures the remaining exposure to the price move from $p_\tau$ to $p_T$ after reducing positions by $x_i$. Throughout, we restrict attention to accounts that are solvent at time~$\tau$.

\begin{assumption}[Solvency]\label{ass:cross_margin_solvency}
  $E_i>0$ for all $i=1,\dots,n$.
\end{assumption}
As before, equity at time $\tau$ is unchanged by the reduction $x_i$, since $e_i(x_i,p_\tau)=E_i$ for all $x_i$.
By analogy with \eqref{eq:exchange_loss} and \eqref{eq:per.account.shortfall}, the corresponding shortfall is
\begin{equation}\label{eq:per.account.shortfall.cross}
\sigma_i(x_i,p_T)=(-e_i(x_i,p_T))_+,
\end{equation}
and the total exchange loss is
\[
\Lscr(x,p_T)\defeq\sum_{i=1}^n\sigma_i(x_i,p_T).
\]

In the single-asset case, post-ADL equity is monotone in the terminal price $p_T$. Upward moves uniformly harm shorts and
downward moves uniformly harm longs, which induces a natural ordering of accounts by ``distance to default'' and makes
leverage a meaningful proxy for default risk. In the multi-asset case, equity depends on the entire price vector and portfolios may be
partially hedged across assets. There is generally no canonical stress direction and no a priori risk ordering across accounts, since
insolvency depends on the joint distribution of $(p_T^1,\dots,p_T^d)$.

In multi-asset settings, exchanges often report leverage-type statistics that extend the one-dimensional notion
(motivated, for instance, by the discussion of `portfolio margin' in \citep{binanceADL}). One natural extension is the \emph{gross leverage} of account $i$ after reduction $x_i$,
defined by
\[
\ell_i(x_i)\defeq \frac{\sum_{k=1}^d |p_\tau^k\,(q_i^k-x_i^k)|}{E_i},
\]
where the numerator is the gross exposure of the residual portfolio valued at the reference prices~$p_\tau$.
Gross leverage is easy to compute but does not necessarily reflect default risk. Since it ignores the dependence structure of returns, it can be large even when exposures offset across assets, as in a delta-neutral portfolio. We use ``gross'' to emphasize that such netting is not reflected.

With this extended setup in mind, we model the exchange's ADL decision analogously to
Section~\ref{sec:isolated.margining}. Given a random close-out price $p_T$, the exchange selects reductions
$x\in\mathcal X$ to minimize a prescribed risk functional $\rho(\cdot)$ of the induced loss:
\begin{equation}\label{eq:adl_general_risk_program_intro}
\begin{array}{ll}
\minimize_{x\in(\R^d)^n} & \displaystyle \rho\big(\Lscr(x,p_T)\big) \\
\subjectto & x\in\mathcal X.
\end{array}
\end{equation}

General risk measures such as  $\rho(\cdot)=\mathrm{CVaR}_\beta(\cdot)$ couple accounts through tail scenarios of the aggregate loss. In the multi-asset setting, individual account shortfalls need not be comonotone, and
feasible ADL allocations are generally not totally ordered by a single exposure variable (such as leverage). Consequently,
different risk functionals can rank allocations differently and lead to different optimal ADL policies. In fact,
unlike the single-asset case in which a water-filling allocation is always optimal, the sets of optimizers associated
with distinct confidence levels $\beta$ can be \emph{disjoint}; see Appendix~\ref{app:distinct.opt.cvar} for an example.

For the subsequent analysis, we focus on the expected loss $\rho(\cdot)=\E[\cdot]$ and specialize~\eqref{eq:adl_general_risk_program_intro}  to
\begin{equation}\label{eq:adl_multiasset_exp_loss}
\begin{array}{ll}
\minimize_{x\in(\R^d)^n} & \displaystyle \E\big[\Lscr(x,p_T)\big] \\
\subjectto & x\in\mathcal X.
\end{array}
\end{equation}
The following ensures that~\eqref{eq:adl_multiasset_exp_loss} is well defined.
\begin{assumption}\label{ass:integrable_pT_cross}
The random vector $p_T$ is integrable.
\end{assumption}

While the expected loss criterion is more tractable thanks to additive separability across accounts, the coupling of assets remains a central complication that generally prevents reduction to single-variable problems. We emphasize that both sides of the ADL problem can be multi-dimensional in our formulation: the vector $Q$ means that several assets undergo ADL simultaneously, and each affected account may hold several of those assets, as well as other assets not subject to ADL but still contributing to potential losses.

\subsection{Separable Decomposition for the Expected Loss}\label{sec:sep_decomp_for_exp_loss}

As noted, the expected loss objective is additively separable across accounts:
\begin{equation}\label{eq:separable.across.accounts}
\E\big[\Lscr(x,p_T)\big]=\sum_{i=1}^n \E\big[\sigma_i(x_i,p_T)\big].
\end{equation}
Accordingly, the only coupling across accounts in \eqref{eq:adl_multiasset_exp_loss} arises through the vector
clearing constraint $\sum_{i=1}^n x_i=Q$. This suggests a decomposition based on convex duality, in which the clearing
constraint is priced by a multiplier $\lambda\in\R^d$ and each account solves an independent subproblem.

Define the per-account feasible set
\begin{equation}\label{eq:per.acct.feasible}
\mathcal Y_i \defeq \Big\{x_i\in\R^{d}: l_i^k \le x_i^k \le u_i^k\ \ \forall\, k\Big\},
\qquad i=1,\dots,n.
\end{equation}
Let $\mathcal Y \defeq \prod_{i=1}^n \mathcal Y_i$ be the joint set of constraints, and define the
partial Lagrangian
\[
\widehat{\mathcal L}(x,\lambda)
\defeq
\sum_{i=1}^n \E\!\left[\sigma_i(x_i,p_T)\right]
+
\lambda^\top\Bigl(\sum_{i=1}^n x_i-Q\Bigr)
=
-\lambda^\top Q+\sum_{i=1}^n\Bigl(\E\!\left[\sigma_i(x_i,p_T)\right]+\lambda^\top x_i\Bigr).
\]
For $\lambda\in\R^d$, define the primal per account value function by
\begin{equation}\label{eq:account-primal}
  \phi_i(\lambda)\defeq
  \min_{x_i\in\mathcal Y_i}\Big\{\E\!\left[\sigma_i(x_i,p_T)\right]+\lambda^\top x_i\Big\},
  \qquad i=1,\dots,n,
\end{equation}
and
the best-response correspondence
\begin{equation}\label{eq:best.resp.cross.margin}
X_i(\lambda)\defeq \argmin_{x_i\in\mathcal Y_i}\
\Big\{\E\!\left[\sigma_i(x_i,p_T)\right]+\lambda^\top x_i\Big\},
\qquad i=1,\dots,n.
\end{equation}
The associated dual function is
\[
g(\lambda)
\defeq
\min_{x\in\mathcal Y}\ \widehat{\mathcal L}(x,\lambda)
=
-\lambda^\top Q+\sum_{i=1}^n \phi_i(\lambda).
\]

The next proposition gives a convenient characterization of optimality in terms of best responses
and the dual function,
which will be exploited in Section~\ref{sec:numerical_implementation} below to develop an
efficient numerical approach.
\begin{proposition}[Dual decomposition]\label{prop:market_clearing_kkt}
Let Assumptions~\ref{ass:cross_margin_feasible}--\ref{ass:integrable_pT_cross}
hold. Then:
\begin{enumerate}
 \item[(i)]  The program \eqref{eq:adl_multiasset_exp_loss} is convex and admits
  an optimal allocation.

\item[(ii)] An allocation $x^\star=(x_1^\star,\dots,x_n^\star)\in(\R^d)^n$ is optimal if
and only if there exists
$\lambda^\star\in\R^d$
such that
\[
  x_i^\star \in X_i(\lambda^\star)\ \ \text{for all }i,
  \qquad\text{and}\qquad
  \sum_{i=1}^n x_i^\star = Q.
\]

\item[(iii)] Any such $\lambda^\star$ maximizes the dual function $g$, which, for any $\lambda$, has supergradient
  set
  \begin{equation}\label{eq:subdiff-g-main}
          \partial g(\lambda)=
    \Bigl\{
    \sum_{i=1}^n x_i-Q:\ x_i\in X_i(\lambda)\ \forall i
    \Bigr\}.
  \end{equation}

\end{enumerate}
\end{proposition}

Proposition~\ref{prop:market_clearing_kkt} admits a natural economic interpretation.
The multiplier $\lambda^\star\in\R^d$ acts as an asset-by-asset vector of \emph{shadow prices} for deleveraging capacity.
Given $\lambda$, each account $i$ solves the independent subproblem \eqref{eq:account-primal}
which trades off its expected shortfall contribution against the linear charge $\lambda^\top x_i$ for consuming reductions
in each asset. The clearing condition $\sum_{i=1}^n x_i^\star=Q$ selects $\lambda^\star$ so that the resulting
best responses match the required aggregate reduction. In this sense, the optimal ADL allocation can be viewed as a competitive equilibrium in which the exchange posts $\lambda^\star$
and accounts respond optimally subject to feasibility.

\subsection{Numerical Solution: ADMM with Sample-Average Approximation}\label{sec:numerical_implementation}

\Cref{prop:market_clearing_kkt} explains why the problem is structurally amenable to
decomposition: the expected-loss objective is separable across accounts, while the market-clearing
constraint introduces the only global coupling. At exchange scale, however, the exact expectation
layer is computationally intractable, so we replace it by a sample-average approximation (SAA) and
solve the resulting finite-scenario problem by the alternating direction method of multipliers
(ADMM), a standard operator-splitting method for structured convex optimization
\citep{boyd2011distributed}. The same SAA reformulation also provides a practical numerical route
for the low-dimensional examples considered below.

\paragraph{SAA reformulation.}
Draw $S$ independent price-increment scenarios
$\Delta p^{(1)},\dots,\Delta p^{(S)}$ from the calibrated price model, where $\Delta p^{(s)}\defeq p_\tau-p_T^{(s)}$ for scenario $s$. For account $i$, define
\begin{equation}\label{eq:saa-fi}
f_i(x_i)
\;\defeq\;
\frac{1}{S}\sum_{s=1}^S
\bigl(-E_i - (q_i-x_i)^\top\Delta p^{(s)}\bigr)_+
+ I_{\Yscr_i}(x_i),
\end{equation}
where $I_{\Yscr_i}$ is the indicator of the per-account box. Each $f_i$ is closed, proper, and
convex, being the sum of finitely many convex piecewise-linear terms restricted to a box. The SAA
counterpart of \eqref{eq:adl_multiasset_exp_loss} is therefore
\[
\min_{x\in(\R^d)^n}
\left\{
\sum_{i=1}^n f_i(x_i)
\ \colon\
\sum_{i=1}^n x_i = Q
\right\}.
\]

\paragraph{Consensus splitting.}
Introduce an auxiliary copy $z=(z_1,\dots,z_n)\in(\R^d)^n$ and impose $x_i=z_i$ for all $i$
together with the clearing condition $\sum_i z_i=Q$. The SAA problem becomes
\begin{equation}\label{eq:admm-split}
\begin{array}{ll}
\minimize_{x,z} &
  \displaystyle\sum_{i=1}^n f_i(x_i)
  + I_\mathcal{C}(z)\\[6pt]
\subjectto & x-z = 0,
\end{array}
\end{equation}
where $\mathcal{C}\defeq\{z\in(\R^d)^n:\sum_{i=1}^n z_i=Q\}$. Although the first block
decomposes across accounts, the optimization variables are grouped into the standard two ADMM blocks $x=(x_1,\dots,x_n)$ and $z=(z_1,\dots,z_n)$. The per-account box constraints are enforced in the $x$-update through the indicators $I_{\Yscr_i}(x_i)$, and the $z$-update is the Euclidean projection onto the market-clearing set $\mathcal C=\{z\in(\mathbb R^d)^n:\sum_i z_i=Q\}$.

\paragraph{ADMM updates.}
In scaled dual-variable form, the ADMM iteration for \eqref{eq:admm-split} is
\begin{align}
x_i^{k+1}
&\in
\argmin_{x_i\in\Yscr_i}
\left\{
f_i(x_i)
+ \frac{\rho}{2}\|x_i-c_i^k\|_2^2
\right\},
\quad
c_i^k \defeq z_i^k-u_i^k,
\quad i=1,\dots,n,
\label{eq:admm-x}\\[4pt]
z^{k+1}
&=
\Pi_{\mathcal C}(x^{k+1}+u^k),
\label{eq:admm-z}\\[4pt]
u^{k+1}
&=u^k+x^{k+1}-z^{k+1},
\label{eq:admm-u}
\end{align}
where $\rho>0$ is the augmented-Lagrangian penalty parameter and
$u=(u_1,\dots,u_n)$ is the scaled dual variable. Because
$f(x)=\sum_i f_i(x_i)$ and the quadratic penalty is blockwise additive, the
$x$-update~\eqref{eq:admm-x} separates exactly across accounts. The only global coupling enters
through the projection step~\eqref{eq:admm-z}, which enforces market clearing by projecting
$x^{k+1}+u^k$ onto the affine set $\mathcal C$. The full procedure is stated in
Algorithm~\ref{alg:admm} (Appendix~\ref{app:algorithm-admm}).

\paragraph{Convergence.}
For fixed penalty parameter $\rho$, the present formulation is an exact
two-block convex consensus ADMM, so the standard convergence theory of
\citet[Section~3.2]{boyd2011distributed} applies. In particular, the
primal and dual residuals converge to zero and the objective values
converge to the optimal value of the SAA problem. A formal statement and
proof are deferred to Appendix~\ref{app:algorithm-admm}.

\begin{remark}
The convergence result is asymptotic and does not by itself determine
the iteration budgets relevant in practice. For the exchange-scale
problem, the quantitatively important question is therefore how quickly
the iterates recover the economically relevant part of the loss
reduction. We return to that question in the exchange-scale implementation and in the empirical convergence diagnostics reported in \cref{app:rr-convergence}.
\end{remark}

Each ADMM iteration solves $n$ independent per-account proximal subproblems of the
form~\eqref{eq:admm-x}, followed by one projection onto $\mathcal{C}$ and a dual update, both of
which have cost $O(nd)$. The overall per-iteration cost depends on $n$, $d$, and the scenario
count $S$, while the accountwise proximal updates remain fully parallelizable. The exchange-scale
implementation details that make these updates tractable in the large-instance regime are
described in \Cref{sec:admm-impl}.

\subsection{Single Factor Model and Clipped Water-Filling}\label{sec:single_factor}

We now formulate an idealized model where asset returns are fully driven by a single factor. We will show that the solution of the ADL problem then admits a generalized water-filling structure, thus drawing a parallel with the isolated-margin  case of \cref{sec:isolated.margining}. Importantly, water-filling needs to be applied to a quantity that we call factor leverage, not to gross leverage. Factor leverage correctly reflects the ``distance to default'' in this setting, whereas gross leverage ignores hedging effects. Structurally, the key simplification in the single factor model is that each account's loss depends on its post-ADL portfolio only through a \emph{scalar} factor exposure.

\begin{assumption}[Single factor]\label{ass:single_factor}
There exist $v\in\R^d$ and a scalar random variable $\epsilon$ such that
\[
p_T = p_\tau + \epsilon v,
\]
where $\epsilon$ admits a strictly positive density on $\R$ and $\E[|\epsilon|]<\infty$.
\end{assumption}

For the remainder of \cref{sec:single_factor}, we assume that Assumptions~\ref{ass:cross_margin_feasible}--\ref{ass:single_factor} hold. Next, we introduce factor leverage, taking over the role that leverage played in \cref{sec:isolated.margining}.

\begin{definition}[Factor leverage]\label{def:factor_leverage}
For each account $i$ and allocation $x_i$, define the \emph{factor leverage} by
\[
\ell_i^{(v)}(x_i)\defeq \frac{v^\top (q_i-x_i)}{E_i}.
\]
\end{definition}

Thanks to Assumptions~\ref{ass:cross_margin_solvency} and~\ref{ass:single_factor}, the equity~\eqref{eq:equity.simplified} can be written as
\begin{equation}\label{eqn:equity_factor_lvg}
e_i(x_i,p_T)=E_i-\epsilon\,v^\top(q_i-x_i)=E_i\bigl(1-\epsilon\,\ell_i^{(v)}(x_i)\bigr),
\end{equation}
so the expected shortfall contribution of account $i$ depends on $x_i$ only through the factor leverage $\ell_i^{(v)}(x_i)$. For each account~$i$, recall from~\eqref{eq:per.acct.feasible} the set $\mathcal Y_i$ of feasible allocations and introduce the feasible factor-leverage interval
\[
[\underline\ell_i,\overline\ell_i]\defeq \big\{\ell_i^{(v)}(x_i):\ x_i\in\mathcal Y_i\big\},
\qquad i=1,\dots,n.
\]
This set is indeed a closed, bounded interval since $\ell_i^{(v)}(\cdot)$ is an affine transformation and $\mathcal Y_i$ is convex and compact.

We observe that the aggregate constraint $\sum_{i=1}^n x_i=Q$ implies a fixed equity-weighted factor exposure, denoted $L^{(v)}$ below. Indeed, for any feasible allocation $x\in\mathcal{X}$,
\begin{equation}\label{eq:sum_factor_fixed}
\sum_{i=1}^n  E_i \ell_i^{(v)}(x_i)=\sum_{i=1}^n v^\top(q_i-x_i)
=
v^\top\Big(\sum_{i=1}^n q_i - Q\Big) \defeq L^{(v)}.
\end{equation}
We can then think of different allocations as redistributing this fixed total across accounts.
Under certain implementability conditions, we will see that a clipped water-filling rule for the factor leverage
$\ell_i^{(v)}(x_i)$ is optimal: In principle, we would like to minimize the maximal factor leverage $\ell_i^{(v)}(x_i)$ in~\eqref{eq:sum_factor_fixed}. However, this water-filling is clipped to $[\underline\ell_i,\overline\ell_i]$ because an account's factor leverage can only be reduced as long as the account contains the asset(s) being deleveraged --- in contrast to the isolated-margining case, the factor leverage may remain high due to other assets in the portfolio that are outside the reach of ADL. See \cref{fig:factor_water_filling_clipping} for an illustration.

Intuitively, the general idea is still to first find a target factor leverage $t^\star$ that clears the market and then allocate quantities so as to drain leverages above the target; however, the clipping means that some accounts will need to have different targets, denoted by $\ell^\star_i$ below. The following theorem formalizes this structure.

\begin{figure}[tbh]
\centering

\begin{minipage}[t]{0.48\textwidth}
\centering
\begin{tikzpicture}
\begin{axis}[
    width=\textwidth,
    height=0.88\linewidth,
    ybar stacked,
    bar width=18pt,
    ymin=0, ymax=8.8,
    xmin=0.6, xmax=3.8,
    xtick={1.0,1.8,2.6,3.4},
    xticklabels={1,2,3,4},
    xlabel={Account $i$},
    ylabel={Factor leverage},
    axis lines=left,
    tick align=outside,
    ymajorgrids,
    grid style={dotted, gray!35},
    tick label style={font=\small},
    yticklabel style={opacity=0},
    label style={font=\small},
    clip=false,
    legend cell align={left},
    legend style={
        at={(0.5,1.05)},
        anchor=south,
        legend columns=3,
        draw=none,
        fill=none,
        font=\small,
    },
    legend image code/.code={%
        \draw[draw=black,#1] (0cm,-0.09cm) rectangle (0.34cm,0.09cm);
    },
]

\addlegendimage{fill=blue!25}
\addlegendentry{Removed by ADL}
\addlegendimage{fill=blue!80}
\addlegendentry{Post-ADL}
\addlegendimage{black, thick}


\addplot+[draw=black, fill=blue!80] coordinates
    {(1.0,6.0) (1.0,5.0) (1.8,6.0) (2.6,4.0) (3.4,2.0)}; 
\addplot+[draw=black, fill=blue!25] coordinates
    {(1.0,2.0) (1.8,0.0) (2.6,0.0) (3.4,0.0)};

\draw[black, dashed, thick] (axis cs:0.7,6.0) -- (axis cs:3.7,6.0);
\node[font=\small, anchor=south west] at (axis cs:3.52,6.0) {$t^\star$};

\draw[black, thick] (axis cs:0.88,5.0) -- (axis cs:1.12,5.0);
\node[font=\small, anchor=west] at (axis cs:1.14,5.0) {$\underline{\ell}_1$};

\node[font=\small, anchor=south] at (axis cs:1.0,8.08) {$\ell^{(v)}_1(0)$};
\node[font=\small, anchor=south] at (axis cs:1.8,6.08) {$\ell^{(v)}_2(0)$};
\node[font=\small, anchor=south] at (axis cs:2.6,4.08) {$\ell^{(v)}_3(0)$};
\node[font=\small, anchor=south] at (axis cs:3.4,2.08) {$\ell^{(v)}_4(0)$};

\end{axis}
\end{tikzpicture}

\end{minipage}
\hfill
\begin{minipage}[t]{0.48\textwidth}
\centering
\begin{tikzpicture}
\begin{axis}[
    width=\textwidth,
    height=0.88\linewidth,
    ybar stacked,
    bar width=18pt,
    ymin=0, ymax=8.8,
    xmin=0.6, xmax=3.8,
    xtick={1.0,1.8,2.6,3.4},
    xticklabels={1,2,3,4},
    xlabel={Account $i$},
    axis lines=left,
    tick align=outside,
    ymajorgrids,
    grid style={dotted, gray!35},
    tick label style={font=\small},
    yticklabel style={opacity=0},
    label style={font=\small},
    clip=false,
]

\addplot+[draw=black, fill=blue!80, forget plot] coordinates
    {(1.0,5.0) (1.8,3.5) (2.6,3.5) (3.4,2.0)};
\addplot+[draw=black, fill=blue!25, forget plot] coordinates
    {(1.0,3.0) (1.8,2.5) (2.6,0.5) (3.4,0.0)};

\draw[black, dashed, thick] (axis cs:0.7,3.5) -- (axis cs:3.7,3.5);
\node[font=\small, anchor=south west] at (axis cs:3.52,3.5) {$t^\star$};

\draw[black, thick] (axis cs:0.88,5.0) -- (axis cs:1.12,5.0);
\node[font=\small, anchor=west] at (axis cs:1.14,5.0) {$\underline{\ell}_1$};

\draw[->, black, thick] (axis cs:1.0,4.2) -- (axis cs:1.0,4.95);

\node[font=\small, anchor=south] at (axis cs:1.0,8.08) {$\ell^{(v)}_1(0)$};
\node[font=\small, anchor=south] at (axis cs:1.8,6.08) {$\ell^{(v)}_2(0)$};
\node[font=\small, anchor=south] at (axis cs:2.6,4.08) {$\ell^{(v)}_3(0)$};
\node[font=\small, anchor=south] at (axis cs:3.4,2.08) {$\ell^{(v)}_4(0)$};

\end{axis}
\end{tikzpicture}

\end{minipage}

\caption{Water-filling with clipping. The target threshold $t^\star$ equalizes post-ADL factor leverage for interior accounts,
but account-specific bounds can bind. In the right panel, account $1$ cannot be reduced below $\underline{\ell}_1$, so its
post-ADL factor leverage is clipped at $\underline{\ell}_1$ rather than being reduced to~$t^\star$.}
\label{fig:factor_water_filling_clipping}
\end{figure}

\begin{theorem}[Verification theorem for factor water-filling]\label{thm:factor_filling_verification}
For any $\eta\in\R$, define the \emph{target factor leverage} for account $i$ by
\begin{equation}\label{eq:factor_fill_target_main}
  \ell_i^\star(\eta)\in\argmin_{z\in[\underline\ell_i,\overline\ell_i]}
   \big\{\psi(z)-\eta z\big\},
\quad i=1,\dots,n, \qquad \text{where}\quad \psi(z)\defeq \E\big[(\epsilon z-1)_+\big],\quad z\in\R.
\end{equation}
\begin{enumerate}
\item[(i)] The minimizer $\ell_i^\star(\eta)$ is unique for any $\eta\in\R$, and $\eta\mapsto \ell_i^\star(\eta)$ is continuous and nondecreasing.
\item[(ii)] For any $\eta\in\R$, $\ell_i^\star(\eta)$ satisfies the clipped water-filling rule\footnote{The expressions in~\eqref{eq:water_fill_piecewise_main} are well-defined since the convex function $\psi:\mathbb{R}\to\mathbb{R}_+$ is continuously differentiable with a strictly increasing derivative $\psi'$; cf.\ Appendix~\ref{app:prop.psi}.}
\begin{equation}\label{eq:water_fill_piecewise_main}
\ell_i^\star(\eta)=
\begin{cases}
\underline \ell_i, & \eta\le \psi'(\underline \ell_i),\\
(\psi')^{-1}(\eta), & \psi'(\underline \ell_i)<\eta<\psi'(\overline \ell_i),\\
\overline \ell_i, & \eta\ge \psi'(\overline \ell_i).
\end{cases}
\end{equation}
\item[(iii)] There exists at least one $\eta^\star\in\R$ such that
\begin{equation}\label{eq:eta_budget_main}
\sum_{i=1}^n E_i\,\ell_i^\star(\eta^\star)=v^\top\Big(\sum_{i=1}^n q_i - Q\Big).
\end{equation}
\end{enumerate}
Given $\eta^\star$ satisfying \eqref{eq:eta_budget_main}, if there exists an allocation $x^\star\in\mathcal X$ such that
\[
\ell_i^{(v)}(x_i^\star)=\ell_i^\star(\eta^\star)\quad\text{for all }i,
\]
then $x^\star$ is an optimal allocation for \eqref{eq:adl_multiasset_exp_loss}.
\end{theorem}

By monotonicity of the map $\eta\mapsto \ell_i^\star(\eta)$, a smaller right-hand side in the budget equation
\eqref{eq:eta_budget_main} corresponds to a smaller value of $\eta^\star$ and therefore to smaller targets
$\ell_i^\star(\eta^\star)$ for each account $i$. This yields the (clipped) water-filling interpretation: as the aggregate
factor-exposure target is tightened, accounts are reduced along the common ``water level'' until they hit their individual
constraints.

Mathematically, Theorem~\ref{thm:factor_filling_verification} is a \emph{verification} result. It states that if the
clipped one-dimensional minimizers \eqref{eq:water_fill_piecewise_main} can be realized by some feasible allocation
$x\in\mathcal X$, then that allocation is globally optimal. In particular, the targets $\ell_i^\star(\eta^\star)$ are always
compatible with the per-account feasibility constraints $\mathcal Y_i$, by construction. The only remaining requirement is the
clearing condition $\sum_{i=1}^n x_i=Q$. In general multi-asset ADL, such a realization need not exist because $\sum_{i=1}^n x_i=Q$ is a vector constraint, whereas
\eqref{eq:eta_budget_main} specifies only a single scalar target for the accounts. A degenerate example in which no feasible allocation achieves
the water-filling targets is given in Appendix~\ref{app:water.filling.fail}.

However, it turns out that in many relevant situations, factor water-filling is implementable (and therefore optimal). One particular case is when only a single asset is subject to ADL, that is, only one component of $Q$ is nonzero. Note that this situation is still quite different from the isolated-margin case in \cref{sec:isolated.margining}, because the accounts being deleveraged now contain further assets that influence potential future losses.

\begin{theorem}[Factor water-filling for single-asset ADL]\label{thm:factor_filling_one_asset}
Suppose that only asset $k_0$ is subject to ADL, i.e.\ $Q^k=0$ for all $k\neq k_0$, and $v^{k_0}\neq 0$. Then clipped water-filling on factor leverage is implementable and optimal for~\eqref{eq:adl_multiasset_exp_loss}.

Specifically, let $\eta^\star$ be any solution of the
budget equation \eqref{eq:eta_budget_main}, and define the water-filling targets $\ell_i^\star(\eta^\star)$ by
\eqref{eq:factor_fill_target_main}. Then there exists an allocation $x^\star\in\mathcal X$ such that
\[
\ell_i^{(v)}(x_i^\star)=\ell_i^\star(\eta^\star)\quad\text{for all }i,
\]
and $x^\star$ is an optimal allocation for~\eqref{eq:adl_multiasset_exp_loss}.
\end{theorem}

Intuitively, \cref{thm:factor_filling_one_asset} holds because, when only a single asset is deleveraged, each account's factor
leverage $\ell_i^{(v)}(x_i)$ is an affine function of the scalar decision $x_i^{k_0}$. As a result, the one-dimensional
water-filling targets are automatically implementable while satisfying the (scalar) clearing constraint.

The single-asset case described by \cref{thm:factor_filling_one_asset} is a special case of a much more general sufficient condition for multi-asset ADL detailed in Appendix~\ref{app:suff.conditions.water.filling}.  It is based on a technical \emph{connectivity} property across the actively delevered assets (i.e., the nonzero components of $Q$), capturing the idea that some accounts have partial reductions in multiple assets. Mathematically, it is a weak form of an interior point condition, avoiding degenerate boundary cases like the example of Appendix~\ref{app:water.filling.fail}. Intuitively, such connectivity is most plausible in large exchanges when many large accounts maintain genuinely cross-asset portfolios, creating connections across the assets being delevered.

\begin{theorem}[Factor water-filling for multi-asset ADL]\label{thm:factor_filling_multi_asset}
Suppose that $v^k\neq 0$ unless $Q^k=0$, and that $Q$ and the accounts being delevered satisfy the connected partial deleveraging condition\footnote{See Definition~\ref{def:con.part.del} and Proposition~\ref{prop:gen.water.filling.connected}. Loosely speaking, connected partial deleveraging means that each asset subject to ADL is \emph{partially} reduced for at least one account and that the assets are linked through overlapping reductions. In other words, you can move from any delevered asset to any other by ``hopping'' through accounts that are partially reduced in more than one of these assets.} detailed in Appendix~\ref{app:suff.conditions.water.filling}. 
Then clipped water-filling on factor leverage is again implementable and optimal for~\eqref{eq:adl_multiasset_exp_loss}, as in \cref{thm:factor_filling_one_asset}.
\end{theorem}

\Cref{thm:factor_filling_multi_asset} shows that the single-asset water-filling intuition survives in the multi-asset cross-margin setting if assets are driven by a single factor. Under that structure, the residual factor leverage determines default risk, rather than gross leverage. Relative to the single-asset case, the only additional obstacle is whether the scalar water-filling targets can actually be implemented under the vector clearing constraint, and \cref{thm:factor_filling_multi_asset} provides a sufficient condition for this.

\begin{remark}
A similar message carries over to the $\mathrm{CVaR}_\beta$ formulation. Although $\mathrm{CVaR}_\beta$ generally couples accounts through tail scenarios and therefore does not admit the same separability as expected loss, the single-factor structure restores a one-dimensional ordering by factor leverage and the optimal ADL rule is again similar to the single-asset case --- once $\beta$ is fixed and assuming we further restrict to accounts with positive factor exposure, the allocation problem takes a clipped water-filling form in factor-leverage space. Thus, the extension from expected loss to $\mathrm{CVaR}$ preserves the basic economic insight: the exchange optimally allocates deleveraging by equalizing a marginal tail-risk criterion across accounts until individual feasibility constraints become binding. A more detailed discussion can be found in \cref{sec:Cvar_one_factor}.
\end{remark}

\subsection{Numerical Example}\label{sec:numerical_example_crossmargin}
In this section, we illustrate ADL under a bivariate price model for BTC and ETH, and discuss the accuracy of the single-factor approximation discussed in \cref{sec:single_factor}. Throughout, we fix the account state at the ADL time~$\tau$ and vary only the required aggregate BTC reduction. Specifically, we set
\[
Q=(Q^{\mathrm{BTC}},Q^{\mathrm{ETH}})=(Q^{\mathrm{BTC}},0),
\]
so that ETH positions are not deleveraged. Recall, however, that the ETH positions affect future losses and therefore the optimal BTC allocation. Terminal close-out prices are generated from a correlated bivariate geometric Brownian motion over the horizon $\Delta=T-\tau=10/365$. Concretely,
\[
p_T^k
=
p_\tau^k
\exp\!\left(
-\frac12 (\sigma^k_{\mathrm{ann}})^2\Delta
+
\sigma^k_{\mathrm{ann}}\sqrt{\Delta}\,Z^k
\right),
\qquad
k\in\{\mathrm{BTC},\mathrm{ETH}\},
\]
where $(Z^{\mathrm{BTC}},Z^{\mathrm{ETH}})$ is standard bivariate normal with correlation $\rho_{\text{corr}}$. Using approximate spot prices of February~5, 2026 for $p_\tau$ and a high correlation consistent with stressed market conditions, we set
\[
p_\tau^{\mathrm{BTC}}=\$67{,}000,
\qquad
p_\tau^{\mathrm{ETH}}=\$1{,}900,
\qquad
\sigma^{\mathrm{BTC}}_{\mathrm{ann}}=60\%,
\qquad
\sigma^{\mathrm{ETH}}_{\mathrm{ann}}=75\%,
\qquad
\rho_{\text{corr}}=0.85.
\]
For each value of $Q^{\mathrm{BTC}}$, we solve the SAA approximation of the expected-loss problem by the ADMM scheme of \cref{sec:numerical_implementation}. The expectation is approximated with $S=2048$ GBM scenarios, and the ADMM iteration is run with penalty $\rho=30$ and an iteration budget of $500$. 

To compare the full model with the theory of \cref{sec:single_factor}, we also solve the same ADL problem under the following one-factor approximation. Let $\Sigma_{\Delta p}$ denote the covariance matrix of the price increment $p_T-p_\tau$ implied by the GBM calibration, and let $(\lambda_1,u_1)$ be its leading eigenpair. We then set
\[
v \defeq \sqrt{\lambda_1}\,u_1,
\qquad
p_T = p_\tau + \epsilon v,
\qquad
\epsilon\sim N(0,1).
\]
The calculation of $v$ is detailed in \cref{sec:appendix_numerical_example_crossmargin}, where we also illustrate the approximation with a scatter plot. The one-factor model retains the dominant covariance mode and suppresses the orthogonal direction. In particular, ``factor exposure'' will refer to the factor leverage $\ell_i^{(v)}$ computed with this leading covariance direction~$v$.

We consider $n=4$ cross-margin accounts. All four are short BTC, but their ETH positions differ substantially, generating different exposures to the dominant covariance mode. Table~\ref{tab:multiasset_state} reports the account state at time~$\tau$. Positions are signed: positive entries denote shorts and negative entries denote longs. Account~4 is approximately factor-neutral despite having the largest gross leverage, account~2 is partially hedged, and accounts~1 and~3 are the most exposed to the dominant market factor.

\begin{table}[htbp]
\centering
\caption{Account state at the ADL time~$\tau$, including gross leverage $\ell_i(0)$ and factor leverage $\ell_i^{(v)}(0)$.}
\label{tab:multiasset_state}
\begin{tabular}{ccccccc}
\toprule
Account
& $q_i^{\mathrm{BTC}}$
& $q_i^{\mathrm{ETH}}$
& $m_i$ in \$1000
& $E_i$ in \$1000
& $\ell_i(0)$
& $\ell_i^{(v)}(0)$ \\
\midrule
1 & 8.0  & 323.0   & 137.5 & 242.1 & 4.8 & 0.49 \\
2 & 10.0 & $-38.7$ & 85.3  & 143.0 & 5.2 & 0.41 \\
3 & 8.0  & 326.2   & 75.4  & 180.6 & 6.4 & 0.66 \\
4 & 7.0  & $-190.0$& 43.9  & 116.9 & 7.1 & 0.07 \\
\bottomrule
\end{tabular}
\end{table}

\begin{figure}[htbp]
\centering
\includegraphics[width=0.86\textwidth]{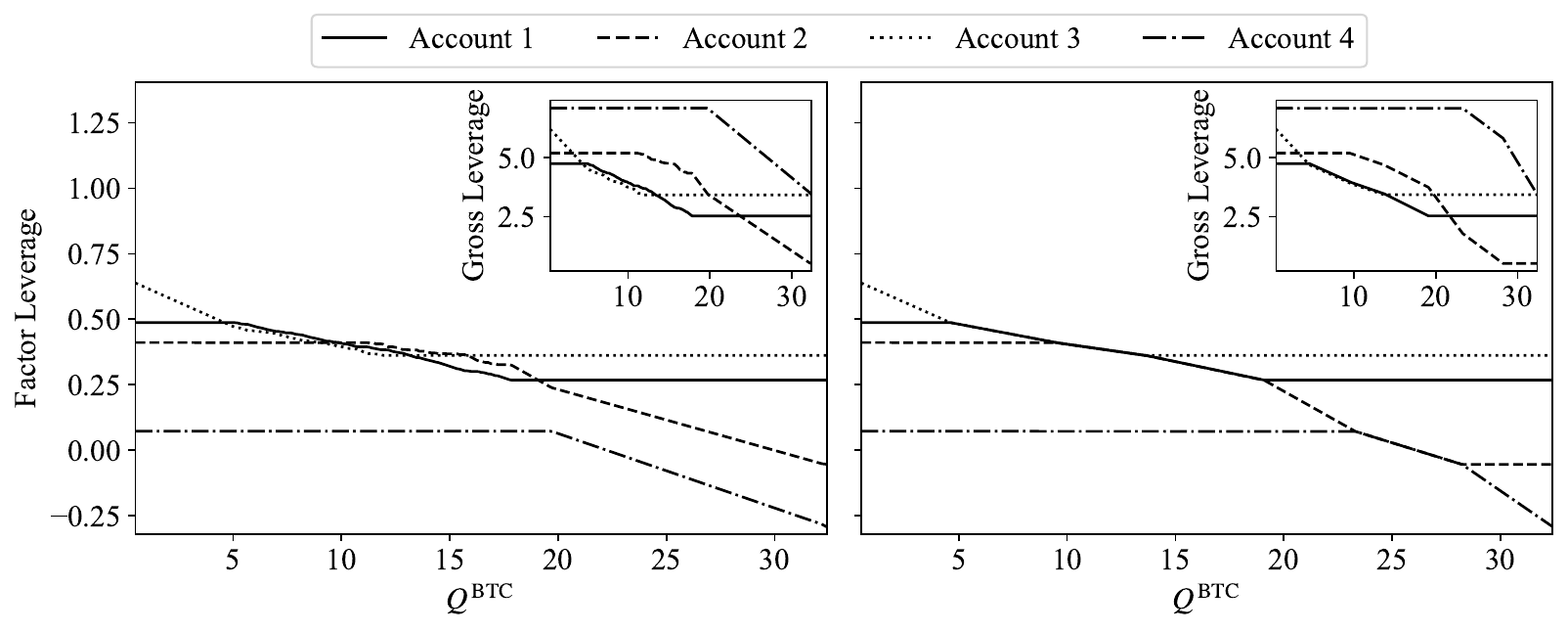}
\caption{Risk-optimal BTC deleveraging paths under the bivariate GBM model, computed from the SAA--ADMM approximation with $2048$ scenarios, $\rho=30$, and $500$ iterations (left panel), and under the one-factor approximation (right panel). The main panels plot post-ADL factor leverage $\ell_i^{(v)}(x_i^\star)$ against the BTC deleveraging budget $Q^{\mathrm{BTC}}$; the insets
report gross leverage $\ell_i(x_i^\star)$ for the same allocations.}
\label{fig:adl_cross_comparison}
\end{figure}

The left panel in \cref{fig:adl_cross_comparison} shows the ADL results obtained from the bivariate GBM model through the SAA--ADMM approximation,  whereas the right panel shows the results
under the one-factor approximation. Consistent with the theoretical results, the right panel
illustrates water-filling in factor exposure: the accounts with the largest factor exposure are
progressively leveled, except for the clipping that occurs when feasibility constraints become
binding. The left panel shows a qualitatively similar behavior, illustrating that the overall
shape of the allocation paths is governed primarily by the dominant covariance mode and that the
water-filling rule captures the principal characteristics.
Quantitatively, the two models are close, but not identical. In the one-factor approximation, the dominant covariance mode~$v$ completely governs the shape of the optimal allocation paths. The full GBM model retains two features that are absent from the rank-one approximation: lognormal nonlinearity and residual risk in the covariance direction orthogonal to~$v$. These effects are economically most visible at the points where the active set changes. Once the dominant factor exposures of the currently active accounts have been brought close together, the neglected second direction becomes relatively more important, and the SAA--ADMM allocation in the full bivariate model can deviate modestly from the exact water-filling pattern.

Both panels highlight the distinction between factor leverage and gross leverage. The insets show that gross leverage would give a misleading ranking of risk. In particular, account~4 has the largest gross leverage throughout much of the range, yet its factor exposure is initially close to zero and it receives little or no deleveraging until BTC reductions become large. Conversely, accounts~1 and~3 are reduced first even though their gross leverage is lower. What drives the optimal ADL decision is not gross notional exposure, but exposure to the dominant risk factor after cross-asset hedges are taken into account.


\section{The Exchange-Scale Problem}
\label{sec:exchange-scale}

A perpetual futures exchange operating a cross-margin book faces a qualitatively different
computational problem from the stylized instances studied in \Cref{sec:numerical_implementation}.
When an ADL event is triggered, the exchange must solve the
expected-loss program \eqref{eq:adl_multiasset_exp_loss} across a universe of tens of thousands of
accounts, each holding positions in dozens of actively traded assets, with an aggregate reduction
target that can reach hundreds of millions of dollars in notional within a matter of seconds. The
October~10, 2025 Hyperliquid event is a concrete instance of this problem. It is also a natural
case study: it is one of the largest publicly observable ADL episodes of 2025, it unfolds in
multiple distinct waves across a broad asset universe, and it generated substantial public
discussion about how losses were distributed. For our purposes, the key feature is that the event
yields an observed realized ADL allocation $x^{\mathrm{Hyperliquid}}$ that can be compared with
model-based allocations on the reconstructed event.

This section addresses two implementation challenges and one empirical question. On the
implementation side, exchange scale rules out direct expectation-based solution methods and
requires an implementation of the SAA--ADMM scheme of \Cref{sec:numerical_implementation} that
remains tractable when both the number of accounts and the number of active assets are large. In
addition, the active coins are highly heterogeneous in price level and support structure, so the
implementation must normalize variable scales and exploit per-account sparsity in order to remain
numerically stable and parallelizable. On the empirical side, the central question is comparative:
given a reconstructed ADL event, how does the realized Hyperliquid allocation compare with the
sequential one-factor Water-filling benchmark induced by \Cref{sec:single_factor} and with the
Risk-optimal allocation computed from the same SAA--ADMM scheme?

The comparison is informative for two reasons. First, the theory above characterizes the
Risk-optimal allocation and identifies a one-factor Water-filling benchmark, but it does not say
how far current exchange practice is from either object once portfolios are cross-margined and
assets are heterogeneous. Second, the comparison reveals not only the objective gap but also its
mechanism. The empirical comparison will show that the relevant differences are not exhausted by a
single expected-loss ranking. They also concern how deleveraging notional is concentrated across
users, how strongly reductions align with factor leverage, and to what extent the realized event
resembles partial portfolio reduction as opposed to near-full directional closeout.

\Cref{sec:empirical} introduces the October~10, 2025 Hyperliquid event, the data sources, the
reconstruction of pre-event account states, and the resulting wave-level optimization inputs.
\Cref{sec:admm-impl} then explains how the abstract ADMM scheme is instantiated at exchange scale.
Finally, \Cref{sec:oct10-results} focuses on the first wave $\tau_1$ and compares Hyperliquid,
Water-filling, and the Risk-optimal allocation in terms of sampled expected loss and user-level
targeting. The corresponding $\tau_3$ exercise and the robustness comparison with the directional
Water-filling benchmark are deferred to the appendix.
 
\subsection{Empirical Setup}\label{sec:empirical}

\paragraph{The ADL Event.}

On October~10, 2025, the perpetual futures exchange Hyperliquid
experienced a concentrated auto-deleveraging episode in which a large
number of accounts were partially or fully closed within a five-minute
window. Our empirical analysis is based on the 75-coin core universe
constructed in Appendix~\ref{app:state-reconstruction}. On this universe, the total realized ADL
notional on October~10 is $\$2.047\times 10^9$. Within the event
window $21{:}16$--$21{:}21$ UTC, the ADL flow unfolds in multiple
distinct bursts, identifiable by sharp increases in cumulative
notional volume. We mark three reference times
\[
\tau_1 = 21{:}16{:}05\ \mathrm{UTC},\qquad
\tau_2 = 21{:}16{:}56\ \mathrm{UTC},\qquad
\tau_3 = 21{:}17{:}06\ \mathrm{UTC},
\]
which correspond to the onset of the three dominant ADL waves visible
in Figure~\ref{fig:adl-notional-time}.

\begin{figure}[htbp]
\centering
\includegraphics[width=0.78\textwidth]{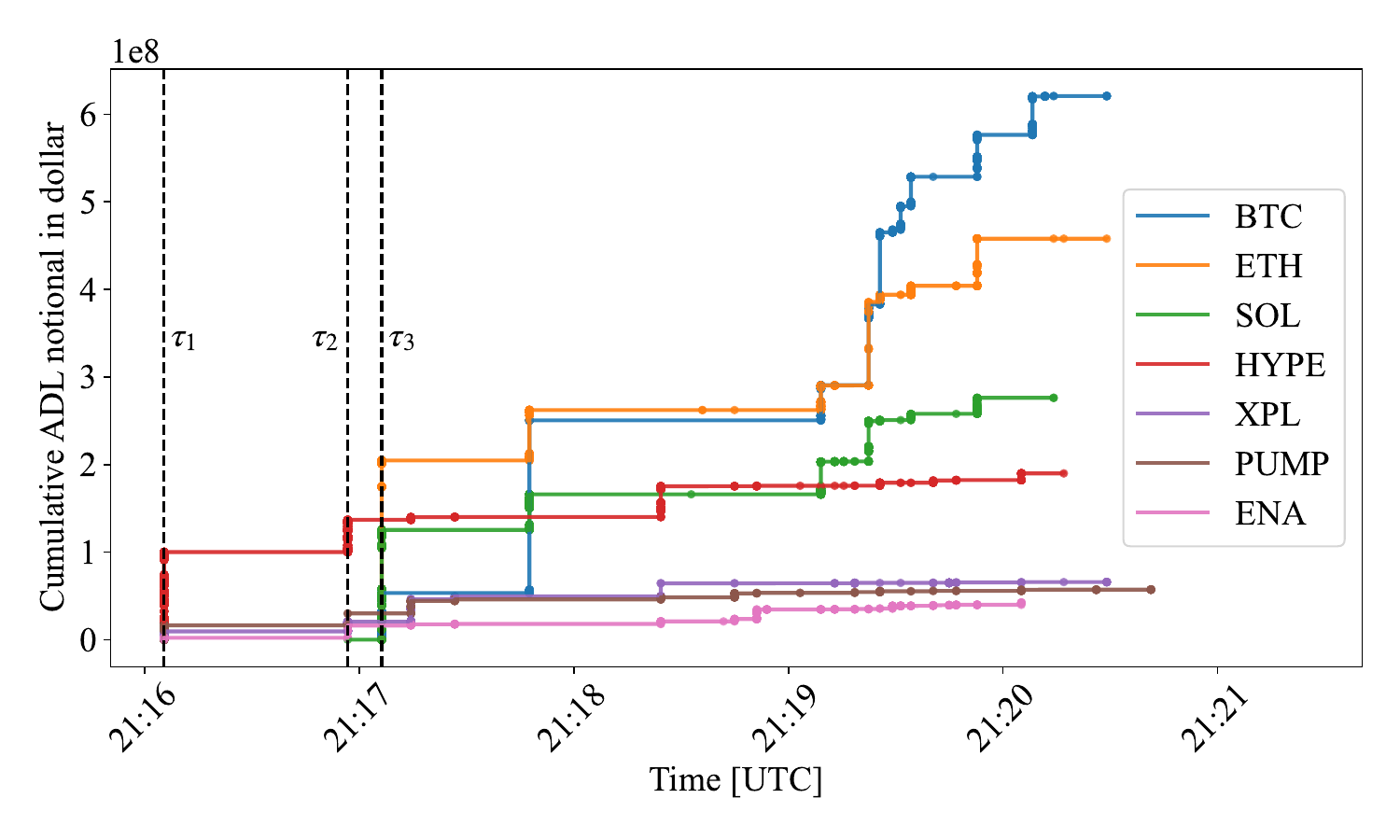}
\caption{Cumulative notional ADL volume by coin during the
  October~10, 2025 event, restricted to the five-minute window
  $21{:}16$--$21{:}21$ UTC\@. Vertical dashed lines mark $\tau_1$,
  $\tau_2$, and $\tau_3$. The seven largest coins by total ADL
  notional are shown; each step corresponds to an individual ADL fill.}
\label{fig:adl-notional-time}
\end{figure}

\begin{figure}[htbp]
\centering
\includegraphics[width=0.76\textwidth]{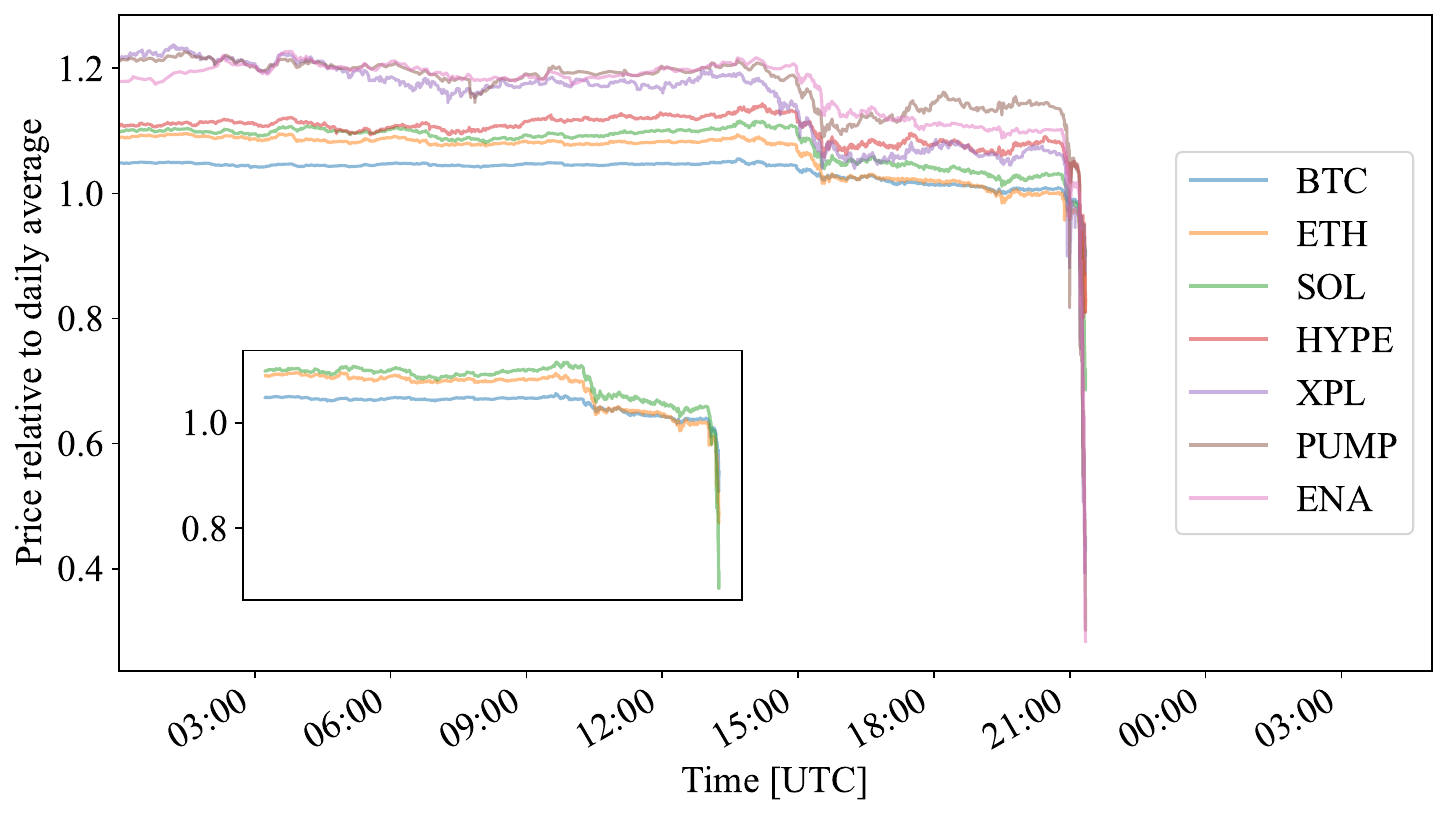}
\caption{Intraday price paths on October~10, 2025, normalized by each
  coin's average price up to the ADL event. The main panel shows the
  most actively traded affected coins, while the inset isolates the
  most liquid names. The cross-section is highly heterogeneous: several
  thinner coins lose more than $50\%$, whereas BTC, ETH, and SOL still
  fall by roughly $10$--$20\%$.}
\label{fig:pre-adl-price-stress}
\end{figure}

The three waves differ substantially in scale and composition. The
first wave $\tau_1$ has total realized ADL notional
$\$1.769\times 10^8$; the dominant coins are HYPE, PUMP, and
FARTCOIN. The second wave $\tau_2$ is smaller at $\$9.33\times 10^7$.
The third wave $\tau_3$ has a total
realized ADL notional $\$3.866\times 10^8$, concentrated mainly in
ETH, SOL, BTC, and XRP.

Figure~\ref{fig:pre-adl-price-stress} provides complementary market
context by plotting intraday price paths on October~10, 2025, normalized
by each coin's average price over the day up to the ADL event. The
figure shows that the exchange was not facing an isolated
liquidation in a single asset, but a broad and rapid repricing across
the active universe. Even the deepest markets, such as BTC, ETH, and
SOL, decline by roughly $10$--$20\%$, while several thinner names lose more
than half their value relative to their intraday average. In such
conditions, the allocation problem must be solved under tight time
constraints, and the exchange must distribute a large aggregate
reduction across accounts exposed to assets with heterogeneous
liquidity and common return directions.

\paragraph{Data sources.}
The empirical input combines two data sources. First, tick-level trade
records are obtained from the public Hyperliquid REST API\@. These
records contain the trade time, user address, coin, signed size, side,
direction label, execution price, and liquidation mark price where
applicable, allowing us to identify the realized ADL flow separately
from market-driven liquidations and voluntary trades. Second,
account-level state variables --- position sizes, entry prices, and
margin-used quantities per user and coin --- are obtained from Allium,
which indexes on-chain Hyperliquid account data. The Hyperliquid tape
therefore tells us what was actually delevered in the event, while the
Allium snapshots determine which pre-event account states can be
reconstructed reliably for the optimization input.

\paragraph{Data reconstruction and limitations.}
Because Allium only covers accounts with observable activity since
September~22, 2025, not every user appearing in the
ADL tape can be assigned a strict pre-event state. For each
reconstruction time $\tau_j$, we therefore recover the optimization
input by taking the most recent available account snapshot before
$\tau_j$ and rolling it forward with the Hyperliquid
trade tape. The optimization is carried out on the retained
reconstructed universe rather than on the full raw event stream. At
$\tau_1$, this is why the final ADL target volume used in the
optimization, $\$1.539\times 10^8$, is smaller than the raw event
volume $\$1.769\times 10^8$. Appendix~\ref{app:state-reconstruction}
contains the full reconstruction pipeline, the stage-by-stage counts,
and the precise treatment of unmatched users, missing snapshots, and
coin-level restrictions.

\paragraph{Price model inputs.}\label{sec:pm_inputs}
The multivariate GBM serves two purposes in the empirical exercise.
First, it generates the finite set of price-drop scenarios
$\{\Delta p^{(s)}\}_{s=1}^S$, with $\Delta p^{(s)} \defeq p_\tau - p_T^{(s)}$, used to evaluate the sampled exchange loss. Second, the same calibration supplies the
one-factor direction described previously in \cref{sec:numerical_example_crossmargin}.
We estimate the model from hourly log-returns
via EWMA-standardized covariance estimation; see
Appendix~\ref{sec:covariance} for the full calibration. The key object
is the conditional covariance of dollar price drops at the ADL
trigger time $\tau$,
\[
\Sigma_{\tau,\Delta}^{\Delta p}
\approx
\operatorname{diag}(p_\tau)\,
\bigl(\Delta\,\operatorname{diag}(\sigma_\tau)\,R_\tau\,
\operatorname{diag}(\sigma_\tau)\bigr)\,
\operatorname{diag}(p_\tau),
\]
where $\Delta$ is the liquidation horizon in years,
$\sigma_\tau\in\mathbb R_+^d$ is the vector of annualized conditional
volatilities at time $\tau$, and $R_\tau$ is the corresponding
correlation matrix. Appendix~\ref{sec:covariance} explains how these
objects are obtained. Let $(\lambda_1,u_1)$ denote the leading eigenpair of $\Sigma_{\tau,\Delta}^{\Delta p}$, and define $
v \defeq \sqrt{\lambda_1}\,u_1.$ This vector is the one-factor direction used by Water-filling in the
exchange-scale implementation below.

\subsection{Exchange-Scale Implementation}
\label{sec:admm-impl}

The Water-filling benchmark is implemented as sequential one-factor
Water-filling on the active asset set. The factor direction
$v \defeq \sqrt{\lambda_1}\,u_1$ is fixed from the covariance matrix in
\Cref{sec:pm_inputs}, while active coins are ordered by realized ADL
notional. At each step only the current coin is delevered, and account
priorities are recomputed from the current one-factor leverage before
moving to the next coin. Thus the benchmark preserves the one-factor
prioritization logic of \Cref{sec:single_factor} while enforcing the
multi-asset clearing constraint sequentially. To show that this
one-factor sequential Water-filling is superior to the isolated-margin
procedure that ranks accounts by leverage in the current coin only, we
report robustness checks against the directional benchmark in
Appendix~\ref{app:rr-robustness}.

The Risk-optimal allocation instantiates the ADMM scheme of
\Cref{sec:numerical_implementation} after three exact reformulations: restriction to the active
subset, reparameterization into dollar units, and decomposition into support-aware accountwise
proximal updates. Operationally, the accountwise step is then executed in parallel, with singleton
accounts handled in closed form and multi-asset support groups solved by cached warm-started
Gurobi models to high accuracy. The practical calibration details are collected in
Appendix~\ref{app:exchange-scale-admm}.

\subsection{Empirical Evaluation}
\label{sec:oct10-results}

\paragraph{Loss and performance metric.}
We evaluate the expected loss of each allocation using the finite-scenario sample-average approximation
\begin{equation}\label{eq:saa-obj}
L_S(x) \defeq \frac{1}{S}\sum_{s=1}^S \sum_{i=1}^n
\bigl(-E_i - (q_i - x_i)^\top \Delta p^{(s)}\bigr)_+,
\end{equation}
computed over $S$ independent GBM price-drop draws $\{\Delta p^{(s)}\}$, where
$\Delta p^{(s)} \defeq p_\tau - p_T^{(s)}$. Appendix~\ref{app:rr-robustness}
reports the corresponding Monte Carlo robustness check for this finite-scenario evaluation.
To avoid over-emphasizing a favorable single draw, the main text reports averages over $50$
independent $S=128$ scenario sets. We compare three allocations: the realized exchange allocation
$x^{\mathrm{Hyperliquid}}$, the Water-filling benchmark $x^{\mathrm{Water\text{-}filling}}$, and
the jointly calibrated Risk-optimal allocation $x^{\mathrm{Risk\text{-}optimal}}$. To measure
performance relative to the maximal amount of active-set risk that can be removed, we define the
relative risk-reduction statistic
\begin{equation}\label{eq:rr_def}
\mathrm{RR}(x) \defeq
\frac{L_S(0_{\supp(Q)}) - L_S(x)}
     {L_S(0_{\supp(Q)}) - L_S(x^{\mathrm{Risk\text{-}optimal}})},
\end{equation}
where $\supp(Q)\defeq\{k:Q^k\neq 0\}$ denotes the active asset set and
$0_{\supp(Q)}$ is the zero allocation restricted to these coordinates.
By construction, $\mathrm{RR}(x^{\mathrm{Risk\text{-}optimal}})=1$. The
realized Hyperliquid allocation and Water-filling take values strictly
below $1$ because they remove smaller fractions of the maximal reducible
active-set risk. Thus values closer to $1$ indicate better performance.

\paragraph{Expected loss.}

\begin{table}[H]
\centering
\small
\begin{tabular}{p{2.0cm} >{\raggedleft\arraybackslash}p{2.2cm} >{\raggedleft\arraybackslash}p{1.9cm} >{\raggedleft\arraybackslash}p{1.8cm} >{\raggedleft\arraybackslash}p{2cm}}
\toprule
 \textbf{Allocation} &
\textbf{Mean $L_{128}$} &
\textbf{s.d.\ of $L_{128}$} &
\textbf{Mean $\mathrm{RR}$} & \textbf{Runtime} \\
\midrule
  Hyperliquid &
   $5.383\times 10^8$ & $1.62\times 10^7$ & $0.325$ & -- \\
 Water-filling &
   $5.322\times 10^8$ & $1.54\times 10^7$ & $0.901$ & $3.70\text{ s}$ \\
  Risk-optimal& 
   $5.312\times 10^8$ & $1.56\times 10^7$ & $1.000$ & $1.15\times 10^5\text{ s}$ \\
\bottomrule
\end{tabular}
\caption{Comparison of out-of-sample sampled-loss performance and computational efficiency for the first ADL wave $\tau_1$. The first column lists the candidate allocations $x^{\mathrm{Hyperliquid}}$, $x^{\mathrm{Water\text{-}filling}}$, and $x^{\mathrm{Risk\text{-}optimal}}$. Columns two to four evaluate these fixed allocations over $50$ independent $S=128$ GBM/SAA scenario sets and report the mean and standard deviation of sampled expected loss together with the mean risk reduction defined in \cref{eq:rr_def}. The last column reports the wall-clock time required to derive each allocation from the target imbalance vector and reconstructed account states.}
\label{tab:hyperliquid-wave-results}
\end{table}

To construct the Risk-optimal allocation, we solve the SAA problem once on a
fixed reference scenario set of size $S=128$ for $\tau_1$ and use the
resulting optimizer as the candidate policy. Appendix~\ref{app:rr-convergence}
plots the associated fixed-sample convergence path. Along this run, the
relative risk reduction already reaches $0.952$ by iteration $500$ and $0.966$
by iteration $1{,}000$, so most economically meaningful improvement arrives
early; the long reported runtime mainly reflects the cost of driving the ADMM
residuals to convergence.

Table~\ref{tab:hyperliquid-wave-results} should therefore be read in two
parts. The first four columns report an out-of-sample Monte Carlo evaluation
over $50$ independent $S=128$ GBM/SAA scenario sets. Hyperliquid is observed
directly, and Water-filling is computed directly from the target ADL notional
and reconstructed account states, so neither allocation depends on
these evaluation samples. By contrast, the Risk-optimal allocation is
optimized on that fixed reference scenario set and then, like the other two
allocations, evaluated on each of the $50$ independent samples. The last
column instead reports the wall-clock time required to derive the allocations.

Across these $50$ out-of-sample evaluations, Hyperliquid attains
$L_{128}=5.383\times 10^8$ and $\mathrm{RR}=0.325$, Water-filling improves to
$L_{128}=5.322\times 10^8$ and $\mathrm{RR}=0.901$, and the Risk-optimal
allocation improves further to $L_{128}=5.312\times 10^8$ with
$\mathrm{RR}=1.000$. The corresponding standard deviations of the loss are all
of order $1.5\times 10^7$, indicating that the Monte Carlo dispersion is of
similar scale across the three candidate allocations. More importantly,
Appendix~\ref{app:rr-robustness} shows that the ordering
$L_S(x^{\mathrm{Risk\text{-}optimal}}) <
L_S(x^{\mathrm{Water\text{-}filling}}) <
L_S(x^{\mathrm{Hyperliquid}})$ is stable across the evaluation samples, which
is reassuring given the visible finite-sample variability. The
table is therefore best read as a loss--runtime tradeoff: Water-filling is the
fast constructive benchmark, while the Risk-optimal allocation closes the
remaining gap at substantially greater computational cost.

Appendix~\ref{app:rr-robustness} gives the fuller robustness discussion,
including the fixed-sample comparison on the reference scenario set, the corresponding $\tau_3$
exercise, and the comparison with the directional Water-filling benchmark,
which remains dominated by one-factor Water-filling. For $\tau_3$, the
reported Risk-optimal allocation is the terminal iterate at the iteration cap
$20{,}000$; the corresponding terminal residuals are reported in
Appendix~\ref{app:rr-convergence}.

\paragraph{Allocations.}
\label{sec:admm-eval-diagnostics}

The improvement delivered by the Risk-optimal allocation is associated
with structural differences in user-level targeting. The model-based
allocations do not merely touch a different number of users; they also
redistribute deleveraging notional across users in a systematically
different way from the realized Hyperliquid event. At $\tau_1$,
one-factor Water-filling is the most concentrated rule in the upper
tail, while the Risk-optimal allocation redistributes part of this mass
across a broader set of economically material users. The next
paragraphs make this distinction precise using the four panels of
Figure~\ref{fig:comparison-stats}, while Appendix~\ref{app:user-level-tables}
reports the corresponding thresholded user-count comparisons for
$\tau_1$ and $\tau_3$.

Before computing the diagnostics below, we apply a cleaning threshold to
the model-based allocations: entries below one cent are discarded for
this user-level analysis. Their total allocated volume is below
$\$50$ for both Water-filling and the Risk-optimal allocation, so they
have no economic relevance here. The realized Hyperliquid allocation
$x^{\mathrm{Hyperliquid}}$ is already aligned with this threshold; the corresponding
user-count comparisons are reported in Appendix~\ref{app:user-level-tables}.

Figure~\ref{fig:comparison-stats} collects four complementary user-level diagnostics for $\tau_1$. The first two panels are based on the per-user deleveraging notional $ \sum_{k=1}^d p_\tau^k\,\lvert x_i^k\rvert$. The upper-left and upper-right panels show distinct concentration patterns for each allocation. In the upper-left panel, the one-factor Water-filling
curve decays slowest and therefore carries the heaviest upper tail. At
$\tau_1$, its largest single-user deleveraging notional is
$\$6.2\times 10^7$, compared with $\$3.1\times 10^7$ for the
Risk-optimal allocation and $\$1.6\times 10^7$ for the realized
Hyperliquid event. This top user alone accounts for $40.2\%$ of total
Water-filling deleveraging, compared with $20.1\%$ for the
Risk-optimal allocation and $10.6\%$ for Hyperliquid. The upper-right
panel confirms the same pattern cumulatively: by the top ten users,
Water-filling already accounts for $76.9\%$ of total deleveraging,
compared with $57.3\%$ for the Risk-optimal allocation and $40.9\%$
for the realized Hyperliquid event. The Risk-optimal allocation remains
concentrated, but redistributes part of this mass across a broader set
of economically material accounts than Water-filling.

The lower-left panel plots the per-user deleveraging fraction
\[
\mathrm{frac}_i(x)\defeq
\frac{\sum_{k=1}^d p_\tau^k\,\min\{\lvert x_i^k\rvert,c_i^k(Q)\}}
     {\sum_{k=1}^d p_\tau^k\,c_i^k(Q)},
\qquad
c_i^k(Q)\defeq u_i^k-l_i^k,
\]
where $c_i^k(Q)$ is the directional capacity of account $i$ in asset
$k$, so that $\mathrm{frac}_i(x)\in[0,1]$ measures the fraction of the user's
deleverageable portfolio that is removed. The realized Hyperliquid event
places substantially more mass near $\mathrm{frac}_i(x)=1$, fully deleveraging a
large fraction of affected accounts. By contrast, both model-based
allocations usually remove only a moderate fraction of directional
capacity.

The lower-right panel reinforces this picture. Restricting attention to
accounts with pre-ADL equity exceeding one dollar, we order users by
pre-ADL factor leverage $\ell_i^{(v)}(0)$ as defined in
\cref{def:factor_leverage} and plot the mean deleveraging fraction by
decile. The Water-filling
allocation is the lowest curve of the three and increases more smoothly
with factor leverage. Risk-optimal deleverages larger fractions on
average and rises more steeply, while the realized Hyperliquid event is
substantially less ordered, assigning large deleveraging fractions
across a broad range of factor-leverage levels. 

Appendix~\ref{app:user-level-tables}
reports the corresponding $\tau_3$ diagnostics together with the
thresholded user-count comparison and the analogous concentration
patterns in the upper-left and upper-right panels. At $\tau_3$, the
lower-left panel remains qualitatively similar, and the Water-filling
and Risk-optimal decile profiles in the lower-right panel remain close
to their $\tau_1$ counterparts, but the realized Hyperliquid profile is
visibly more ordered. At the same time, the upper-left and upper-right
panels show that the relative concentration ranking of Hyperliquid and
Water-filling reverses in the upper tail, so the realized allocation is
less far from the model-based rules when the active set is restricted to
only four coins. Even in this case, however, the main distortion
remains the same: Hyperliquid still places substantially more mass near
full deleveraging of affected accounts.

\begin{figure}[H]
\centering
\includegraphics[width=0.95\textwidth]{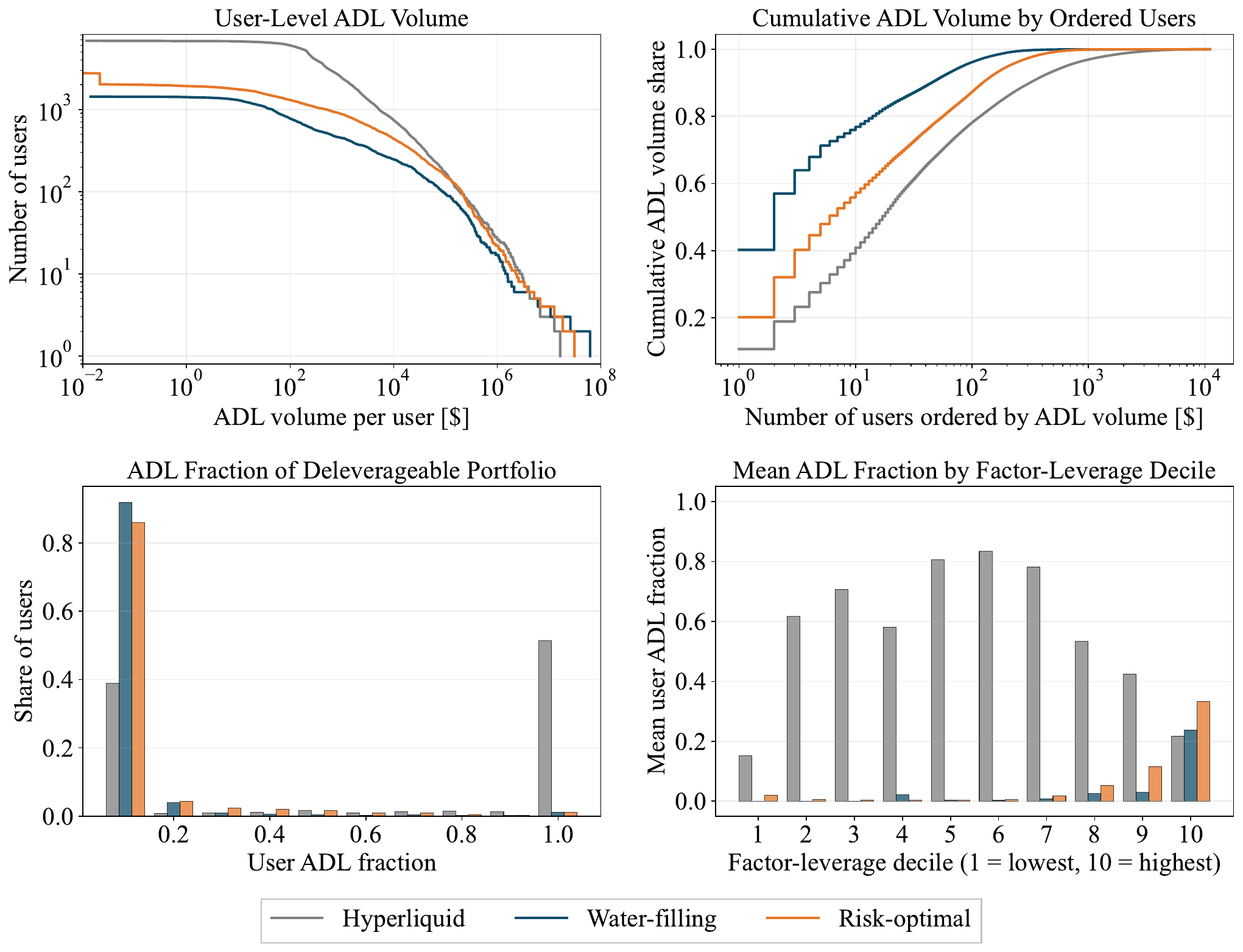}
\caption{User-level diagnostics for $\tau_1$: user count above
  deleveraging-notional threshold, cumulative top-share concentration curve, histogram of
  deleveraging fractions, and factor-leverage-decile targeting plot.}
\label{fig:comparison-stats}
\end{figure}

\paragraph{Code and data availability.}
Code and cleaned data supporting the findings are available in the accompanying GitHub repository at \url{https://github.com/nataschahey/adl-paper-repro}.


\clearpage

\bibliographystyle{abbrvnat}
\bibliography{biblio}

\newpage
\appendix

\section{Details and Proofs for Single-Asset Isolated Margining}\label{app:isolated.margining}

We begin this supplementary appendix with a lemma that will be used repeatedly to prove properties of the ADL objectives.

\begin{lemma}\label{lem:shortfall_prop}
For each account $i$,
\begin{enumerate}
    \item[(i)] the shortfall $\sigma_i(x_i,p)$ is convex and piecewise affine in each argument, and
    \item[(ii)] for every $x_i\in[0,q_i]$ the map $p\mapsto \sigma_i(x_i,p)$ is nondecreasing.
\end{enumerate}
Consequently, for each $p$ the loss $x\mapsto \Lscr(x,p)=\sum_{i=1}^n\sigma_i(x_i,p)$ is convex and piecewise affine,
and for each $x$ the map $p\mapsto \Lscr(x,p)$ is convex, piecewise affine, and nondecreasing.
\end{lemma}

\begin{proof}
Fix $i$. For $x_i\in[0,q_i]$ and $p\in\R$ we can write
\[
-e_i(x_i,p)=(q_i-x_i)p + x_ip_\tau - q_ip_i^{(e)}-m_i,
\]
which is affine in $x_i$ for fixed $p$, and affine in $p$ for fixed $x_i$, with slope $q_i-x_i\ge 0$ in $p$.
Since $u\mapsto (u)_+$ is convex and nondecreasing, the composition
$\sigma_i(x_i,p)=(-e_i(x_i,p))_+$ is convex and piecewise affine in $x_i$ and in $p$, and is nondecreasing in $p$.
The claims for $\Lscr$ follow by summation.
\end{proof}

\subsection{Proof of Theorem~\ref{thm:risk_neutral}}\label{app:rn.thm}

The proof of Theorem~\ref{thm:risk_neutral} will be tackled in parts. First, we will show that if $x$ solves \eqref{eq:opt_rn} then (c) holds (i.e., (a) $\Rightarrow$ (c)). Next, we will characterize the unique solution to (b). We will complete the equivalence of (a)--(c) by showing that any admissible $x\in \mathcal{X}$ satisfying (c) is given by the solution to (b) \emph{and} the solution to (b) solves (a) (i.e., (c) $\Rightarrow$ (b) $\Rightarrow$ (a)). Finally, we incorporate the pointwise statement (d) by proving that the water-filling
allocation is a simultaneous pointwise minimizer of realized loss (which gives
(b) $\Rightarrow$ (d)) and by observing that (d) $\Rightarrow$ (a) follows from evaluating
at $p=p_T$ and taking expectations.

In particular, the statement of the theorem follows from
Propositions~\ref{prop:equalization} and~\ref{prop:minimax_solution},
Lemma~\ref{lem:pointwise_water_filling}, and
Proposition~\ref{prop:triple_equiv} below.


\subsubsection{Leverage Equalization}
We begin with some preliminaries. Under Assumption~\ref{ass:feasible} and the standing condition $Q>0$, Slater's condition is trivially satisfied for the optimization problem \eqref{eq:opt_rn} since the point
\[
\bar x_i \;\defeq\; \frac{Q}{\sum_{j=1}^n q_j}\,q_i,\qquad i=1,\dots,n
\]
satisfies $0<\bar x_i<q_i$ for all $i$ and $\sum_i \bar x_i=Q$.
Hence strong duality holds and the KKT conditions are necessary and sufficient for optimality.

Introduce multipliers $\lambda\in\R$ for $\sum_i x_i=Q$,
$\nu_i\ge 0$ for $-x_i\le 0$, and $\mu_i\ge 0$ for $x_i-q_i\le 0$.
An optimal primal--dual quadruple $(x^\star,\lambda^\star,\nu^\star,\mu^\star)$ satisfies
\begin{align}
&0\in \partial_{x_i}V(x^\star)+\lambda^\star-\nu_i^\star+\mu_i^\star,
\quad i=1,\dots,n, \label{eq:kkt_stat}\\
&x^\star\in\mathcal X, \label{eq:kkt_primal}\\
&\nu_i^\star\ge 0,\ \mu_i^\star\ge 0,\quad i=1,\dots,n, \label{eq:kkt_dual}\\
&\nu_i^\star x_i^\star=0,\quad \mu_i^\star(x_i^\star-q_i)=0,\quad i=1,\dots,n. \label{eq:kkt_cs}
\end{align}

For each fixed $p$, we have $\frac{\partial}{\partial x_i}(-e_i(x_i,p))=p_\tau-p$.
Therefore, the subdifferential of the individual account shortfall satisfies
\[
\partial_{x_i}\sigma_i(x_i,p)=\partial_{x_i}\big((-e_i(x_i,p))_+\big)
=
(p_\tau-p)\cdot
\begin{cases}
\{1\}, & e_i(x_i,p)<0,\\
[0,1], & e_i(x_i,p)=0,\\
\{0\}, & e_i(x_i,p)>0.
\end{cases}
\]
Since $e_i(x_i,\cdot)$ is affine in $p_T$ with nonzero slope for $x_i<q_i$, Assumption~\ref{ass:p_T_regular} implies that attaining \emph{exactly} zero equity is a probability zero event, $\PR\big(e_i(x_i,p_T)=0\big)=0$. As a result, we can obtain the pointwise partial derivative for the objective
\begin{equation}\label{eq:grad_rn_noatom}
\frac{\partial V}{\partial x_i}(x)
=
\E\!\left[(p_\tau-p_T)\,\I{e_i(x_i,p_T)\le 0}\right], \quad x_i\in[0,q_i).
\end{equation}
This is justified by Assumption~\ref{ass:p_T_regular}, which implies that $\sigma_i(x_i,p_T)$ is differentiable almost surely for $x_i\in[0,q_i)$ and allows us to pass the $x_i$-derivative inside the expectation by dominated convergence.

Observe that the insolvency event $e_i(x_i,p_T)\le 0$ for each account admits a threshold representation in terms of $p_T$ and the account's \emph{bankruptcy (zero-equity) price} $p_i^{(z)}(x_i)$:
\begin{equation}\label{eq:ze_price}
e_i(x_i,p_T)\le 0
\quad\Longleftrightarrow\quad
p_T\ge p_i^{(z)}(x_i),
\qquad
p_i^{(z)}(x_i)\defeq p_\tau+\frac{E_i}{q_i-x_i},
\end{equation}
with the convention $p_i^{(z)}(q_i)=+\infty$. Using \eqref{eq:def_leverage}, this can be rewritten as
\begin{equation}\label{eq:ze_price_leverage}
p_i^{(z)}(x_i)=p_\tau\Big(1+\ell_i(x_i)^{-1}\Big),
\end{equation}
with the convention $\ell^{-1}=+\infty$ when $\ell=0$. Combining \eqref{eq:grad_rn_noatom} and \eqref{eq:ze_price} gives
\begin{equation}\label{eq:grad_rn_noatom_2}
\frac{\partial V}{\partial x_i}(x)
=
\E\!\left[(p_\tau-p_T)\,\I{p_T\ge p_i^{(z)}(x_i)}\right].
\end{equation}
On $\{p_T\ge p_i^{(z)}(x_i)\}$ we have $p_T\ge p_i^{(z)}(x_i)\ge p_\tau$, hence $(p_\tau-p_T)\le 0$ and
$\partial V/\partial x_i(x)\le 0$.

Using \eqref{eq:ze_price_leverage}, the event $\{p_T\ge p_i^{(z)}(x_i)\}$ is equivalent to
$\{p_T-p_\tau\ge p_\tau \ell_i(x_i)^{-1}\}$. Define next, for $\ell\in[0,\infty)$, the marginal expected shortfall exposure, $v(\ell)$, of an account with leverage $\ell$:
\begin{equation}
\label{eq:def_marginal_loss_leverage}
v(\ell)\defeq
\E\!\left[(p_T-p_\tau)\,\I{\,p_T-p_\tau \ge p_\tau \ell^{-1}}\right],
\qquad v(0)=0.
\end{equation}
Equivalently, $v(\ell)$ is the ``tail expectation'' of $(p_T-p_\tau)$ above a leverage-indexed threshold. For $x_i\in[0,q_i)$, we arrive at the derivative representation
\begin{equation}\label{eq:grad_equals_v}
\frac{\partial V}{\partial x_i}(x)
= -\,v(\ell_i(x_i)), \quad x_i\in[0,q_i),
\end{equation}
which admits a continuous extension to $x_i=q_i$ with value $0$.

Indeed, as $x_i\uparrow q_i$ we have $\ell_i(x_i)\downarrow 0$, and by definition $v(0)=0$. Moreover, for $\ell>0$ we have
\[
v(\ell)=\E\!\left[(p_T-p_\tau)\,\I{\,p_T-p_\tau \ge p_\tau \ell^{-1}}\right].
\]
As $\ell\downarrow 0$, the threshold $p_\tau\ell^{-1}\to+\infty$, so the indicator decreases pointwise to zero and is dominated by $1$, while
\[
0\le (p_T-p_\tau)\,\I{\,p_T-p_\tau \ge p_\tau \ell^{-1}} \le (p_T-p_\tau)_+,
\]
with $(p_T-p_\tau)_+$ integrable under our standing assumptions. By dominated convergence,
\[
\lim_{\ell\downarrow 0}v(\ell)=0=v(0).
\]
Therefore,
\begin{equation}\label{eqn:subgradient.estimate}
\lim_{x_i\uparrow q_i}\frac{\partial V}{\partial x_i}(x)
=
-\lim_{\ell\downarrow 0} v(\ell)
=
-v(0)=0.
\end{equation}
Thus the marginal effect of additional buyback vanishes as the account is fully closed. Since $V(\cdot)$ is convex in $x_i$, this implies that any element of the subgradient set $g_i\in\partial_{x_i}V(x)$ when $x_i=q_i$ satisfies $g_i\geq 0$.

Since $\ell_i(\cdot)$ is strictly decreasing in $x_i$ and $v(\cdot)$ is nondecreasing in $\ell$,
the map $x_i\mapsto \frac{\partial V}{\partial x_i}(x)$ is nondecreasing. Under our assumptions it is straightforward to check that $v$ is strictly increasing on
$[0,\max_i \ell_i(0)]$.

\begin{lemma}\label{lem:v_strict}
$v(\cdot)$ is strictly increasing on $[0,\max_i \ell_i(0)]$.
\end{lemma}

\begin{proof}
A simple sufficient condition for this monotonicity is given by Assumption \ref{ass:p_T_regular}. Namely, that $p_T$ admits a density $f_T$
such that $f_T(p)>0$ for all $p\ge p_\tau$.
Indeed, let $a(\ell)\defeq p_\tau(1+\ell^{-1})$ (with $a(0)=+\infty$). For any $0\le \ell_1<\ell_2\le \max_i \ell_i(0)$,
\[
v(\ell_2)-v(\ell_1)
=
\E\!\left[(p_T-p_\tau)\,\I{a(\ell_2)\le p_T<a(\ell_1)}\right]
=
\int_{a(\ell_2)}^{a(\ell_1)} (p-p_\tau)\,f_T(p)\,dp
>0.
\]
The claim follows.
\end{proof}

The following proposition shows that (a) implies (c).

\begin{proposition}[Equalization of post-ADL leverage]\label{prop:equalization}
Under Assumptions~\ref{ass:feasible}--\ref{ass:p_T_regular}, if $x^\star$ solves \eqref{eq:opt_rn} then $x_i^\star<q_i$ for all $i$ and there exists a cutoff level $\bar\ell\in[0,\max_i\ell_i(0)]$ such that
\begin{equation*}
\begin{aligned}
&x_i^\star>0 &&\Longrightarrow\quad \ell_i(x_i^\star)=\bar\ell,\\
&x_i^\star=0 &&\Longrightarrow\quad \ell_i(0)\le \bar\ell.
\end{aligned}
\end{equation*}
In particular, all accounts whose positions are reduced share the same post-ADL leverage.
\end{proposition}

\begin{proof}
We verify the claims of Proposition~\ref{prop:equalization} in turn.

First, we show that $x_i^\star<q_i$ for all $i$. Suppose for contradiction that $x_k^\star=q_k$ for some $k$. Then $\nu_k^\star=0$ (since $x_k^\star>0$) and $\ell_k(x_k^\star)=0$, so using \eqref{eq:kkt_stat} we get \[ 0\in \partial_{x_k}V(x^\star)+\lambda^\star+\mu_k^\star\]
Hence there exists $g_k\in \partial_{x_k}V(x^\star)$ such that
\[
0=g_k+\lambda^\star+\mu_k^\star.
\]
As noted earlier (by \eqref{eqn:subgradient.estimate}) we have $g_k\ge 0$, so $\lambda^\star+\mu_k^\star=-g_k\le 0$, and therefore $\lambda^\star\le-\mu_k^\star\le 0$.
On the other hand, since $Q<\sum_i q_i$, there exists an index $j$ with $x_j^\star<q_j$, so $\mu_j^\star=0$. If $x_j^\star>0$, then $\nu_j^\star=0$ and \eqref{eq:kkt_stat} gives $\lambda^\star=v(\ell_j(x_j^\star))>0$, because $\ell_j(x_j^\star)>0$ and $v$ is strictly increasing with $v(0)=0$. If instead $x_j^\star=0$, then \eqref{eq:kkt_stat} gives $\lambda^\star=v(\ell_j(0))+\nu_j^\star\ge v(\ell_j(0))>0$. In either case $\lambda^\star>0$, contradicting $\lambda^\star\le 0$. Therefore $x_i^\star<q_i$ for all $i$.

Since $Q>0$ and $\sum_i x_i^\star=Q$, there exists at least one index $h$ such that $x_h^\star>0$. Combined with the previous step, this yields $0<x_h^\star<q_h$. By \eqref{eq:kkt_stat}--\eqref{eq:kkt_cs} and \eqref{eq:grad_equals_v},
\[
-v(\ell_h(x_h^\star))+\lambda^\star=0,
\]
so $\lambda^\star=v(\ell_h(x_h^\star))$. By Lemma~\ref{lem:v_strict}, $v$ is invertible on $[0,\max_i\ell_i(0)]$, and we may therefore define
\[
\bar\ell \defeq v^{-1}(\lambda^\star)\in[0,\max_i\ell_i(0)].
\]

Now let $i$ be such that $0<x_i^\star<q_i$. Then $\nu_i^\star=\mu_i^\star=0$, and \eqref{eq:kkt_stat} together with \eqref{eq:grad_equals_v} gives
\[
-v(\ell_i(x_i^\star))+\lambda^\star=0.
\]
Hence $v(\ell_i(x_i^\star))=\lambda^\star=v(\bar\ell)$, and strict monotonicity of $v$ implies
\[
\ell_i(x_i^\star)=\bar\ell.
\]

Finally, if $x_i^\star=0$, then $\mu_i^\star=0$ and \eqref{eq:kkt_stat} gives
\[
-v(\ell_i(0))+\lambda^\star-\nu_i^\star=0,
\]
hence $v(\ell_i(0))\le \lambda^\star$ because $\nu_i^\star\ge 0$. Since $\lambda^\star=v(\bar\ell)$ and $v$ is strictly increasing, we conclude that
\[
\ell_i(0)\le \bar\ell.
\]
\end{proof}

\subsubsection{The Minimax Problem}\label{sec:minimax}
We study the minimax leverage program \eqref{eq:minmax}.
Fix a candidate leverage cap $t\ge 0$ and ask what
minimum buyback from each account is required to ensure $\ell_i(x_i)\le t$.
In view of expression~\eqref{eq:def_leverage}
the constraint $\ell_i(x_i)\le t$ is equivalent to
\begin{equation}\label{eq:cap_ineq}
x_i \ge q_i-\frac{E_i}{p_\tau}t,
\end{equation}
together with $x_i\ge 0$. This motivates the definition of the vector function $y(t)\defeq (y_1(t),\dots,y_n(t))$ where
\[
y_i(t)\;\defeq\;\Big(q_i-\frac{E_i}{p_\tau}\,t\Big)_+,\qquad i=1,\dots,n,
\]
which is the minimal buyback required from account $i$ under the leverage cap $t$. In particular, setting $x_i=y_i(t)$ yields
\begin{equation*}
\ell_i(y_i(t))=\min\{\ell_i(0),\,t\},\qquad i=1,\dots,n.
\end{equation*}

\begin{lemma}\label{lem:dominance_leverage_cap}
For any $x\in\mathcal X$,
\[
\max_{1\le i\le n}\ell_i(x_i)\le t
\quad\Longrightarrow\quad
x_i\ge y_i(t)\ \ \forall i \ \ \text{and} \ \ Q\;\ge\; \sum_{i=1}^n y_i(t).
\]
\end{lemma}

\begin{proof}
Fix $t\ge 0$ and $x\in\mathcal X$ with $\max_i\ell_i(x_i)\le t$. Then $\ell_i(x_i)\le t$ for each $i$, and hence
\eqref{eq:cap_ineq} holds componentwise. Combining \eqref{eq:cap_ineq} with $x_i\ge 0$ yields
$x_i\ge \big(q_i-\frac{E_i}{p_\tau}t\big)_+=y_i(t)$ for all $i$. Summing over $i$ and using $\sum_i x_i=Q$ gives
$Q\ge \sum_i y_i(t)=G(t)$.
\end{proof}

Summing these per-account requirements yields the total buyback needed to enforce the cap $t$ across the system.
Accordingly, define
\[
G(t)\defeq \sum_{i=1}^n y_i(t), \quad t\ge 0,
\]
which we interpret as the aggregate deleveraging demand induced by the leverage level $t$. This function has several key properties and is critical to the solution of the minimax problem.

\begin{lemma}\label{lem:G_properties_unique_level}\mbox{}
\begin{enumerate}
\item[(i)] $G$ is continuous and nonincreasing on $[0,\infty)$.
\item[(ii)] $G(0)=\sum_{i=1}^n q_i$, and $G(t)=0$ for all $t\ge \max_{1\le i\le n}\ell_i(0)$.
\item[(iii)] $G$ is strictly decreasing on $\big[0,\max_i \ell_i(0)\big]$.
\item[(iv)] There exists a unique $t^\star\in\big(0,\max_i \ell_i(0)\big)$ such that $G(t^\star)=Q$.
\end{enumerate}
\end{lemma}

\begin{proof}
    We treat each statement in turn.

(i) is immediate.

(ii) At $t=0$, $y_i(0)=q_i$, so $G(0)=\sum_i q_i$. For $t\ge \max_i \ell_i(0)$ we have, for each $i$,
\[
t\ge \ell_i(0)=\frac{p_\tau q_i}{E_i}
\quad\Longleftrightarrow\quad
q_i-\frac{E_i}{p_\tau}t\le 0
\quad\Longleftrightarrow\quad
y_i(t)=0,
\]
hence $G(t)=0$.

(iii) Fix $0\le t_1<t_2\le \max_i\ell_i(0)$ and choose an index $k$ such that
$\ell_k(0)=\max_i\ell_i(0)$. Then, for any
$t\in[0,\ell_k(0)]$ we have $t\le \ell_k(0)=\frac{p_\tau q_k}{E_k}$. Equivalently, $q_k-\frac{E_k}{p_\tau}t\ge 0$ and thus
\[
y_k(t)=q_k-\frac{E_k}{p_\tau}t \qquad \text{for all }t\in[0,\ell_k(0)].
\]
In particular, $y_k(t_2)<y_k(t_1)$. On the other hand, for every $j$, $y_j(\cdot)$ is nonincreasing, so
$y_j(t_2)\le y_j(t_1)$. Summing across indices yields $G(t_2)<G(t_1)$. Since $t_1<t_2$ were arbitrary in $\big[0,\max_i \ell_i(0)\big]$ the strict decrease follows.

(iv) By (ii) and the standing assumption $0<Q<\sum_i q_i$, we have $G(0)>\!Q$ and
$G(\max_i\ell_i(0))=0<\!Q$. By continuity from (i), the intermediate value theorem yields
existence of $t^\star\in(0,\max_i\ell_i(0))$ such that $G(t^\star)=Q$. Uniqueness follows from
strict decrease in (iii).
\end{proof}

Combining the preceding definitions and properties, we obtain the minimax solution.

\begin{proposition}[Minimax solution]\label{prop:minimax_solution} Let $t^\star$ be the unique level from
Lemma~\ref{lem:G_properties_unique_level}(iv) such that $G(t^\star)=Q$, and define $x^\star\defeq
y(t^\star)$.
Then $x^\star\in\mathcal X$, $x^\star$ is the unique optimizer of the minimax leverage problem \eqref{eq:minmax}
and $\max_{1\le i\le n}\ell_i(x_i^\star)=t^\star$.
\end{proposition}

\begin{proof}
We break up the proof into a verification of the component claims.

\emph{Feasibility.} By definition, $y_i(t^\star)\in[0,q_i]$ for all $i$ (since $(\cdot)_+\ge 0$ and
$q_i-\frac{E_i}{p_\tau}t^\star\le q_i$), and by Lemma~\ref{lem:G_properties_unique_level}(iv),
$\sum_i y_i(t^\star)=G(t^\star)=Q$. Hence $x^\star\in\mathcal X$.

\emph{Optimality.} Let $x\in\mathcal X$ be arbitrary and set $t\defeq \max_i \ell_i(x_i)\ge 0$.
Since each $\ell_i(\cdot)$ is decreasing on $[0,q_i]$ and $x_i\ge 0$, we have $\ell_i(x_i)\le \ell_i(0)$ for all $i$,
and hence $t\le \max_i \ell_i(0)$.
Then $\max_i \ell_i(x_i)\le t$ holds trivially, so by Lemma~\ref{lem:dominance_leverage_cap}
\[
Q=\sum_{i=1}^n x_i \;\ge\; \sum_{i=1}^n y_i(t)=G(t).
\]
Because $G$ is strictly decreasing on $[0,\max_i\ell_i(0)]$ and $G(t^\star)=Q$
(Lemma~\ref{lem:G_properties_unique_level}), the inequality $Q\ge G(t)$ implies $t\ge t^\star$.
Thus every feasible $x$ satisfies $\max_i \ell_i(x_i)\ge t^\star$, while $x^\star$ achieves
$\max_i \ell_i(x_i^\star)=t^\star$. Therefore $x^\star$ is optimal and the optimal value is $t^\star$.

\emph{Uniqueness.} Now let $x\in\mathcal X$ be any optimal solution. Then $\max_i \ell_i(x_i)\le t^\star$, which implies
$x_i\ge y_i(t^\star)$ for all $i$. Summing and using $\sum_i x_i=\sum_i y_i(t^\star)=Q$ yields
\[
0=\sum_{i=1}^n\big(x_i-y_i(t^\star)\big)
\qquad\text{with}\qquad x_i-y_i(t^\star)\ge 0\ \ \forall i,
\]
which forces $x_i=y_i(t^\star)$ for all $i$. Hence the optimizer is unique and equals $x^\star=y(t^\star)$.

\emph{Maximum Leverage.} Fix $i$. If $x_i^\star>0$, then
$x_i^\star=q_i-\frac{E_i}{p_\tau}t^\star$, so
\[
\ell_i(x_i^\star)=\frac{p_\tau(q_i-x_i^\star)}{E_i}
=\frac{p_\tau\left(\frac{E_i}{p_\tau}t^\star\right)}{E_i}=t^\star.
\]
If $x_i^\star=0$, then $q_i-\frac{E_i}{p_\tau}t^\star\le 0$, i.e.\ $t^\star\ge \frac{p_\tau q_i}{E_i}=\ell_i(0)$,
so $\ell_i(x_i^\star)=\ell_i(0)\le t^\star$. Therefore $\max_i \ell_i(x_i^\star)\le t^\star$.
Since $Q>0$, at least one component of $x^\star$ is strictly positive, so for that index the leverage equals $t^\star$,
and hence $\max_i \ell_i(x_i^\star)=t^\star$.
\end{proof}

\subsubsection{Pointwise Optimality}

We next show that the water-filling allocation identified above is optimal not only after
averaging over the terminal price distribution, but also pointwise at every price level. It is a trivial consequence of pointwise optimality that (d) $\Rightarrow$ (a).

\begin{lemma}[Pointwise optimality of water-filling]\label{lem:pointwise_water_filling}
Let Assumptions~\ref{ass:feasible}--\ref{ass:solvent} hold, and let $x^\star=y(t^\star)$
be the water-filling allocation from Proposition~\ref{prop:minimax_solution}. Then
\[
\mathcal{L}(x^\star,p)\le \mathcal{L}(x,p),
\qquad x\in\mathcal X,\ p\in\mathbb R.
\]
\end{lemma}

\begin{proof}
For $x\in\mathcal X$, set $r_i:=q_i-x_i$ and introduce the leverage variable
$u_i:=\ell_i(x_i)=p_\tau r_i/E_i$. Since $x_i=q_i-E_i u_i/p_\tau$, the feasibility
constraints $x\in\mathcal X$ are equivalent to
\[
0\le u_i\le \bar u_i:=\frac{p_\tau q_i}{E_i},\qquad i=1,\ldots,n,
\]
together with
\[
\sum_{i=1}^n E_i u_i=M,\qquad
M:=p_\tau\left(\sum_{i=1}^n q_i-Q\right).
\]
Thus the feasible set in leverage variables is
\[
\mathcal U:=
\left\{
u\in\mathbb R_+^n:
0\le u_i\le \bar u_i,\ 
\sum_{i=1}^n E_i u_i=M
\right\}.
\]
By Proposition~\ref{prop:minimax_solution}, the water-filling allocation $x^\star=y(t^\star)$
corresponds to the leverage vector $u_i^\star=\min\{\bar u_i,t^\star\}$.

Fix $p\in\mathbb R$. If $p\le p_\tau$, then
$e_i(x_i,p)=E_i+(q_i-x_i)(p_\tau-p)\ge E_i>0$ for every $i$ and every feasible $x$.
Hence $\mathcal{L}(x,p)=0$ for all $x\in\mathcal X$, and the desired inequality is immediate.

Suppose now that $p>p_\tau$. Then
\[
\mathcal{L}(x,p)=\sum_{i=1}^n E_i\phi_p(u_i),
\qquad
\phi_p(u):=\left(\frac{p-p_\tau}{p_\tau}u-1\right)_+.
\]
The function $\phi_p$ is finite-valued, convex, and nondecreasing on $[0,\infty)$.

We prove the following more general claim: for every finite-valued convex nondecreasing
function $\phi:[0,\infty)\to\mathbb R$, the vector $u^\star$ minimizes
$\sum_i E_i\phi(u_i)$ over $\mathcal U$. Choose $g^\star\in\partial\phi(t^\star)$.
For each $i$ with $\bar u_i>t^\star$, set $h_i:=g^\star$. For each $i$ with
$\bar u_i\le t^\star$, choose $h_i\in\partial\phi(\bar u_i)$ such that
$h_i\le g^\star$; this is possible by monotonicity of the subdifferential of a convex
function, and if $\bar u_i=t^\star$ we simply take $h_i=g^\star$.

Let $u\in\mathcal U$. By the subgradient inequality,
\[
\sum_{i=1}^n E_i\{\phi(u_i)-\phi(u_i^\star)\}
\ge
\sum_{i=1}^n E_i h_i(u_i-u_i^\star).
\]
If $\bar u_i>t^\star$, then $h_i=g^\star$. If $\bar u_i\le t^\star$, then
$u_i^\star=\bar u_i$ and $u_i-u_i^\star\le 0$, while $h_i\le g^\star$; hence
$h_i(u_i-u_i^\star)\ge g^\star(u_i-u_i^\star)$. Therefore
\[
\sum_{i=1}^n E_i h_i(u_i-u_i^\star)
\ge
g^\star\sum_{i=1}^n E_i(u_i-u_i^\star)
=0,
\]
where the last equality uses $u,u^\star\in\mathcal U$. Thus $u^\star$ minimizes
$\sum_i E_i\phi(u_i)$ over $\mathcal U$. Applying this claim to $\phi=\phi_p$ proves
$\mathcal{L}(x^\star,p)\le\mathcal{L}(x,p)$ for all $x\in\mathcal X$ when $p>p_\tau$.
Together with the case $p\le p_\tau$, this proves the result.
\end{proof}

\subsubsection{The ADL--Minimax--Equalization Equivalence}

The preceding sections reported the unique solution to the minimax problem (b) and verified that (b) $\Rightarrow$ (d) $\Rightarrow$ (a). To complete the proof of Theorem~\ref{thm:risk_neutral} we present a final proposition that completes the equivalence with the relation (c) $\Rightarrow$ (b) $\Rightarrow$ (a).

\begin{proposition}[Equivalence]\label{prop:triple_equiv}
Under Assumptions~\ref{ass:feasible}--\ref{ass:p_T_regular},  for any $x\in\mathcal X$, (a), (b), (c) and (d) are equivalent.
\end{proposition}

\begin{proof}
Since  we have already shown (b) $\Rightarrow$ (d) $\Rightarrow$ (a), it suffices to show the equivalence of (a)--(c).

By Lemma~\ref{lem:G_properties_unique_level}, there exists a unique $t^\star\in(0,\max_i\ell_i(0))$
such that $G(t^\star)=Q$. By Proposition~\ref{prop:minimax_solution}, the minimax leverage problem \eqref{eq:minmax}
has the unique optimizer $y(t^\star)$.

\emph{(a) $\Rightarrow$ (c).} This follows directly from Proposition \ref{prop:equalization}.

\emph{(c) $\Rightarrow$ (b).}  If $x_i>0$, then $\ell_i(x_i)=t$, and since
$\ell_i(x)=\frac{p_\tau(q_i-x)}{E_i}$ this implies
$x_i=q_i-\frac{E_i}{p_\tau}t>0$, hence
$x_i=\big(q_i-\frac{E_i}{p_\tau}t\big)_+$.
If instead $x_i=0$, then $\ell_i(0)\le t$ implies $\frac{p_\tau q_i}{E_i}\le t$, i.e.\
$q_i-\frac{E_i}{p_\tau}t\le 0$, hence again
$x_i=\big(q_i-\frac{E_i}{p_\tau}t\big)_+$.
Thus $x=y(t)$.

Using feasibility $\sum_i x_i=Q$, we have
\[
Q=\sum_{i=1}^n x_i
=\sum_{i=1}^n \Big(q_i-\frac{E_i}{p_\tau}t\Big)_+
=G(t).
\]
Since the equation $G(t)=Q$ has the unique solution $t^\star$, it follows that $t=t^\star$ and therefore
$x=y(t^\star)=x^\star$.

\emph{(b) $\Rightarrow$ (a).} A solution to (a) trivially exists by the Weierstrass Extreme Value Theorem since $\mathcal{X}$ is compact and the objective is continuous. Since any solution to (a) satisfies (c), and any feasible allocation satisfying (c) must coincide with the unique minimax optimizer from (b), every solution to (a) equals the solution to (b).

This proves the equivalence and uniqueness.
\end{proof}

\subsection{Proofs of \cref{prop:leverage_ordering,thm:sybil_waterfilling}}\label{app:pf.sybil.waterfilling}

\begin{proof}[Proof of \cref{prop:leverage_ordering}]
Since $p > 0$ and $E_i(p), E_j(p) > 0$, the inequality $\ell_i(p) \geq \ell_j(p)$ is equivalent to $q_i\, E_j(p) - q_j\, E_i(p) \geq 0$. Expanding,
\begin{align*}
q_i\, E_j(p) - q_j\, E_i(p)
&= q_i\bigl[q_j(p_j^{(e)} - p) + m_j\bigr] - q_j\bigl[q_i(p_i^{(e)} - p) + m_i\bigr] \\
&= q_i q_j\bigl(p_j^{(e)} - p_i^{(e)}\bigr) + q_i\, m_j - q_j\, m_i.
\end{align*}
All terms involving $p$ cancel, so the sign of this expression is independent of~$p$.
\end{proof}

\begin{proof}[Proof of \cref{thm:sybil_waterfilling}]
For any fixed $t\ge0$, define the attacker's per-account lower bound
\[
y_k(t)=\Big(q_k-\frac{E_k}{p_\tau}t\Big)_+,
\qquad
y_A(t)=\Big(q^A-\frac{E^A}{p_\tau}t\Big)_+.
\]
Since $(\cdot)_+$ is subadditive, we have
\begin{equation}\label{eq:sybil_pointwise}
\sum_{k=1}^K y_k(t)
=
\sum_{k=1}^K \Big(q_k-\frac{E_k}{p_\tau}t\Big)_+
\ \ge\
\Big(\sum_{k=1}^K q_k-\frac{\sum_{k=1}^K E_k}{p_\tau}t\Big)_+
=
\Big(q^A-\frac{E^A}{p_\tau}t\Big)_+
=
y_A(t).
\end{equation}
Let
\[
G_{\mathrm{oth}}(t)\defeq \sum_{j\in\mathcal N} y_j(t),\qquad
G_{\mathrm{att}}^{K}(t)\defeq\sum_{k=1}^K y_k(t),\qquad
G_{\mathrm{att}}^{1}(t)\defeq y_A(t),
\]
and define total demand functions
\[
G_{\mathrm{tot}}^{K}(t)\defeq G_{\mathrm{oth}}(t)+G_{\mathrm{att}}^{K}(t),
\qquad
G_{\mathrm{tot}}^{1}(t)\defeq G_{\mathrm{oth}}(t)+G_{\mathrm{att}}^{1}(t).
\]
By \eqref{eq:sybil_pointwise}, $G_{\mathrm{tot}}^{K}(t)\ge G_{\mathrm{tot}}^{1}(t)$ for all $t\ge0$.

Let $t_K^\star$ and $t_1^\star$ be the unique solutions to
$G_{\mathrm{tot}}^{K}(t)=Q$ and $G_{\mathrm{tot}}^{1}(t)=Q$, respectively (uniqueness follows from the strict
decrease of the corresponding total water-level function on the relevant interval, as in Lemma~\ref{lem:G_properties_unique_level}).
Evaluating at $t=t_1^\star$ gives
\[
G_{\mathrm{tot}}^{K}(t_1^\star)\ \ge\ G_{\mathrm{tot}}^{1}(t_1^\star)=Q=G_{\mathrm{tot}}^{K}(t_K^\star).
\]
Since $G_{\mathrm{tot}}^{K}$ is strictly decreasing, this implies $t_K^\star\ge t_1^\star$.

Because each $y_j(\cdot)$ is nonincreasing, $G_{\mathrm{oth}}(\cdot)$ is nonincreasing, and hence
$G_{\mathrm{oth}}(t_K^\star)\le G_{\mathrm{oth}}(t_1^\star)$. Using market clearing in each economy,
\[
\sum_{k=1}^K x_k^{\star,K}
=
G_{\mathrm{att}}^{K}(t_K^\star)
=
Q-G_{\mathrm{oth}}(t_K^\star)
\ \ge\
Q-G_{\mathrm{oth}}(t_1^\star)
=
G_{\mathrm{att}}^{1}(t_1^\star)
=
x_A^{\star,1},
\]
which proves the claim.
\end{proof}

\subsection{Mathematical Formalization of Section~\ref{sec:path_indep} and Proof of Theorem~\ref{thm:tc-waterfilling-isolated}}\label{sec:Formalization_TC_and_proof_of_axiamatioc_char}

Recall that at the ADL time $\tau$, account $i$ is described by a short position size $q_i\ge 0$
and an equity level $E_i>0$ evaluated at the execution price $p_\tau$; cf.~\eqref{eq:def_Ei}.
We collect these into the state vector
\[
s=(q_i,E_i)_{i=1}^n \in \mathcal S \defeq (\R_+\times\R_{++})^n.
\]
To avoid overloading the notation $\ell_i(\cdot)$, which we reserve for leverage as a function of the
\emph{allocation} $x_i$, we introduce $\lambda_i(s)$ for leverage as a function of the \emph{state}:
\[
\lambda_i(s)\defeq \frac{p_\tau q_i}{E_i},
\qquad
\mathcal M(s)\defeq \{i:\lambda_i(s)=\max_{1\le j\le n}\lambda_j(s)\}.
\]
Thus $\mathcal M(s)$ is the set of accounts attaining the maximum leverage under the configuration $s$. Given an allocation $x\in\mathcal X$, we can define the corresponding post-ADL state
\[
s(x)\defeq (q_i-x_i,\;E_i)_{i=1}^n \in \mathcal S,
\]
so that the leverage after allocating $x_i$ satisfies
\[
\ell_i(x_i)=\frac{p_\tau(q_i-x_i)}{E_i}=\lambda_i\big(s(x)\big).
\]
We emphasize that $E_i$ is measured at the execution price $p_\tau$, so buying back at $p_\tau$ only reshuffles equity
between entry price and cash margin and does not change $E_i$.

Using this notation, an ADL mechanism can be formally described by a family of maps
\[
\{F_Q:\mathcal S\to\mathcal S\}_{Q\ge 0},
\]
where $F_Q(s)$ is the post-ADL state after a total buyback quantity $Q$ has been allocated across accounts. The following formalizes the natural assumption that these mechanisms can only reduce existing positions by transacting at $p_\tau$ and must clear the quantity $Q$.

\begin{assumption}
\label{ass:market-clearing-F-isolated}
For every $s=(q_i,E_i)_{i=1}^n\in\mathcal S$ and every $Q\in[0,\sum_{i=1}^n q_i]$, write
\begin{equation}\label{eq:sf-state-update-isolated}
F_Q(s)
=
\bigl(q_i-x_i(s,Q),\; E_i\bigr)_{i=1}^n.
\end{equation}
Then the vector $x(s,Q)=(x_i(s,Q))_{i=1}^n$ satisfies
\begin{equation}\label{eq:market-clearing-isolated}
x_i(s,Q)\in[0,q_i] \ \ \text{for all }i \ \text{and}
\ \
\sum_{i=1}^n x_i(s,Q)=Q.
\end{equation}
\end{assumption}

The next two conditions are the mathematical versions of Definitions~\ref{def:path_indep_informal} and~\ref{def:leverage_priority_informal} in the main text.

\begin{assumption}[Path-independence]
\label{ass:time-consist-F-isolated}
For any state $s\in\mathcal S$ and $Q_1,Q_2\ge 0$ with $Q_1+Q_2\le \sum_i q_i$,
\[
F_{Q_1+Q_2}(s)=F_{Q_2}\circ F_{Q_1}(s).
\]
\end{assumption}

\begin{assumption}[Leverage priority]
\label{ass:leverage-priority-F-isolated}
Fix any state $s\in\mathcal S$. Then for sufficiently small $Q>0$, one has $x_i(s,Q)=0$ for all $i\notin\mathcal M(s)$.
\end{assumption}

\subsubsection{Proof of Theorem \ref{thm:tc-waterfilling-isolated}}\label{app:pf.thm.axioms}

Our goal is to show that an ADL mechanism satisfies Assumptions~\ref{ass:market-clearing-F-isolated}--\ref{ass:leverage-priority-F-isolated} if and only if, for every initial state $s\in\mathcal S$ and every $Q\in(0,\sum_i q_i)$, it coincides with the minimax leverage policy.

We first provide a supporting lemma on necessary properties of any ADL mechanism.

\begin{lemma}[Absolute continuity]\label{lem:tc-implies-reg-isolated}
Let Assumptions~\ref{ass:market-clearing-F-isolated} and \ref{ass:time-consist-F-isolated} hold and fix $s\in\mathcal S$. For each $i$ the trajectory $Q\mapsto x_i(s,Q)$ is nonnegative, nondecreasing,
and $1$-Lipschitz on $[0,\sum_i q_i]$. In particular, it is absolutely continuous in $Q$, and its a.e.\ derivative
$\frac{\partial}{\partial Q}x_i(s,Q)$ satisfies
\[
0\le \frac{\partial}{\partial Q}x_i(s,Q)\le 1\quad\text{for a.e.\ }Q,
\qquad
\sum_{i=1}^n \frac{\partial}{\partial Q}x_i(s,Q)=1\quad\text{for a.e.\ }Q.
\]
\end{lemma}

\begin{proof}
Let $s\in\mathcal S$ and $0\le Q_1\le Q_2\le \sum_i q_i$. Set $\Delta Q\defeq Q_2-Q_1$.
Path-independence gives
\[
F_{Q_2}(s)=F_{\Delta Q}(F_{Q_1}(s)).
\]
Writing both sides using \eqref{eq:sf-state-update-isolated} yields the identity
\begin{equation}\label{eq:cocycle-isolated}
x_i(s,Q_2)=x_i(s,Q_1)+x_i(F_{Q_1}(s),\Delta Q),\qquad i=1,\dots,n.
\end{equation}
By \eqref{eq:market-clearing-isolated} applied at the intermediate state $F_{Q_1}(s)$, we have
$x_i(F_{Q_1}(s),\Delta Q)\ge 0$. Moreover, since $\sum_i x_i(F_{Q_1}(s),\Delta Q)=\Delta Q$, we also have
$x_i(F_{Q_1}(s),\Delta Q)\le \Delta Q$. Substituting into \eqref{eq:cocycle-isolated} gives
\[
0\le x_i(s,Q_2)-x_i(s,Q_1)\le \Delta Q = Q_2-Q_1.
\]
Thus $Q\mapsto x_i(s,Q)$ is nondecreasing and $1$-Lipschitz. The derivative properties follow from
standard facts for absolutely continuous functions and by differentiating $\sum_i x_i(s,Q)=Q$ a.e.
\end{proof}

We can now show the main result.

\begin{proof}[Proof of Theorem~\ref{thm:tc-waterfilling-isolated}]
We prove the two implications separately.

\medskip
\noindent\emph{The water-filling rule satisfies the axioms.}
Fix a state $s=(q_i,E_i)_{i=1}^n\in\mathcal S$ and define
\[
G_s(t)
\defeq
\sum_{i=1}^n
\Bigl(q_i-\frac{E_i}{p_\tau}t\Bigr)_+,
\qquad t\ge 0.
\]
The function $G_s$ is continuous and strictly decreasing on
$[0,\lambda_{\max}(s)]$, where
\[
\lambda_{\max}(s)
\defeq
\max_{1\le i\le n}\frac{p_\tau q_i}{E_i},
\]
with
\[
G_s(0)=\sum_{i=1}^n q_i,
\qquad
G_s(\lambda_{\max}(s))=0.
\]
For $Q\in(0,\sum_i q_i)$, let $t_s(Q)$ be the unique threshold satisfying
\[
G_s(t_s(Q))=Q.
\]
At the boundary points, set
\[
t_s(0)\defeq \lambda_{\max}(s),
\qquad
t_s\Bigl(\sum_i q_i\Bigr)\defeq 0.
\]
Define
\[
x_i^{\mathrm{WF}}(s,Q)
\defeq
\Bigl(q_i-\frac{E_i}{p_\tau}t_s(Q)\Bigr)_+
\]
and
\[
F_Q^{\mathrm{WF}}(s)
\defeq
\bigl(q_i-x_i^{\mathrm{WF}}(s,Q),\,E_i\bigr)_{i=1}^n.
\]

Market clearing holds immediately, since
\[
0\le x_i^{\mathrm{WF}}(s,Q)\le q_i
\]
and
\[
\sum_{i=1}^n x_i^{\mathrm{WF}}(s,Q)
=
G_s(t_s(Q))
=
Q.
\]

We next verify leverage priority. Let
\[
\mathcal M(s)
\defeq
\{i:\lambda_i(s)=\lambda_{\max}(s)\}.
\]
If $\mathcal M(s)=\{1,\dots,n\}$, the claim is immediate. Otherwise define
\[
\lambda_{\mathrm{sub}}(s)
\defeq
\max_{i\notin\mathcal M(s)}\lambda_i(s)
<
\lambda_{\max}(s)
\]
and set
\[
\delta
\defeq
G_s(\lambda_{\mathrm{sub}}(s))>0.
\]
If $0<Q<\delta$, then
\[
G_s(t_s(Q))
=
Q
<
G_s(\lambda_{\mathrm{sub}}(s)).
\]
Since $G_s$ is strictly decreasing,
\[
t_s(Q)>\lambda_{\mathrm{sub}}(s).
\]
Hence, for every $i\notin\mathcal M(s)$,
\[
\lambda_i(s)\le \lambda_{\mathrm{sub}}(s)<t_s(Q),
\]
and therefore
\[
x_i^{\mathrm{WF}}(s,Q)=0.
\]
Thus, for sufficiently small quantities, ADL is imposed only on maximally
leveraged accounts.

Finally, we verify path independence. Fix $Q_1,Q_2\ge 0$ such that
\[
Q_1+Q_2\le \sum_{i=1}^n q_i,
\]
and write
\[
t_1\defeq t_s(Q_1),
\qquad
t_{12}\defeq t_s(Q_1+Q_2).
\]
Since $G_s$ is decreasing and $Q_1+Q_2\ge Q_1$, we have
\[
t_{12}\le t_1.
\]
After the first step, the residual position of account $i$ is
\[
q_i^{(1)}
=
q_i-x_i^{\mathrm{WF}}(s,Q_1)
=
\min\Bigl\{q_i,\frac{E_i}{p_\tau}t_1\Bigr\}.
\]
Since $t_{12}\le t_1$, for each $i$ we have
\[
\Bigl(q_i^{(1)}-\frac{E_i}{p_\tau}t_{12}\Bigr)_+
=
\Bigl(q_i-\frac{E_i}{p_\tau}t_{12}\Bigr)_+
- 
\Bigl(q_i-\frac{E_i}{p_\tau}t_1\Bigr)_+.
\]
The restarted demand function therefore satisfies
\begin{align*}
G_{F_{Q_1}^{\mathrm{WF}}(s)}(t_{12})
&=
\sum_{i=1}^n
\Bigl(q_i^{(1)}-\frac{E_i}{p_\tau}t_{12}\Bigr)_+ \\
&=
\sum_{i=1}^n
\left[
\Bigl(q_i-\frac{E_i}{p_\tau}t_{12}\Bigr)_+
-
\Bigl(q_i-\frac{E_i}{p_\tau}t_1\Bigr)_+
\right] \\
&=
G_s(t_{12})-G_s(t_1) \\
&=
(Q_1+Q_2)-Q_1 \\
&=
Q_2.
\end{align*}
By uniqueness of the restarted threshold, together with the boundary
conventions above, the second-step threshold is $t_{12}$. The residual position after
both steps is therefore
\begin{align*}
\min\Bigl\{q_i^{(1)},\frac{E_i}{p_\tau}t_{12}\Bigr\}
&=
\min\Bigl\{
\min\Bigl\{q_i,\frac{E_i}{p_\tau}t_1\Bigr\},
\frac{E_i}{p_\tau}t_{12}
\Bigr\} \\
&=
\min\Bigl\{q_i,\frac{E_i}{p_\tau}t_{12}\Bigr\},
\end{align*}
where the last equality follows from $t_{12}\le t_1$. This is exactly
the residual position obtained by applying the water-filling rule directly
with quantity $Q_1+Q_2$. Therefore,
\[
F_{Q_2}^{\mathrm{WF}}
\bigl(F_{Q_1}^{\mathrm{WF}}(s)\bigr)
=
F_{Q_1+Q_2}^{\mathrm{WF}}(s).
\]
Thus the water-filling mechanism satisfies
Assumptions~\ref{ass:market-clearing-F-isolated}--\ref{ass:leverage-priority-F-isolated}.

\medskip
\noindent\emph{Any mechanism satisfying the axioms is the water-filling rule.}
Conversely, suppose that an ADL mechanism satisfies
Assumptions~\ref{ass:market-clearing-F-isolated}--\ref{ass:leverage-priority-F-isolated}.

Fix an initial state $s=(q_i,E_i)_{i=1}^n\in\mathcal S$ and a target quantity
$Q_1\in(0,\sum_i q_i)$. Define the maximum leverage at quantity $Q\in[0,Q_1]$ by
\[
t(Q)\defeq \max_{1\le i\le n}\lambda_i(F_{Q}(s))
=\max_{1\le i\le n}\frac{p_\tau\,(q_i-x_i(s,Q))}{E_i}.
\]
By Lemma~\ref{lem:tc-implies-reg-isolated}, each $x_i(s, \cdot)$ is continuous and nondecreasing. Hence $Q\mapsto \lambda_i(F_{Q}(s))$ is continuous and nonincreasing, and therefore $t(\cdot)$ is continuous and
nonincreasing.

\medskip
\noindent\emph{Claim 1.} If $\lambda_i(s)<t(Q_1)$, then $x_i(s,Q_1)=0$.

\smallskip
\noindent
Assume for contradiction that $x_i(s,Q_1)>0$ and let
\[
Q_0\defeq\inf\{Q\in[0,Q_1]:x_i(s,Q)>0\}.
\]
Continuity gives $x_i(s,Q_0)=0$, hence $\lambda_i(F_{Q_0}(s))=\lambda_i(s)<t(Q_1)\le t(Q_0)$, so
$i\notin\mathcal M(F_{Q_0}(s))$. By leverage-priority at $F_{Q_0}(s)$ there exists $\varepsilon>0$ such that
\[
x_i(F_{Q_0}(s),h)=0\qquad \forall h\in[0,\varepsilon].
\]
By path-independence, for all $h\in[0,\varepsilon]$,
\[
x_i(s,Q_0+h)=x_i(s,Q_0)+x_i(F_{Q_0}(s),h)=0,
\]
contradicting the definition of $Q_0$. Hence $x_i(s,Q_1)=0$.

\medskip
\noindent\emph{Claim 2.} If $\lambda_i(s)\ge t(Q_1)$, then $\lambda_i(F_{Q_1}(s))=t(Q_1)$.

\smallskip
\noindent
We always have $\lambda_i(F_{Q_1}(s))\le t(Q_1)$. Suppose for contradiction that
$\lambda_i(F_{Q_1}(s))<t(Q_1)$. Set
\[
\varepsilon\defeq t(Q_1)-\lambda_i(F_{Q_1}(s))>0.
\]
Since $\lambda_i(s)\ge t(Q_1)>t(Q_1)-\varepsilon/2$ and $Q\mapsto\lambda_i(F_{Q}(s))$ is continuous and
nonincreasing, the intermediate value theorem yields some $Q_*<Q_1$ such that
\[
\lambda_i(F_{Q_*}(s))=t(Q_1)-\frac{\varepsilon}{2}.
\]
By path-independence, $F_{Q_1}(s)=F_{Q_1-Q_*}(F_{Q_*}(s))$. Define the restarted maximum leverage
\[
\tilde t(h)\defeq \max_{1\le j\le n}\lambda_j(F_h(F_{Q_*}(s))),\qquad h\ge 0,
\]
so that $\tilde t(Q_1-Q_*)=t(Q_1)$. Since $\lambda_i(F_{Q_*}(s))<\tilde t(Q_1-Q_*)$, Claim~1 applied to the
initial state $F_{Q_*}(s)$ and terminal quantity $Q_1-Q_*$ implies $x_i(F_{Q_*}(s),Q_1-Q_*)=0$. Since the proof of Claim~1 applies to any initial state, this invocation is valid. Using path-independence in the form
$x_i(s,Q_1)=x_i(s,Q_*)+x_i(F_{Q_*}(s),Q_1-Q_*)$ (see \eqref{eq:cocycle-isolated} in Lemma~\ref{lem:tc-implies-reg-isolated}), we obtain $x_i(s,Q_1)=x_i(s,Q_*)$, hence
$\lambda_i(F_{Q_1}(s))=\lambda_i(F_{Q_*}(s))=t(Q_1)-\varepsilon/2>\lambda_i(F_{Q_1}(s))$,
a contradiction. We conclude $\lambda_i(F_{Q_1}(s))=t(Q_1)$.

Combining Claims~1--2 yields, for every $i$,
\[
\lambda_i(F_{Q_1}(s))=\min\{\lambda_i(s),\,t(Q_1)\}.
\]
Since $\lambda_i(F_{Q_1}(s))=\frac{p_\tau(q_i-x_i(s,Q_1))}{E_i}$ and $\lambda_i(s)=\frac{p_\tau q_i}{E_i}$,
this is equivalent to
\[
x_i(s,Q_1)=\Big(q_i-\frac{E_i}{p_\tau}\,t(Q_1)\Big)_+.
\]
Summing over $i$ and using market clearing $\sum_i x_i(s,Q_1)=Q_1$ gives
\[
Q_1=\sum_{i=1}^n\Big(q_i-\frac{E_i}{p_\tau}\,t(Q_1)\Big)_+.
\]
By the strict decrease of $t\mapsto \sum_i\big(q_i-\frac{E_i}{p_\tau}t\big)_+$ on
$[0,\max_i\lambda_i(s)]$ (cf.~Lemma~\ref{lem:G_properties_unique_level}), $t(Q_1)$ is the unique solution to the equation
\[
Q_1=\sum_{i=1}^n\Big(q_i-\frac{E_i}{p_\tau}\,t\Big)_+.
\]
Hence the induced allocation
$x_i(s,Q_1)=\big(q_i-\frac{E_i}{p_\tau}t(Q_1)\big)_+$ coincides with the water-filling allocation of Theorem~\ref{thm:risk_neutral}.
\end{proof}

\subsection{Connection to the Claims Literature}\label{app:claims_correspondence}

There is an exact connection between the claims framework of \citet{Thomson2003,Thomson2015} and the water-filling rule in the single-asset isolated-margin model. If $x_i$ is the buyback imposed on account~$i$ and $r_i \defeq q_i - x_i$ is the residual position left open after ADL, then the ADL feasibility constraints become
\[
0 \leq r_i \leq q_i \qquad\text{and}\qquad \sum_i r_i = \sum_i q_i - Q.
\]
Thus the residual positions form a claims problem with claims $q_i$ and endowment $\sum_i q_i - Q$. The water-filling allocation of Theorem~\ref{thm:risk_neutral} can be written in terms of residual positions as
\[
r_i^\star = \min\!\Big\{q_i,\,\frac{E_i}{p_\tau}\,t^\star\Big\},
\]
which is precisely the \emph{weighted constrained equal awards rule} for the residual positions, with weights proportional to $E_i/p_\tau$. Equivalently, in the buyback variables $x_i = q_i - r_i$, the rule is the dual weighted constrained equal losses rule. In the risk-neutral case, our expected-loss objective selects this allocation uniquely; for CVaR objectives, the same water-filling allocation remains a canonical optimizer, although the optimizer need not be unique (cf.\ Theorem~\ref{thm:waterfilling_nonunique}). For spectral and more general monotone risk objectives, water-filling optimality is established in Proposition~\ref{prop:wf.optimal.monotone}.

The interpretation, however, is different from the classical claims model. In the claims literature, weights are typically primitive priority or entitlement parameters. In our setting, the weights $E_i/p_\tau$ arise endogenously from account equity and from the bankruptcy thresholds induced by future price moves. Thus our contribution in the isolated-margin case is not the abstract water-filling formula itself, but the derivation of the particular weighted claims rule from an exchange shortfall-risk minimization problem.

Several properties of the minimax leverage policy also have counterparts in the claims literature. The path-independence property of Section~\ref{sec:path_indep} is analogous in spirit to composition properties studied in the claims framework, and our Sybil resistance requirement is related to the axiom of no advantageous splitting. However, the ADL axioms operate on account states $(q_i,E_i)$ and are tied to leverage and future default risk, rather than to claims alone.

\subsection{Proof of Theorem~\ref{thm:waterfilling_nonunique}}
\label{app:cvar_degeneracy}

We break the verification of the theorem claims into two parts. The first part shows that the water-filling strategy $x^{\mathrm{WF}}$ is optimal for \eqref{eq:adl_primal_cvar}. The second shows that the optimizer need not be unique in general. The proof of the theorem follows directly from Lemma~\ref{lem:wf.optimal} and Proposition~\ref{prop:cvar_nonunique}.

\subsubsection{Optimality of $x^{\mathrm{WF}}$}

Before we begin, we provide a useful representation for the account shortfall that will be used repeatedly.

\begin{lemma}\label{lem:sigma_explicit}
Fix $i$ and $x_i\in[0,q_i]$. For any price $p\in\R$,
\[
\sigma_i(x_i,p)=
(q_i-x_i)\bigl(p-p_i^{(z)}(x_i)\bigr)_+.
\]
\end{lemma}

\begin{proof}
Fix $x_i\in[0,q_i]$ and $p\in\R$. Starting from
\[
e_i(x_i,p)=q_i\bigl(p_i^{(e)}-p\bigr)-x_i\bigl(p_\tau-p\bigr)+m_i,
\]
we expand and regroup terms:
\begin{align*}
e_i(x_i,p)
&=q_i p_i^{(e)}-q_i p -x_i p_\tau + x_i p + m_i\\
&=\bigl(q_i p_i^{(e)}+m_i-x_i p_\tau\bigr) - (q_i-x_i)p.
\end{align*}
Hence
\[
-e_i(x_i,p)=(q_i-x_i)p-\bigl(q_i p_i^{(e)}+m_i-x_i p_\tau\bigr).
\]
If $x_i<q_i$, define
\[
p_i^{(z)}(x_i)\defeq \frac{q_i p_i^{(e)}+m_i-x_i p_\tau}{q_i-x_i}.
\]
Then
\[
-e_i(x_i,p)=(q_i-x_i)\bigl(p-p_i^{(z)}(x_i)\bigr),
\]
and therefore
\[
\sigma_i(x_i,p)=\big(-e_i(x_i,p)\big)_+
=(q_i-x_i)\bigl(p-p_i^{(z)}(x_i)\bigr)_+.
\]
If $x_i=q_i$, then $e_i(q_i,p)=q_i(p_i^{(e)}-p_\tau)+m_i=E_i>0$ is constant in $p$, so
$\sigma_i(q_i,p)=0$. With the convention $p_i^{(z)}(q_i)=+\infty$, the formula still holds since
$(p-\infty)_+=0$.
\end{proof}

We can now derive necessary and sufficient conditions for optimality.

\begin{lemma}\label{lem:cvar_kkt_reduced}
Define, for $x_i\in[0,q_i]$,
\begin{equation*}
\bar p_i(x_i)\defeq \max\{p_\beta,\;p_i^{(z)}(x_i)\},
\qquad
h_\beta(x_i)\defeq \E\!\left[(p_T-p_\tau)\,\I{p_T\ge \bar p_i(x_i)}\right].
\end{equation*}
Under Assumptions~\ref{ass:feasible}--\ref{ass:p_T_regular}, a point $x^\star\in\mathcal X$ solves
\eqref{eq:adl_primal_cvar} if and only if there exists $\theta^\star\in\R$ such that for each $i$,
\[
\begin{aligned}
&0<x_i^\star<q_i \ \Longrightarrow\ h_\beta(x_i^\star)=\theta^\star,\\
&x_i^\star=0 \ \Longrightarrow\ h_\beta(0)\le \theta^\star,\\
&x_i^\star=q_i \ \Longrightarrow\ h_\beta(q_i)\ge \theta^\star.
\end{aligned}
\]
Moreover, feasibility and the assumption $0<Q<\sum_{i=1}^n q_i$ imply that $\theta^\star\ge 0$.
Indeed, there exists an index $j$ such that $x_j^\star<q_j$.
If $0<x_j^\star<q_j$, then
\[
\theta^\star=h_\beta(x_j^\star)\ge 0.
\]
If $x_j^\star=0$, then
\[
h_\beta(0)\le \theta^\star,
\]
and hence $\theta^\star\ge 0$ since $h_\beta(0)\ge 0$.
Consequently, if $x_i^\star=q_i$ for some $i$, then
\[
0=h_\beta(q_i)\ge \theta^\star\ge 0,
\]
so $\theta^\star=0$.
\end{lemma}

\begin{proof}
In what follows we will exploit the separable form of the objective that was derived from the comonotonicity of the individual account shortfalls. The Rockafellar--Uryasev representation \citep{rockafellar2000cvar} for CVaR gives that
\[
\sum_{i=1}^n\mathrm{CVaR}_\beta\big(\sigma_i(x_i, p_T)\big)
=
\sum_{i=1}^n f_i(x_i),
\]
where
\[
f_i(x_i)
\defeq
\inf_{\alpha_i\in\R}
\left\{
\alpha_i
+
\frac{1}{1-\beta}\,
\E\!\left[\big(\sigma_i(x_i,p_T)-\alpha_i\big)_+\right]
\right\}.
\]

Letting $\alpha=(\alpha_1,\dots,\alpha_n)$ we can write the CVaR optimization equivalently as:
\begin{equation*}
\begin{array}{ll}
\minimize_{(x,\alpha)\in\R^{2n}} & \tilde{V}(x,\alpha) \defeq \sum_{i=1}^n\left\{
\alpha_i
+
\frac{1}{1-\beta}\,
\E\!\left[\big(\sigma_i(x_i,p_T)-\alpha_i\big)_+\right]
\right\} \\
\subjectto & x\in\mathcal X.
\end{array}
\end{equation*}
For each $i$, the map $(x_i,\alpha_i)\mapsto \alpha_i + \frac{1}{1-\beta}\E[(\sigma_i(x_i,p_T)-\alpha_i)_+]$
is (jointly) convex in $(x_i,\alpha_i)$. So, this ``lifted'' representation is a convex problem.  Moreover, partial minimization over $\alpha_i$ preserves convexity, so $f_i$ (and the original CVaR optimization problem) is also convex. Finally, as in Appendix~\ref{app:rn.thm}, under Assumption~\ref{ass:feasible} Slater's condition is trivially satisfied and so the KKT conditions are both necessary and sufficient for optimality.

As before, introduce multipliers $\lambda\in\R$ for $\sum_i x_i=Q$,
$\nu_i\ge 0$ for $-x_i\le 0$, and $\mu_i\ge 0$ for $x_i-q_i\le 0$.
An optimal primal--dual quintuple $(x^\star,\alpha^\star,\lambda^\star,\nu^\star,\mu^\star)$ satisfies
\begin{align}
&0\in \partial_{x_i}\tilde V(x^\star,\alpha^\star)+\lambda^\star-\nu_i^\star+\mu_i^\star,
\quad i=1,\dots,n, \nonumber\\
&0\in \partial_{\alpha_i}\tilde V(x^\star,\alpha^\star),
\quad i=1,\dots,n, \label{eq:kkt_cvar_stat2}\\
&x^\star\in\mathcal X,\nonumber\\
&\nu_i^\star\ge 0,\ \mu_i^\star\ge 0,\quad i=1,\dots,n, \nonumber\\
&\nu_i^\star x_i^\star=0,\quad \mu_i^\star(x_i^\star-q_i)=0,\quad i=1,\dots,n. \nonumber
\end{align}

\paragraph{The $\alpha_i$-stationarity condition.}
Fix $i$ and abbreviate $Z_i \defeq \sigma_i(x_i,p_T)$.
For each realization of $p_T$, the map $\alpha_i\mapsto (Z_i-\alpha_i)_+$ is convex and its
subdifferential is
\[
\partial_{\alpha_i}(Z_i-\alpha_i)_+
=
\begin{cases}
\{-1\}, & Z_i>\alpha_i,\\[2pt]
[-1,0], & Z_i=\alpha_i,\\[2pt]
\{0\}, & Z_i<\alpha_i.
\end{cases}
\]
Since $\tilde V$ is a sum across $i$, one can compute $\partial_{\alpha_i}\tilde V$ by a direct calculation of the one-sided difference quotients as
\[
\partial_{\alpha_i}\tilde V(x,\alpha)
=
\bigg[\,1-\frac{1}{1-\beta}\PR(Z_i\ge \alpha_i),\ \
1-\frac{1}{1-\beta}\PR(Z_i> \alpha_i)\,\bigg].
\]
where $[a,b]$ denotes the closed interval.

Thus using the stationarity condition \eqref{eq:kkt_cvar_stat2} and selecting the extreme points of the subgradient interval yields the pair of inequalities
\[
1-\frac{1}{1-\beta}\,\PR(Z_i>\alpha_i)\ \ge\ 0
\ \ge\
1-\frac{1}{1-\beta}\,\PR(Z_i\ge \alpha_i).
\]
Rearranging gives the quantile condition
\begin{equation*}
\PR(Z_i\le \alpha_i)\ \ge\ \beta
\ \ge\
\PR(Z_i<\alpha_i).
\end{equation*}
That is, any optimal $\alpha_i^\star$ is a $\beta$-quantile of $\sigma_i(x_i^\star,p_T)$. But, for each fixed $x_i$, the map $p\mapsto \sigma_i(x_i,p)$ is continuous and nondecreasing, and (under Assumption~\ref{ass:p_T_regular})
$\,p_T$ is atomless. Then, under the standing assumptions, the $\beta$-quantile of
$\sigma_i(x_i,p_T)$ satisfies (see \cite[Theorem 3.1]{hanbali2022monotone})
\begin{equation}\label{eqn:quantile.sigma}
\mathrm{VaR}_\beta\big(\sigma_i(x_i,p_T)\big)
= \sigma_i(x_i,p_\beta).
\end{equation}
Hence, we may take $\alpha_i^\star = \sigma_i(x_i^\star,p_\beta)$.

\paragraph{The $x_i$-stationarity condition.}
Fix $i$ and notice that since $\sigma_i(x_i,p)\ge 0$ for all $(x_i,p)$, the minimization in $\alpha_i$ can be restricted to
$\alpha_i\ge 0$ without loss of generality (since $\alpha_i^\star = \sigma_i(x_i^\star,p_\beta)$).

For $\alpha_i\ge 0$ we have the identity
\begin{equation*}
(\sigma_i(x_i,p_T)-\alpha_i)_+
=
(-e_i(x_i,p_T)-\alpha_i)_+,
\end{equation*}
because $\sigma_i=0$ on $\{e_i>0\}$ and $\sigma_i=-e_i$ on $\{e_i\le 0\}$.

For each fixed $p$, the map $x_i\mapsto -e_i(x_i,p)-\alpha_i$ is affine with slope $p_\tau-p$ for $x_i<q_i$.
Hence, for each $(x_i,\alpha_i)$, a subgradient of $x_i\mapsto (-e_i(x_i,p)-\alpha_i)_+$ is given by
\[
(p_\tau-p)\cdot  \I{-e_i(x_i,p)>\alpha_i},
\]
with the usual interval of subgradients at the kink $\{-e_i(x_i,p)=\alpha_i\}$.
Under Assumption~\ref{ass:p_T_regular}, the kink event has probability zero for $x_i<q_i$,  hence for $x_i<q_i$
\begin{equation}\label{eq:subgrad.1}
\partial_{x_i}\tilde V(x^\star,\alpha^\star)
=
\left\{\frac{1}{1-\beta}\E\Big[(p_\tau-p_T)\,\I{-e_i(x_i^\star,p_T)>\alpha_i^\star}\Big]\right\}.
\end{equation}

Define the threshold price
\[
\bar p_i(x_i,\alpha_i)
\defeq
\frac{q_i p_i^{(e)} + m_i + \alpha_i - x_i p_\tau}{q_i-x_i},
\qquad x_i<q_i.
\]
A rearrangement yields
\[
-e_i(x_i,p)-\alpha_i
=
(q_i-x_i)\bigl(p-\bar p_i(x_i,\alpha_i)\bigr),
\]
and since $q_i-x_i>0$,
\[
\{-e_i(x_i,p)>\alpha_i\}
=
\{p>\bar p_i(x_i,\alpha_i)\}.
\]
Substituting into \eqref{eq:subgrad.1} gives
\begin{equation*}
\partial_{x_i}\tilde V(x^\star,\alpha^\star)
=
\left\{\frac{1}{1-\beta}\,
\E\!\left[(p_\tau-p_T)\,
 \I{p_T>\bar p_i(x_i^\star,\alpha_i^\star)}\right]
\right\}.
\end{equation*}

Finally, note that
\[
p_i^{(z)}(x_i)=\frac{q_i p_i^{(e)} + m_i - x_i p_\tau}{q_i-x_i}
\qquad\Rightarrow\qquad
\bar p_i(x_i,\alpha_i)=p_i^{(z)}(x_i)+\frac{\alpha_i}{q_i-x_i}.
\]
Moreover, since $p\mapsto \sigma_i(x_i,p)=(q_i-x_i)\,(p-p_i^{(z)}(x_i))_+$ is continuous and nondecreasing, the $\beta$-quantile satisfies (see \eqref{eqn:quantile.sigma} and Lemma~\ref{lem:sigma_explicit})
\[
\alpha_i^\star=\mathrm{VaR}_\beta\big(\sigma_i(x_i^\star,p_T)\big)
=\sigma_i(x_i^\star,p_\beta)
=(q_i-x_i^\star)\bigl(p_\beta-p_i^{(z)}(x_i^\star)\bigr)_+.
\]
Therefore
\[
\bar p_i(x_i^\star,\alpha_i^\star)
=
p_i^{(z)}(x_i^\star)+\bigl(p_\beta-p_i^{(z)}(x_i^\star)\bigr)_+
=
\max\{p_\beta,\;p_i^{(z)}(x_i^\star)\}.
\]
As a result, by comparing with the definition in the statement of the Lemma, we see that $\bar p_i(x_i^\star)=\bar p_i(x_i^\star,\alpha_i^\star)$ by the calculation above. We obtain, for $x_i^\star<q_i$,
\[
\partial_{x_i}\tilde V(x^\star,\alpha^\star)
=
\left\{-\frac{1}{1-\beta}\,h_\beta(x_i^\star)\right\}.
\]

\begin{remark}\label{rem:convex_onesided_subgrad}
In what follows we will make use of the following standard facts for convex functions. Let $g:[a,b]\to\mathbb R$ be convex.
\begin{enumerate}
\item For every $x\in(a,b)$ the one-sided derivatives $g'_-(x)$ and $g'_+(x)$ exist, are finite, and satisfy
$g'_-(x)\le g'_+(x)$. Moreover $g'_-$ is left-continuous and $g'_+$ is right-continuous on $(a,b)$.
\item At the endpoints,
\[
\partial g(a)=(-\infty,\,g'_+(a)],\qquad \partial g(b)=[\,g'_-(b),\,\infty),
\]
(with the convention that $( -\infty,-\infty]=\varnothing$ and $[+\infty,\infty)=\varnothing$) so that $\sup\partial g(a)=g'_+(a)$ and $\inf\partial g(b)=g'_-(b)$. If $g$ is differentiable on $(a,b)$, then
$g'_+(a)=\lim_{x\downarrow a}g'(x)$ and $g'_-(b)=\lim_{x\uparrow b}g'(x)$.
\end{enumerate}
\end{remark}

\paragraph{The endpoint $x_i=q_i$.}
Suppose that $x_i^\star=q_i$. Then
\[
\alpha_i^\star
=
\sigma_i(q_i,p_\beta)
=
0.
\]
Since
\[
p_i^{(z)}(x_i)
=
\frac{q_i p_i^{(e)}+m_i-x_i p_\tau}{q_i-x_i}
=
p_\tau+\frac{E_i}{q_i-x_i},
\]
we have
\[
p_i^{(z)}(x_i)\longrightarrow+\infty
\qquad\text{as }x_i\uparrow q_i.
\]
Hence, for $x_i<q_i$ sufficiently close to $q_i$,
\[
p_i^{(z)}(x_i)\ge p_\beta,
\]
and therefore
\[
\bar p_i(x_i,\alpha_i^\star)
=
p_i^{(z)}(x_i)+\frac{\alpha_i^\star}{q_i-x_i}
=
p_i^{(z)}(x_i).
\]
Using the derivative formula derived above and the atomlessness of $p_T$,
we obtain, for $x_i<q_i$ sufficiently close to $q_i$,
\[
\frac{\partial}{\partial x_i}
\tilde V(x_1^\star,\dots,x_{i-1}^\star,x_i,
x_{i+1}^\star,\dots,x_n^\star,\alpha^\star)
=
-\frac{1}{1-\beta}h_\beta(x_i).
\]
Moreover,
\[
h_\beta(x_i)
=
\E\!\left[
(p_T-p_\tau)
\I{p_T\ge p_i^{(z)}(x_i)}
\right]
\longrightarrow 0
\qquad\text{as }x_i\uparrow q_i,
\]
by dominated convergence, since
$p_i^{(z)}(x_i)\to+\infty$ and the integrand is eventually dominated by
$(p_T-p_\tau)_+$.
It follows that the left derivative of the coordinate objective
at $q_i$ equals $0$. By convexity,
\[
\partial_{x_i}\tilde V(x^\star,\alpha^\star)
\subset[0,\infty)
\]
and
\[
0\in\partial_{x_i}\tilde V(x^\star,\alpha^\star)
\]
at $x_i^\star=q_i$.

\paragraph{Conclusion.}
Define $\theta^\star\defeq (1-\beta)\lambda^\star$.

\emph{(Only if).}
Assume $x^\star$ is optimal. Consider any index $i$.

\underline{Interior: $0<x_i^\star<q_i$.}
Then $\nu_i^\star=\mu_i^\star=0$ and stationarity gives
\[
0=-\frac{1}{1-\beta}h_\beta(x_i^\star)+\lambda^\star
\quad\Longleftrightarrow\quad
h_\beta(x_i^\star)=\theta^\star.
\]

\underline{Lower bound: $x_i^\star=0$.}
Then $\mu_i^\star=0$ and stationarity yields
\[
0\in \partial_{x_i}\tilde V(x^\star,\alpha^\star)+\lambda^\star-\nu_i^\star.
\]
Since $\partial_{x_i}\tilde V(x^\star,\alpha^\star)\subset(-\infty,\, -\frac{1}{1-\beta}h_\beta(0)]$ at $x_i^\star=0$
(the maximal subgradient is the right derivative), existence of $\nu_i^\star\ge 0$ forces
\[
-\frac{1}{1-\beta}h_\beta(0)+\lambda^\star\ge 0
\quad\Longleftrightarrow\quad
h_\beta(0)\le \theta^\star.
\]

\underline{Upper bound: $x_i^\star=q_i$.}
Then $\nu_i^\star=0$ and stationarity gives
\[
0\in \partial_{x_i}\tilde V(x^\star,\alpha^\star)+\lambda^\star+\mu_i^\star,
\qquad \mu_i^\star\ge 0.
\]
As shown above,
\[
\partial_{x_i}\tilde V(x^\star,\alpha^\star)
\subset[0,\infty)
\]
at $x_i^\star=q_i$. Hence stationarity requires
\[
\lambda^\star\le 0,
\]
or equivalently
\[
\theta^\star\le 0=h_\beta(q_i).
\]
Thus
\[
h_\beta(q_i)\ge\theta^\star.
\]

\emph{(If).}
Conversely, suppose $x^\star\in\mathcal X$ and there exists $\theta^\star$ satisfying the three displayed
implications. Set $\lambda^\star\defeq \theta^\star/(1-\beta)$ and define $\alpha_i^\star\defeq \sigma_i(x_i^\star,p_\beta)$.

By the preceding argument, $\theta^\star\ge 0$.

We construct $\nu^\star,\mu^\star$ coordinatewise to satisfy KKT.

\underline{If $0<x_i^\star<q_i$:}
set $\nu_i^\star=\mu_i^\star=0$. Then stationarity holds because $-\frac{1}{1-\beta}h_\beta(x_i^\star)+\lambda^\star=0$ due to 
$h_\beta(x_i^\star)=\theta^\star$.

\underline{If $x_i^\star=0$:}
set $\mu_i^\star=0$ and choose the maximal subgradient
$g_i^\star\defeq -\frac{1}{1-\beta}h_\beta(0)\in \partial_{x_i}\tilde V(x^\star,\alpha^\star)$.
Define
\[
\nu_i^\star\defeq g_i^\star+\lambda^\star
=\frac{\theta^\star-h_\beta(0)}{1-\beta}\ge 0,
\]
where the inequality uses $h_\beta(0)\le \theta^\star$. Then stationarity holds:
$0=g_i^\star+\lambda^\star-\nu_i^\star$.

\underline{If $x_i^\star=q_i$:}
the condition $h_\beta(q_i)\ge \theta^\star$ gives
\[
0=h_\beta(q_i)\ge \theta^\star.
\]
Together with $\theta^\star\ge 0$, this implies $\theta^\star=0$, and
hence $\lambda^\star=0$. Choose $g_i^\star=0\in \partial_{x_i}\tilde V(x^\star,\alpha^\star)$ and set
$\nu_i^\star=0$, $\mu_i^\star=0$. Then stationarity holds.

Finally, primal feasibility holds by assumption ($x^\star\in\mathcal X$), dual feasibility holds by construction,
and complementary slackness holds by how we set $\nu_i^\star,\mu_i^\star$ in each case.
Therefore the KKT conditions are satisfied, and since the problem is convex with Slater's condition,
$x^\star$ is optimal.
\end{proof}

By interpreting these conditions in terms of the account leverage, we can get a result that is reminiscent of the original equalization characterization from Proposition \ref{prop:equalization}. Observe that
\[\{p_T>\bar p_i(x_i^\star)\} =\{p_T>\max\{p_\beta,\,p_i^{(z)}(x_i^\star)\}\}=\{p_T- p_\tau>\max\{p_\beta-p_\tau,\,p_\tau\,\ell_i(x_i^\star)^{-1}\}\}.\]
As in Section \ref{subsec:cvar}, write $\ell_\beta = \left(\frac{p_\beta-p_\tau}{p_\tau}\right)^{-1}$ with the convention that $\ell_\beta =+\infty$ if $p_\beta \leq p_\tau$. Then, we can define the function $v_\beta(\cdot)$
\[v_\beta(\ell)= \E\!\left[(p_T-p_{\tau})\,
\I{p_T- p_\tau>p_\tau \max\{\ell_\beta^{-1},\,\ell^{-1}\}}\right].\]
We also use the conventions
\[
(+\infty)^{-1}=0,
\qquad
0^{-1}=+\infty.
\]
Since $p_T$ is atomless under Assumption~\ref{ass:p_T_regular}, the strict and weak inequalities are interchangeable in these tail events.
Since the threshold
$\max(p_\beta,p_\tau(1+\ell^{-1}))$ is nonincreasing in $\ell$, the map $\ell\mapsto v_\beta(\ell)$
is nondecreasing.
Moreover, $v_\beta(\ell)\ge 0$ for all $\ell$, because the integrand is nonnegative on the indicator event.
We may immediately restate the preceding lemma as follows.

\begin{corollary}\label{cor:cvar_kkt_reduced_lvg}
Under Assumptions~\ref{ass:feasible}--\ref{ass:p_T_regular}, a point $x^\star\in\mathcal X$ solves
\eqref{eq:adl_primal_cvar} if and only if there exists $\theta^\star\in\R$ such that for each $i$,
\[
\begin{aligned}
&0<x_i^\star<q_i \ \Longrightarrow\ v_\beta(\ell_i(x_i^\star))=\theta^\star,\\
&x_i^\star=0 \ \Longrightarrow\ v_\beta(\ell_i(0))\le \theta^\star,\\
&x_i^\star=q_i \ \Longrightarrow\ v_\beta(\ell_i(q_i))\ge \theta^\star.
\end{aligned}
\]
Moreover, $\theta^\star\ge 0$. If $x_i^\star=q_i$ for some $i$, then
\[
0=v_\beta(\ell_i(q_i))\ge \theta^\star\ge 0,
\]
so $\theta^\star=0$.
\end{corollary}

We can now verify that the water-filling solution is optimal.

\begin{lemma}\label{lem:wf.optimal}
    The water-filling solution is optimal for \eqref{eq:adl_primal_cvar}.
\end{lemma}

\begin{proof}
 Let $x^{\mathrm{WF}}=y(t^\star)$ be the water-filling allocation of Theorem~\ref{thm:risk_neutral}.
Then
\[
\ell_i(x_i^{\mathrm{WF}})=\min\{\ell_i(0),t^\star\}
\qquad\text{for all }i.
\]
Set $\theta^\star\defeq v_\beta(t^\star)$.

Since $0<Q<\sum_i q_i$, we have $t^\star>0$. Hence
\[
q_i-x_i^{\mathrm{WF}}
=
\min\left\{
q_i,\frac{E_i}{p_\tau}t^\star
\right\}
>0,
\]
so $x_i^{\mathrm{WF}}<q_i$ for every $i$.

If $x_i^{\mathrm{WF}}>0$, then
\[
\ell_i(x_i^{\mathrm{WF}})=t^\star,
\]
and therefore
\[
v_\beta(\ell_i(x_i^{\mathrm{WF}}))
=
v_\beta(t^\star)
=
\theta^\star.
\]
If $x_i^{\mathrm{WF}}=0$, then
\[
\ell_i(0)\le t^\star,
\]
and monotonicity of $v_\beta$ yields
\[
v_\beta(\ell_i(0))
\le
v_\beta(t^\star)
=
\theta^\star.
\]
Thus $x^{\mathrm{WF}}$ satisfies the
conditions of Corollary~\ref{cor:cvar_kkt_reduced_lvg}, and is therefore optimal.
\end{proof}

\subsubsection{Non-Uniqueness}

\begin{proposition}\label{prop:cvar_nonunique}
Let $x^\star$ be an optimizer of \eqref{eq:adl_primal_cvar}. Define the stressed set for a feasible allocation $x$ by
\[
\mathcal D(x)\defeq \{\,i:\ \bar p_i(x_i)=p_\beta\,\}=\{\,i:\ \ell_i(x_i)\geq \ell_\beta\,\}.
\]
Suppose there exist $i\neq j$ in $\mathcal D(x^\star)$ and $\varepsilon_0>0$ such that for all
$\varepsilon\in[-\varepsilon_0,\varepsilon_0]$ the perturbed point $x^\star(\varepsilon)$ defined by
\[
x^\star_i(\varepsilon)=x^\star_i+\varepsilon,\qquad
x^\star_j(\varepsilon)=x^\star_j-\varepsilon,\qquad
x^\star_k(\varepsilon)=x^\star_k \ (k\notin\{i,j\})
\]
remains feasible and satisfies $i,j\in\mathcal D(x^\star(\varepsilon))$.
Then $x^\star(\varepsilon)$ is optimal for all $\varepsilon\in[-\varepsilon_0,\varepsilon_0]$, and the set of
optimizers contains a continuum.
\end{proposition}

\begin{proof}
    Fix $i\in\mathcal D(x^\star)$. Since $\bar p_i(x_i)=\max\{p_\beta,p_i^{(z)}(x_i)\}$,
$\bar p_i(x_i)=p_\beta$ means $p_i^{(z)}(x_i)\le p_\beta$. In this regime, the $\beta$-quantile of
$\sigma_i(x_i,p_T)$ is $\mathrm{VaR}_\beta(\sigma_i)=\sigma_i(x_i,p_\beta)$ and the excess loss above
this quantile is
\[
\big(\sigma_i(x_i,p_T)-\mathrm{VaR}_\beta(\sigma_i)\big)_+
=(q_i-x_i)\,(p_T-p_\beta)_+,
\]
so the Rockafellar--Uryasev formula gives
\[
f_i(x_i)=\mathrm{CVaR}_\beta(\sigma_i(x_i,p_T))
=\sigma_i(x_i,p_\beta)+\frac{1}{1-\beta}\E\big[(q_i-x_i)(p_T-p_\beta)_+\big].
\]
Let $\mu_{\mathrm{tail}}\defeq \E[p_T\mid p_T\ge p_\beta]$. Since $p_T$ is atomless,
$\E[(p_T-p_\beta)_+]=(1-\beta)\,(\mu_{\mathrm{tail}}-p_\beta)$, and
$\sigma_i(x_i,p_\beta)=(q_i-x_i)(p_\beta-p_i^{(z)}(x_i))$ on $p_i^{(z)}(x_i)\le p_\beta$.
Hence
\[
f_i(x_i)=(q_i-x_i)\big(\mu_{\mathrm{tail}}-p_i^{(z)}(x_i)\big).
\]
Using $(q_i-x_i)p_i^{(z)}(x_i)=q_i p_i^{(e)}+m_i-x_i p_\tau$ (by definition of $p_i^{(z)}$),
we obtain the affine representation
\[
f_i(x_i)= q_i(\mu_{\mathrm{tail}}-p_i^{(e)}) + x_i(p_\tau-\mu_{\mathrm{tail}})-m_i,
\qquad\text{whenever }\bar p_i(x_i)=p_\beta,
\]
so on the stressed regime every such $f_i$ has the same constant slope $p_\tau-\mu_{\mathrm{tail}}$.

Now fix $i\neq j$ as in the statement and consider $x^\star(\varepsilon)$.
Feasibility holds since $\sum_k x_k^\star(\varepsilon)=Q$ and the box constraints are preserved by assumption.
Moreover, for all $\varepsilon\in[-\varepsilon_0,\varepsilon_0]$ both $i$ and $j$ remain in the stressed regime,
and the objective changes by
\[
f_i(x_i^\star+\varepsilon)+f_j(x_j^\star-\varepsilon)-f_i(x_i^\star)-f_j(x_j^\star)
=(p_\tau-\mu_{\mathrm{tail}})\varepsilon+(p_\tau-\mu_{\mathrm{tail}})(-\varepsilon)=0,
\]
while all other coordinates are unchanged. Therefore
$\sum_k f_k(x_k^\star(\varepsilon))=\sum_k f_k(x_k^\star)$, so $x^\star(\varepsilon)$ is also optimal.
This yields a continuum of optima.
\end{proof}

The following example illustrates that the Proposition is not vacuously true (i.e., it provides a simple example where the stated conditions hold and lead to a continuum of CVaR-optima). 

\begin{example}\label{ex:cvar_continuum}
Fix $p_\tau=1$ and take the terminal price to be
\[
p_T \;=\; 1+Y,
\qquad Y\sim \mathrm{Exp}(1),
\]
so that $p_T$ admits the density $f_T(p)=e^{-(p-1)}\I{p\ge 1}$, which is strictly positive on $[p_\tau,\infty)$.

Let $\beta=1/2$, so that
\[
p_\beta=\mathrm{VaR}_\beta(p_T)=1+\log 2,
\qquad
\mu_{\mathrm{tail}}= \E[p_T\mid p_T\ge p_\beta]=p_\beta+1=2+\log 2,
\]
where the last identity follows from the memoryless property of the exponential distribution.

Consider $n=2$ short accounts with identical parameters
\[
q_1=q_2=4,\qquad p_1^{(e)}=p_2^{(e)}=1,\qquad m_1=m_2=1,
\]
so that $E_1=E_2=q_i(p_i^{(e)}-p_\tau)+m_i=1>0$.  Choose a buyback quantity $Q=1$.
Then the feasible set is
\[
\mathcal{X}=\{(x_1,x_2)\in\R^2:\ x_1+x_2=1,\ 0\le x_1,x_2\le 4\}.
\]

For any $x\in\mathcal{X}$ we have $x_i\le 1$, hence
\[
p_i^{(z)}(x_i)=p_\tau+\frac{E_i}{q_i-x_i}
=1+\frac{1}{4-x_i}\ \le\ 1+\frac{1}{3}=\frac{4}{3}
\ <\ 1+\log 2 = p_\beta.
\]
Therefore $\bar p_i(x_i)=\max\{p_\beta,p_i^{(z)}(x_i)\}=p_\beta$ for both $i=1,2$ and for every feasible $x$,
i.e.\ $\mathcal{D}(x)=\{1,2\}$ for all $x\in\mathcal{X}$. Arguing as in Proposition~\ref{prop:cvar_nonunique}, every feasible allocation $x\in\mathcal{X}$ is optimal for~\eqref{eq:adl_primal_cvar}, so the optimizer set contains a continuum. In particular,
\[
(1/2,\,1/2)\quad\text{and}\quad (1,\,0)\quad\text{and}\quad (0,\,1)
\]
are distinct CVaR-optimal solutions with $x^{\mathrm{WF}}=(Q/2,Q/2)=(1/2,1/2)$ being optimal but not unique.
\end{example}

\subsection{CVaR Closed-Form under Geometric Brownian Motion}\label{app:closeformGBM}

Assume that the terminal price follows a geometric Brownian motion,
\begin{equation}\label{eq:gbm_model}
p_T
=
p_\tau\exp\Bigl((\mu-\tfrac12\sigma^2)\Delta+\sigma\sqrt{\Delta}\,Z\Bigr),
\qquad Z\sim N(0,1),
\qquad \Delta\defeq T-\tau,
\end{equation}
with parameters $\mu\in\R$ and $\sigma>0$. Then $\log p_T$ is Gaussian with
\[
\E[p_T]=p_\tau e^{\mu\Delta},
\qquad
\mathrm{Var}(p_T)=p_\tau^2 e^{2\mu\Delta}\bigl(e^{\sigma^2\Delta}-1\bigr).
\]
Let $z_\beta\defeq\Phi^{-1}(\beta)$ so that (by monotonicity) the $\beta$-quantile of $p_T$ is
\begin{equation}\label{eq:p_beta}
p_\beta=
p_\tau\exp\Bigl((\mu-\tfrac12\sigma^2)\Delta+\sigma\sqrt{\Delta}\,z_\beta\Bigr).
\end{equation}
We will also use the tail conditional mean
\begin{equation*}
\mu_{\mathrm{tail}}
\defeq
\E\bigl[p_T\,\big|\,p_T\ge p_\beta\bigr]
=
\frac{\E\bigl[p_T\mathbf 1_{\{p_T\ge p_\beta\}}\bigr]}{1-\beta}.
\end{equation*}
A direct computation yields a closed form for $\mu_{\mathrm{tail}}$.

\begin{lemma}[Tail mean of GBM]\label{lem:gbm_tail_mean}
Under \eqref{eq:gbm_model}, one has
\begin{equation}\label{eq:gbm_tail_mean_formula}
\mu_{\mathrm{tail}}
=
\frac{p_\tau e^{\mu\Delta}\,\Phi(\sigma\sqrt{\Delta}-z_\beta)}{1-\beta}.
\end{equation}
\end{lemma}

\begin{proof}
Write $p_T=p_\tau e^{(\mu-\frac12\sigma^2)\Delta+\sigma\sqrt{\Delta}Z}$ and set $s\defeq\sigma\sqrt{\Delta}$.
Then
\[
\E\bigl[p_T\mathbf 1_{\{p_T\ge p_\beta\}}\bigr]
=
p_\tau e^{(\mu-\frac12\sigma^2)\Delta}\,
\E\bigl[e^{sZ}\mathbf 1_{\{Z\ge z_\beta\}}\bigr].
\]
Using the identity $\E[e^{sZ}\mathbf 1_{\{Z\ge a\}}]=e^{\frac12 s^2}\Phi(s-a)$ for $Z\sim N(0,1)$,
we obtain
\[
\E\bigl[p_T\mathbf 1_{\{p_T\ge p_\beta\}}\bigr]
=
p_\tau e^{(\mu-\frac12\sigma^2)\Delta}\,e^{\frac12 s^2}\Phi(s-z_\beta)
=
p_\tau e^{\mu\Delta}\Phi(s-z_\beta).
\]
Dividing by $1-\beta$ yields \eqref{eq:gbm_tail_mean_formula}.
\end{proof}

Since $p\mapsto \sigma_i(x,p)$ is nondecreasing and continuous, the $\beta$-quantile of
$\sigma_i(x,p_T)$ is realized by the price quantile $p_\beta$, and thus
\begin{equation}\label{eq:alpha_star_gbm}
\mathrm{VaR}_\beta\bigl(\sigma_i(x,p_T)\bigr)
=
\sigma_i(x, p_\beta)
=
(q_i-x)\,\bigl(p_\beta-p_i^{(z)}(x)\bigr)_+.
\end{equation}
The CVaR admits a piecewise closed form depending on the relative position of $p_i^{(z)}(x)$ and $p_\beta$.

\begin{lemma}[Per-account CVaR under GBM]
Fix $i$ and $x\in[0,q_i)$. Let $p_\beta$ be given by \eqref{eq:p_beta}.
Then
\[
\mathrm{CVaR}_\beta\bigl(\sigma_i(x,p_T)\bigr)
=
\E\bigl[\sigma_i(x,p_T)\,\big|\,p_T\ge p_\beta\bigr]
\]
and is given explicitly as follows.

\smallskip
\noindent (i) If $p_i^{(z)}(x)\le p_\beta$, then
\begin{equation}\label{eq:cvar_stressed}
\mathrm{CVaR}_\beta\bigl(\sigma_i(x,p_T)\bigr)
=
(q_i-x)\bigl(\mu_{\mathrm{tail}}-p_i^{(z)}(x)\bigr),
\end{equation}
where $\mu_{\mathrm{tail}}$ is given by \eqref{eq:gbm_tail_mean_formula}. Equivalently, using
$(q_i-x)p_i^{(z)}(x)=q_i p_i^{(e)}+m_i-x p_\tau$,
\begin{equation}\label{eq:cvar_stressed_affine}
\mathrm{CVaR}_\beta\bigl(\sigma_i(x,p_T)\bigr)
=
q_i\bigl(\mu_{\mathrm{tail}}-p_i^{(e)}\bigr)
+x\bigl(p_\tau-\mu_{\mathrm{tail}}\bigr)
-m_i.
\end{equation}

\smallskip
\noindent (ii) If $p_i^{(z)}(x)\ge p_\beta$, then
\begin{equation}\label{eq:cvar_nonstressed}
\mathrm{CVaR}_\beta\bigl(\sigma_i(x,p_T)\bigr)
=
\frac{q_i-x}{1-\beta}\,\E\bigl[(p_T-p_i^{(z)}(x))_+\bigr].
\end{equation}
Moreover, $\E[(p_T-K)_+]$ for $K>0$ admits the closed form
\begin{equation}\label{eq:call_expectation_physical}
\E\bigl[(p_T-K)_+\bigr]
=
p_\tau e^{\mu\Delta}\Phi(d_1(K)) - K\,\Phi(d_2(K)),
\end{equation}
where
\begin{equation}\label{eq:d1d2_physical}
d_1(K)
\defeq
\frac{\log(p_\tau/K)+(\mu+\frac12\sigma^2)\Delta}{\sigma\sqrt{\Delta}},
\qquad
d_2(K)\defeq d_1(K)-\sigma\sqrt{\Delta}.
\end{equation}
In particular, taking $K=p_i^{(z)}(x)$ in \eqref{eq:call_expectation_physical}--\eqref{eq:d1d2_physical} yields an explicit expression
for \eqref{eq:cvar_nonstressed}.
\end{lemma}

\begin{proof}
We first note from \cref{lem:sigma_explicit} that $p\mapsto\sigma_i(x,p)$ is nondecreasing and continuous. Since $p_T$ is atomless under
\eqref{eq:gbm_model}, the $\beta$-quantile of $\sigma_i(x,p_T)$ is $\sigma_i(x,p_\beta)$, giving \eqref{eq:alpha_star_gbm}.

Although $\sigma_i(x,p_T)$ can have an atom at $0$ (since $\sigma_i(x,p_T)=0$ on $\{p_T\le p_i^{(z)}(x)\}$),
we can still express $\mathrm{CVaR}_\beta$ as a conditional expectation on the \emph{price} tail $\{p_T\ge p_\beta\}$.
Let $F(p)\defeq \PR(p_T\le p)$ denote the cdf of $p_T$; under \eqref{eq:gbm_model}, $F$ is continuous.
For $u\in(0,1)$ define the (left) quantile
\[
p_u \defeq \mathrm{VaR}_u(p_T)=\inf\{p\in\R:\ F(p)\ge u\}.
\]
Set $U\defeq F(p_T)$. Since $F$ is continuous, $U$ is uniformly distributed on $(0,1)$ and $p_T=p_U$ with probability one.

Now write $K\defeq p_i^{(z)}(x)$ and recall $\sigma_i(x,p_T)=(q_i-x)\,(p_T-K)_+$.
For $y\ge 0$ we have
\[
\PR\bigl((p_T-K)_+\le y\bigr)=\PR(p_T\le K+y)=F(K+y),
\]
hence, for all $u\in(0,1)$,
\[
\mathrm{VaR}_u\bigl((p_T-K)_+\bigr)=(p_u-K)_+,
\qquad\text{and therefore}\qquad
\mathrm{VaR}_u\bigl(\sigma_i(x,p_T)\bigr)=(q_i-x)\,(p_u-K)_+=\sigma_i(x,p_u).
\]
Using the quantile-integral definition of $\mathrm{CVaR}_\beta$,
\[
\mathrm{CVaR}_\beta\bigl(\sigma_i(x,p_T)\bigr)
=\frac{1}{1-\beta}\int_\beta^1 \mathrm{VaR}_u\bigl(\sigma_i(x,p_T)\bigr)\,du
=\frac{1}{1-\beta}\int_\beta^1 \sigma_i(x,p_u)\,du.
\]
Since $U$ is uniform on $(0,1)$, for any integrable function $g$ one has
$\E[g(U)\mid U\ge \beta]=\frac{1}{1-\beta}\int_\beta^1 g(u)\,du$.
Applying this with $g(u)=\sigma_i(x,p_u)$ yields
\[
\frac{1}{1-\beta}\int_\beta^1 \sigma_i(x,p_u)\,du
=
\E\bigl[\sigma_i(x,p_U)\,\big|\,U\ge \beta\bigr]
=
\E\bigl[\sigma_i(x,p_T)\,\big|\,p_T\ge p_\beta\bigr],
\]
where the last equality uses $p_T=p_U$ with probability one and $\{U\ge \beta\}=\{p_T\ge p_\beta\}$ (by continuity and monotonicity of $F$).

If $p_i^{(z)}(x)\le p_\beta$, then $(p_T-p_i^{(z)}(x))_+=p_T-p_i^{(z)}(x)$ on $\{p_T\ge p_\beta\}$, so using \cref{lem:sigma_explicit} we obtain
\[
\mathrm{CVaR}_\beta(\sigma_i(x,p_T))
=
(q_i-x)\E\bigl[p_T-p_i^{(z)}(x)\,\big|\,p_T\ge p_\beta\bigr]
=
(q_i-x)\bigl(\mu_{\mathrm{tail}}-p_i^{(z)}(x)\bigr),
\]
which is \eqref{eq:cvar_stressed}. The affine form \eqref{eq:cvar_stressed_affine} follows by substituting
$(q_i-x)p_i^{(z)}(x)=q_i p_i^{(e)}+m_i-x p_\tau$.

If $p_i^{(z)}(x)\ge p_\beta$, then on the event $\{p_T<p_\beta\}$ we also have
$p_T<p_i^{(z)}(x)$, and therefore
\[
\sigma_i(x,p_T)=(q_i-x)\bigl(p_T-p_i^{(z)}(x)\bigr)_+=0.
\]
Hence, using the conditional-tail representation established above,
\[
\mathrm{CVaR}_\beta\bigl(\sigma_i(x,p_T)\bigr)
=
\E\bigl[\sigma_i(x,p_T)\mid p_T\ge p_\beta\bigr]
=
\frac{1}{1-\beta}\E\bigl[\sigma_i(x,p_T)\bigr].
\]
Substituting $\sigma_i(x,p_T)=(q_i-x)\bigl(p_T-p_i^{(z)}(x)\bigr)_+$ yields
\[
\mathrm{CVaR}_\beta\bigl(\sigma_i(x,p_T)\bigr)
=
\frac{q_i-x}{1-\beta}\E\bigl[(p_T-p_i^{(z)}(x))_+\bigr],
\]
which is \eqref{eq:cvar_nonstressed}.
Finally, \eqref{eq:call_expectation_physical}--\eqref{eq:d1d2_physical} are obtained by a standard lognormal
calculation (equivalently, the Black--Scholes call expectation with drift $\mu$ rather than the risk-free rate).
\end{proof}

\section{Details and Proofs for Multi-Asset Cross-Margining}\label{app:one_factor_proofs}

\subsection{Example: Distinct Optimizers for Different Risk Measures}\label{app:distinct.opt.cvar}
In this section, we construct an example where two risk measures $\mathrm{CVaR}_{\beta_j}$ with $\beta_1\neq\beta_2$ lead to different (and unique) optimal allocations.

Consider $d=2$ assets and $n=2$ accounts. Let the ADL time be $\tau$ with reference prices
\[
p_\tau=(1,1),
\qquad
p_i^{(e)}=p_\tau\ \ \text{for }i=1,2,
\]
and write the price changes as $r = p_T-p_\tau$. Then for any ADL allocation $x_i\in\R^2$,
\[
e_i(x_i,p_T)=m_i-(q_i-x_i)^\top r,
\qquad
\sigma_i(x_i,p_T)=\big((q_i-x_i)^\top r-m_i\big)_+,
\qquad
\Lscr(x,p_T)=\sum_{i=1}^2 \sigma_i(x_i,p_T).
\]
Take the signed position vectors (all shorts)
$q_1=(10,0)$, $q_2=(10,10)$,
and margins $m_1=18$, $m_2=40$.
Assume the exchange must reduce aggregate exposure by
$Q=(10,0)$.
By monotone deleveraging, $0\le x_i^k\le q_i^k$ for all $i,k$ and $x_1+x_2=Q$. Moreover, feasibility allows no changes in positions for the second asset. As a result, we can parameterize the feasible set by
\[
a = x_1^1\in[0,10],
\qquad
x_1=(a,0),\qquad x_2=(10-a,0).
\]
Then the post-ADL positions are
\[
q_1-x_1=(10-a,0),\qquad q_2-x_2=(a,10).
\]
Let $p_T$ take three values (expressed via $r=p_T-p_\tau$) with the following probabilities:
\[
r_0=(0,0)\ \text{ w.p. }0.90,\qquad
r_1=(3,0)\ \text{ w.p. }0.05,\qquad
r_2=(1,4)\ \text{ w.p. }0.05.
\]

\begin{example}
    In the stated model, the $\mathrm{CVaR}_\beta$-optimal ADL allocation depends on $\beta$. Specifically, for $\beta\in\{0.90,0.95\}$, the optimizers are unique and given by 
\[
a^\star_{0.90}=4 \neq 3=a^\star_{0.95}.
\]
Equivalently, the optimal allocations are, respectively,
\[
x_1(a^\star_{0.90})=(4,0),\quad x_2(a^\star_{0.90})=(6,0)
\qquad\text{and}\qquad
x_1(a^\star_{0.95})=(3,0),\quad x_2(a^\star_{0.95})=(7,0).
\]   
\end{example}

\begin{proof}
For $a\in[0,10]$, write $x(a)$ for the induced allocation and denote the losses in each scenario by $\Lscr_j(a)= \Lscr\bigl(x(a),\,p_\tau+r_j\bigr)$, $j=0,1,2$. 
In scenario $r_1=(3,0)$,
\[
\Lscr_1(a)
=
\big(3(10-a)-18\big)_+ + \big(3a-40\big)_+
=
(12-3a)_+,
\qquad a\in[0,10],
\]
since $3a\le 30<40$ on $[0,10]$.
In scenario $r_2=(1,4)$,
\[
\Lscr_2(a)
=
\big((10-a)-18\big)_+ + \big(a+10\cdot 4-40\big)_+
=
0 + (a)_+
=
a,
\qquad a\in[0,10],
\]
and in scenario $r_0=(0,0)$ we have $\Lscr_0(a)=0$ for all $a$.
Because $\mathbb{P}(r=r_0)=0.90$ and $\Lscr_0(a)=0$ for all $a$, while
$\mathbb{P}(r\in\{r_1,r_2\})=0.10$ and
\[
\Lscr_1(a)=(12-3a)_+\ge 0,
\qquad
\Lscr_2(a)=a\ge 0,
\qquad a\in[0,10],
\]
the loss is identically zero in scenario $0$ and can be positive only in scenarios $1$ and $2$
(which together have total probability $0.10$). 
In particular, $\mathrm{VaR}_{0.90}(\Lscr(a))=0$,
so the $10\%$ upper tail at level $\beta=0.90$ can be identified, for the purpose of computing $\mathrm{CVaR}_{0.90}$, with scenarios $1$ and $2$.\footnote{Note that there is an atom at $0$, so the $10\%$ tail is not always uniquely determined: it consists of all scenarios with loss strictly greater than $0$, together with an arbitrary selection of additional probability mass at loss $0$ to reach total mass $0.10$. In particular, we may choose a representative tail by taking scenarios $1$ and $2$ (and, when either loss is $0$ for a given $a$, interpreting the scenario as part of the mass at the VaR level).} 

Hence, for $\beta=0.90$,
\[
\mathrm{CVaR}_{0.90}(\Lscr(a))
=
\frac{0.05\,\Lscr_1(a)+0.05\,\Lscr_2(a)}{0.10}
=
\frac{\Lscr_1(a)+\Lscr_2(a)}{2}
=
\frac{(12-3a)_+ + a}{2}.
\]
If $a\le 4$, then $(12-3a)_+=12-3a$ and $\mathrm{CVaR}_{0.90}(\Lscr(a))=6-a$;
if $a\ge 4$, then $(12-3a)_+=0$ and $\mathrm{CVaR}_{0.90}(\Lscr(a))=a/2$.
Thus the unique minimizer is $a^\star_{0.90}=4$.

For $\beta=0.95$, the tail mass is $5\%$; by the quantile-integral definition of CVaR, since each nonzero scenario has probability $0.05$, we have
\[
\mathrm{CVaR}_{0.95}(\Lscr(a))=\max\{\Lscr_1(a),\Lscr_2(a)\}
=\max\{(12-3a)_+,\ a\}.
\]
For $a\le 4$, this equals $\max\{12-3a,\ a\}$, minimized when $12-3a=a$, i.e.\ at $a=3$,
yielding value $3$. For $a\ge 4$, it equals $a$, which is minimized at $a=4$ with value $4$.
Therefore the unique minimizer is
$a^\star_{0.95}=3$.
\end{proof}

\subsection{Proof of Proposition~\ref{prop:market_clearing_kkt}}\label{app:proof.mkt.clearing}

We consider \eqref{eq:adl_multiasset_exp_loss} and write
\[
f_i(x_i)\defeq \E\!\left[\sigma_i(x_i,p_T)\right],
\qquad
f(x)\defeq \E\!\left[\Lscr(x,p_T)\right]=\sum_{i=1}^n f_i(x_i).
\]
For each fixed $p_T\in\R^d$, equity $e_i(x_i,p_T)$ is affine in $x_i$, hence
$\sigma_i(x_i,p_T)=(-e_i(x_i,p_T))_+$ is convex and piecewise affine in $x_i$.
Therefore $f_i(x_i)=\E[\sigma_i(x_i,p_T)]$ is convex as an expectation of convex functions, and it is finite under
Assumption~\ref{ass:integrable_pT_cross}.

Recall the definition of $\mathcal Y_i$ from \eqref{eq:per.acct.feasible} and define $\mathcal Y \defeq \prod_{i=1}^n \mathcal Y_i$,
so that the feasible set is equivalently represented as $\mathcal X = \{x\in\mathcal Y:\ \sum_{i=1}^n x_i = Q\}$.
Each $\mathcal Y_i$ is nonempty, closed, and bounded, hence compact, and Assumption~\ref{ass:cross_margin_feasible} implies
$\mathcal X\neq\emptyset$. Since $f$ is convex and finite on $\mathcal Y$ and $\mathcal X$ is compact, an optimizer of
\eqref{eq:adl_multiasset_exp_loss} exists.

We now turn to the characterization of the solution. Rewrite the primal problem as
\begin{equation}\label{eq:expected_primal_for_proof}
\min_{x\in\mathcal Y}\ \sum_{i=1}^n f_i(x_i)
\qquad \text{s.t.}\qquad
\sum_{i=1}^n x_i = Q .
\end{equation}
Introduce a multiplier $\lambda\in\R^d$ for the coupling constraint and form the partial Lagrangian
\[
\widehat{\mathcal L}(x,\lambda)
\defeq
\sum_{i=1}^n f_i(x_i) + \lambda^\top\Big(\sum_{i=1}^n x_i - Q\Big)
=
-\lambda^\top Q + \sum_{i=1}^n \bigl(f_i(x_i)+\lambda^\top x_i\bigr).
\]
Define the dual function $g(\lambda)\defeq \min_{x\in\mathcal Y}\widehat{\mathcal L}(x,\lambda)$ and the per-account value functions
\[
\phi_i(\lambda)\defeq \min_{x_i\in\mathcal Y_i}\bigl\{f_i(x_i)+\lambda^\top x_i\bigr\},
\qquad i=1,\dots,n,
\]
so that
\[
g(\lambda)= -\lambda^\top Q + \sum_{i=1}^n \phi_i(\lambda),
\qquad
\max_{\lambda\in\R^d} g(\lambda)
=
\max_{\lambda\in\R^d}\min_{x\in\mathcal Y}\widehat{\mathcal L}(x,\lambda).
\]
Note that $\mathcal Y$ is compact so a minimizer always exists. Moreover, each $\phi_i(\cdot)$ is the pointwise minimum over $x_i\in\mathcal Y_i$ of affine functions of $\lambda$. As a result, the functions $\phi_i(\lambda)$ and $g(\lambda)$ are concave. 

If for some $k$ we have $Q^k=\sum_i l_i^k$ or $Q^k=\sum_i u_i^k$, then feasibility forces $x_i^k=l_i^k$ or $x_i^k=u_i^k$
for all $i$, so these coordinates are effectively equality constraints and may be treated as pinned.
On the remaining coordinates where $Q^k\in(\sum_i l_i^k,\sum_i u_i^k)$, one can construct a feasible $\bar x\in\mathcal X$
with $l_i^k<\bar x_i^k<u_i^k$ whenever $u_i^k>l_i^k$ as follows. Since $\sum_{i=1}^n (u_i^k-l_i^k)>0$ we can define
\[
\alpha^k \defeq
\frac{Q^k-\sum_{i=1}^n l_i^k}{\sum_{i=1}^n (u_i^k-l_i^k)} \in (0,1),
\qquad
\bar x_i^k \defeq l_i^k + \alpha^k (u_i^k-l_i^k).
\]
Hence $\mathcal X$ has nonempty relative interior in its affine hull, which yields a Slater-type constraint qualification
for the polyhedral constraints. Consequently strong duality holds for \eqref{eq:expected_primal_for_proof}, and there exists
a primal--dual optimal pair $(x^\star,\lambda^\star)$ forming a saddle point of $\widehat{\mathcal L}$.

Recall the best-response correspondence $X_i(\lambda)$ from \eqref{eq:best.resp.cross.margin}.
Since $\mathcal Y$ is a product set and $\widehat{\mathcal L}$ separates across $i$, we have
\[
\argmin_{x\in\mathcal Y}\ \widehat{\mathcal L}(x,\lambda)
=
\prod_{i=1}^n X_i(\lambda).
\]
Moreover, $\widehat{\mathcal L}$ is differentiable in $\lambda$ with
\[
\nabla_\lambda \widehat{\mathcal L}(x,\lambda)=\sum_{i=1}^n x_i - Q.
\]
Because $g(\lambda)=\min_{x\in\mathcal Y}\widehat{\mathcal L}(x,\lambda)$ is concave, its superdifferential satisfies an extension of Danskin's theorem (cf.~\cite[Thm.~9.27]{shapiro2021lectures}),
\begin{equation}\label{eq:superdiff_g_separable_proof}
\partial g(\lambda)
=
\mathrm{conv}\Big\{\sum_{i=1}^n x_i - Q:\ x_i\in X_i(\lambda)\ \forall i\Big\}
\end{equation}
where $\mathrm{conv}\{\cdot\}$ denotes the convex hull.
Since $\mathcal Y_i$ is compact and $x_i\mapsto f_i(x_i)+\lambda^\top x_i$ is convex and continuous,
$X_i(\lambda)$ is nonempty and convex for each $i$.
It is then straightforward to verify that the convex hull in
\eqref{eq:superdiff_g_separable_proof} is redundant. Thus
\[
\partial g(\lambda)
=
\Big\{\sum_{i=1}^n x_i - Q:\ x_i\in X_i(\lambda)\ \forall i\Big\}.
\]
Since $g(\cdot)$ is a finite concave function, $\lambda^\star$ is an unconstrained maximizer of $g(\cdot)$ if and only if it satisfies $0\in\partial g(\lambda^\star)$. This is, in turn, equivalent to the
existence of selections $x_i^\star\in X_i(\lambda^\star)$ such that $\sum_{i=1}^n x_i^\star=Q$.

Finally, by strong duality, $(x^\star,\lambda^\star)$ is primal--dual optimal if and only if it is a saddle point of
$\widehat{\mathcal L}$, which holds if and only if $x^\star\in\argmin_{x\in\mathcal Y}\ \widehat{\mathcal L}(x,\lambda^\star)$
and $\sum_i x_i^\star=Q$. Using the product characterization of the minimizers yields exactly the stated condition
\[
x_i^\star\in X_i(\lambda^\star)\ \ \text{for all }i,
\qquad\text{and}\qquad
\sum_{i=1}^n x_i^\star=Q.
\]
This proves Proposition~\ref{prop:market_clearing_kkt}. \qed

\subsection{Properties of $\psi(\cdot)$} \label{app:prop.psi}

\begin{proposition}\label{prop:psi_strictly_increasing}
Under Assumption~\ref{ass:single_factor}, $\psi(z)\defeq \E\big[(\epsilon z-1)_+\big]$ is convex and continuously differentiable on $\mathbb R$, with derivative
\begin{equation}\label{eq:psi_derivative_formula}
\psi'(z)=\E\big[\epsilon\,\mathbf 1_{\{\epsilon z>1\}}\big],\qquad z\in\mathbb R.
\end{equation}
Moreover, $\psi'$ is \emph{strictly increasing} on $\mathbb R$.
\end{proposition}

\begin{proof}
For each fixed $x\in\mathbb R$, the map $z\mapsto (xz-1)_+$ is convex. Therefore
\[
\psi(z)=\E\big[(\epsilon z-1)_+\big]
\]
is convex as an expectation of convex functions.

Since under Assumption \ref{ass:single_factor} $\epsilon$ has a strictly positive density $\varphi_\epsilon$ and $\E[|\epsilon|]<\infty$, we can write
\[
\psi(z)=\int_{\mathbb R} (\;xz-1\;)_{+}\,\varphi_\epsilon(x)\,dx.
\]
For $z>0$, the condition $xz-1>0$ is equivalent to $x>1/z$, hence
\begin{equation}\label{eq:psi_r_pos}
\psi(z)=\int_{1/z}^{\infty} (xz-1)\,\varphi_\epsilon(x)\,dx,\qquad z>0.
\end{equation}
For $z<0$, the condition $xz-1>0$ is equivalent to $x<1/z$, hence
\begin{equation}\label{eq:psi_r_neg}
\psi(z)=\int_{-\infty}^{1/z} (xz-1)\,\varphi_\epsilon(x)\,dx,\qquad z<0.
\end{equation}
Finally, $\psi(0)=\E[(-1)_+]=0$.

We differentiate \eqref{eq:psi_r_pos} for $z>0$ using Leibniz' rule. The boundary term vanishes and we obtain
\[
\psi'(z)=\int_{1/z}^\infty x\,\varphi_\epsilon(x)\,dx,\qquad z>0.
\]
Similarly, differentiating \eqref{eq:psi_r_neg} for $z<0$ yields
\[
\psi'(z)=\int_{-\infty}^{1/z} x\,\varphi_\epsilon(x)\,dx,\qquad z<0.
\]
At $z=0$, we compute the one-sided limits:
\[
\lim_{z\downarrow 0}\psi'(z)=\lim_{z\downarrow 0}\int_{1/z}^\infty x\varphi_\epsilon(x)\,dx=0,
\qquad
\lim_{z\uparrow 0}\psi'(z)=\lim_{z\uparrow 0}\int_{-\infty}^{1/z} x\varphi_\epsilon(x)\,dx=0,
\]
using $\E[|\epsilon|]<\infty$ and dominated convergence (as $1/z\to+\infty$ for $z\downarrow 0$ and
$1/z\to-\infty$ for $z\uparrow 0$).
By the right (resp.\ left) continuity of right (resp.\ left) derivatives of convex functions, we conclude that $\psi'$ exists at $0$ with $\psi'(0)=0$. Altogether, we see that $\psi'$ is continuous on $\mathbb R$.

To verify \eqref{eq:psi_derivative_formula}, note that for any $z> 0$,
\[
\int_{1/z}^{\infty} x\,\varphi_\epsilon(x)\,dx
=\int_{\mathbb R} x\,\mathbf 1_{\{xz>1\}}\varphi_\epsilon(x)\,dx,
\]
and similarly for $z<0$, so the integral expressions above are exactly
$\E[\epsilon\,\mathbf 1_{\{\epsilon z>1\}}]$. For $z=0$ the indicator is identically $0$,
so \eqref{eq:psi_derivative_formula} holds for all $z\in\mathbb R$.

We finally show strict monotonicity of $\psi'$.
Let $z<z'$.

\smallskip
\noindent\textbf{Case 1: $0<z<z'$.}
Then $1/z' < 1/z$ and
\[
\psi'(z')-\psi'(z)
=
\int_{1/z'}^\infty x\varphi_\epsilon(x)\,dx-\int_{1/z}^\infty x\varphi_\epsilon(x)\,dx
=
\int_{1/z'}^{1/z} x\varphi_\epsilon(x)\,dx.
\]
On $(1/z',1/z)$ we have $x>0$ and $\varphi_\epsilon(x)>0$, hence the integral is strictly positive.

\smallskip
\noindent\textbf{Case 2: $z<z'<0$.}
Then $1/z' < 1/z <0$ and
\[
\psi'(z')-\psi'(z)
=
\int_{-\infty}^{1/z'} x\varphi_\epsilon(x)\,dx-\int_{-\infty}^{1/z} x\varphi_\epsilon(x)\,dx
=
-\int_{1/z'}^{1/z} x\varphi_\epsilon(x)\,dx.
\]
On $(1/z',1/z)\subset(-\infty,0)$ we have $x<0$ and $\varphi_\epsilon(x)>0$, so
$-\int_{1/z'}^{1/z} x\varphi_\epsilon(x)\,dx>0$.

\smallskip
\noindent\textbf{Case 3: $z<0<z'$.}
Then $\psi'(z)=\int_{-\infty}^{1/z} x\varphi_\epsilon(x)\,dx<0$ (as the integrand is strictly negative),
while $\psi'(z')=\int_{1/z'}^\infty x\varphi_\epsilon(x)\,dx>0$, hence $\psi'(z')>\psi'(z)$.

\smallskip
If $z=0<z'$, then $\psi'(0)=0<\psi'(z')$. If $z<0= z'$, then $\psi'(z)<0=\psi'(0)$.

\smallskip
In all cases, $z<z'$ implies $\psi'(z')>\psi'(z)$, hence $\psi'$ is strictly increasing on $\mathbb R$, by continuity.
\end{proof}

\subsection{Proof of Theorem~\ref{thm:factor_filling_verification}}

By Proposition~\ref{prop:psi_strictly_increasing} $\psi$ is strictly convex. Fix $\eta\in\R$ and define $h_\eta(z)\defeq \psi(z)-\eta z$. Since $\psi$ is strictly convex, $h_\eta$ is strictly convex, so
the minimizer of $h_\eta$ over the interval $[\underline\ell_i,\overline\ell_i]$ is unique; this proves the first part of~(i). The standard first-order/KKT conditions for minimizing a differentiable convex function over an interval lead to the cases
\[
\begin{cases}
\psi'(\ell_i^\star(\eta))=\eta, & \underline\ell_i<\ell_i^\star(\eta)<\overline\ell_i,\\
\psi'(\underline\ell_i)\ge \eta, & \ell_i^\star(\eta)=\underline\ell_i,\\
\psi'(\overline\ell_i)\le \eta, & \ell_i^\star(\eta)=\overline\ell_i.
\end{cases}
\]
Since $\psi'$ is continuous and strictly increasing, this is equivalent to the clipped rule
\eqref{eq:water_fill_piecewise_main}, proving (ii). By using that $(\psi')^{-1}$ must also be strictly increasing on the range of $\psi'$, the same characterization implies that $\eta\mapsto\ell_i^\star(\eta)$ is
continuous and nondecreasing, proving the second part of~(i).

Next, define
\[
H(\eta)\defeq \sum_{i=1}^n E_i\,\ell_i^\star(\eta),
\qquad
R\defeq v^\top\Big(\sum_{i=1}^n q_i-Q\Big).
\]
By (i), $H$ is continuous and nondecreasing. Moreover, by \eqref{eq:water_fill_piecewise_main},
\[
\lim_{\eta\to-\infty}H(\eta)=\sum_{i=1}^n E_i\,\underline\ell_i,
\qquad
\lim_{\eta\to+\infty}H(\eta)=\sum_{i=1}^n E_i\,\overline\ell_i.
\]

In fact, these endpoint values are attained. If
\[
\eta \le \min_{1\le i\le n}\psi'(\underline\ell_i),
\]
then \eqref{eq:water_fill_piecewise_main} implies
\[
\ell_i^\star(\eta)=\underline\ell_i
\qquad\text{for all }i,
\]
so that
\[
H(\eta)=\sum_{i=1}^n E_i\,\underline\ell_i.
\]
Similarly, if
\[
\eta \ge \max_{1\le i\le n}\psi'(\overline\ell_i),
\]
then
\[
\ell_i^\star(\eta)=\overline\ell_i
\qquad\text{for all }i,
\]
and hence
\[
H(\eta)=\sum_{i=1}^n E_i\,\overline\ell_i.
\]

Since $\mathcal X\neq\emptyset$ by Assumption~\ref{ass:cross_margin_feasible}, pick any feasible $x\in\mathcal X$. Then for each
$i$, $\ell_i^{(v)}(x_i)\in[\underline\ell_i,\overline\ell_i]$, so
\[
\sum_{i=1}^n E_i\,\underline\ell_i \;\le\; \sum_{i=1}^n E_i\,\ell_i^{(v)}(x_i)\;\le\;\sum_{i=1}^n E_i\,\overline\ell_i.
\]
Using the identity $\sum_i E_i\ell_i^{(v)}(x_i)=\sum_i v^\top(q_i-x_i)=v^\top(\sum_i q_i-Q)=R$, we obtain
\[
\sum_{i=1}^n E_i\,\underline\ell_i \;\le\; R \;\le\; \sum_{i=1}^n E_i\,\overline\ell_i.
\]
By the intermediate value theorem, there exists $\eta^\star$ such that $H(\eta^\star)=R$, proving (iii).

Fix $\eta^\star$ satisfying \eqref{eq:eta_budget_main} and suppose there exists $x^\star\in\mathcal X$ with
$\ell_i^{(v)}(x_i^\star)=\ell_i^\star(\eta^\star)$ for all $i$.  Set the shadow-price vector $\lambda^\star\defeq \eta^\star v$ and recall \eqref{eqn:equity_factor_lvg}.
For each $i$ and any $x_i\in\mathcal Y_i$ we have
\[
\E[\sigma_i(x_i,p_T)]+(\lambda^\star)^\top x_i
=
E_i\psi(\ell_i^{(v)}(x_i))+\eta^\star v^\top x_i
=
\eta^\star v^\top q_i + E_i\big(\psi(\ell_i^{(v)}(x_i))-\eta^\star \ell_i^{(v)}(x_i)\big),
\]
so minimizing $f_i(x_i)+(\lambda^\star)^\top x_i$ over $x_i\in\mathcal Y_i$ is equivalent to minimizing
$\psi(z)-\eta^\star z$ over $z\in[\underline\ell_i,\overline\ell_i]$. Hence any $x_i\in\mathcal Y_i$ satisfying
$\ell_i^{(v)}(x_i)=\ell_i^\star(\eta^\star)$ is a best response to $\lambda^\star$. In particular,
$x_i^\star$ is a best response for each $i$, and the clearing constraint $\sum_i x_i^\star=Q$ holds because $x^\star\in\mathcal X$.
Therefore, by Proposition~\ref{prop:market_clearing_kkt},
$x^\star$ is optimal for \eqref{eq:adl_multiasset_exp_loss}.

\subsection{Example Where Water-Filling Fails}\label{app:water.filling.fail}

In general, the factor water-filling construction with clipping can fail to produce a solution because it enforces only the scalar constraint from \eqref{eq:sum_factor_fixed},
\[
\sum_i E_i\ell_i^{(v)}(x_i)=v^\top\left(\sum_{i=1}^nq_i -Q\right),
\]
and cannot generally satisfy the \emph{vector} clearing constraint $\sum_i x_i=Q$. We provide a simple illustrative counterexample here.

\begin{example} Suppose there are $d=2$ assets and $n=2$ accounts. Assume the single-factor model
\[
p_T=p_\tau+\epsilon v,\qquad v=(1,1),
\]
where $\epsilon$ satisfies Assumption~\ref{ass:single_factor}. We impose that the exchange has the ADL requirement
\[
Q=(Q^1,Q^2)=(0.2,0.8).
\]
Take equal initial equities $E_1=E_2=1$ and assume that account~1 holds only asset~1 and account~2 holds only asset~2:
\[
q_1=(1,0),\qquad q_2=(0,1).
\]
As a result, any feasible allocation $x=(x_1,x_2)$ satisfies
\[
0\le x_1^1\le 1,\qquad 0\le x_2^2\le 1,
\qquad x_1^2=x_2^1=0.
\]
Since we require $x_1+x_2=Q$, we conclude that the \emph{unique} feasible allocation is
\[
x_1^1=Q^1=0.2,\qquad x_2^2=Q^2=0.8.
\]

In terms of factor leverage, this unique feasible allocation leads to
\[
\ell_1^{(v)}(x_1)=1-x_1^1=0.8,\qquad \ell_2^{(v)}(x_2)=1-x_2^2=0.2,
\qquad \ell_1^{(v)}(x_1)+\ell_2^{(v)}(x_2)=1.
\]
If this problem were instead approached via 1D water filling on the factor leverage, each account solves an identical 1D problem (since their equities and factor leverage constraints are the same),
\[
\ell^\star_i(\eta)\in\argmin_{z\in[0,1]}\ \bigl\{\E\big[(\epsilon z - 1)_+\big]-\eta z\bigr\}, 
\quad i=1,2.
\]
Therefore, by symmetry, the solution produces equal factor leverages
\[
\ell_1^\star(\eta)=\ell_2^\star(\eta)=\ell^\star(\eta)
\]
for every $\eta$. At the same time, the necessary condition \eqref{eq:sum_factor_fixed} at an optimal $\eta^\star$ reads
\[
\ell_1^\star(\eta^\star)+\ell_2^\star(\eta^\star)=1.
\]
Taken together, this yields $\ell^\star(\eta^\star)=0.5$.
But $\ell_1^{(v)}(x_1)=0.5$ would imply $x_1^1=1-\ell_1^{(v)}(x_1)=0.5\neq Q^1=0.2$, and similarly $\ell_2^{(v)}(x_2)=0.5$ would imply
$x_2^2=0.5\neq Q^2=0.8$. Hence there is \emph{no} $\eta^\star$ for which the water-filling exposures can be realized
by an allocation satisfying $\sum_i x_i=Q$.
\end{example}

\subsection{Sufficient Conditions for Water-Filling on Factor Leverage}\label{app:suff.conditions.water.filling}

This section describes a regime in which the expected-loss ADL problem~\eqref{eq:adl_multiasset_exp_loss} admits a generalized water-filling structure,
but applied to \emph{factor leverage} rather than gross leverage. The key simplification is that, under a one-factor
price model, each account's expected shortfall depends on its post-ADL portfolio only through a scalar factor exposure.

Theorem~\ref{thm:factor_filling_verification} is stated as a verification result because, in general multi-asset ADL,
the vector constraint $\sum_i x_i=Q$ need not allow one to realize an arbitrary collection of scalar targets
$\{\ell_i^\star(\eta)\}_{i=1}^n$. A complementary observation is that, \emph{at an optimum}, the dual multiplier
$\lambda^\star$ can inherit a strong structure under a natural ``overlap'' condition across assets.

Recall the definitions
\[
f_i(x_i)\defeq \E\!\left[\sigma_i(x_i,p_T)\right],
\qquad
f(x)\defeq \E\!\left[\Lscr(x,p_T)\right]=\sum_{i=1}^n f_i(x_i).
\]
Under Assumptions~\ref{ass:cross_margin_feasible}--\ref{ass:integrable_pT_cross}, the expected-loss problem is convex and a relative-interior/polyhedral constraint qualification holds; see Appendix~\ref{app:proof.mkt.clearing}.
Consequently, the KKT conditions are necessary and sufficient for optimality.

Introduce multipliers $\lambda\in\R^d$ for the clearing constraint $\sum_i x_i=Q$, and for each $(i,k)$ introduce
$\mu_i^k\ge 0$ for $l_i^k-x_i^k\le 0$ and $\nu_i^k\ge 0$ for $x_i^k-u_i^k\le 0$.
The Lagrangian is
\[
\mathcal L(x,\lambda,\mu,\nu)
= f(x)
+\sum_{k=1}^d \lambda^k\Big(\sum_{i=1}^n x_i^k - Q^k\Big)
+\sum_{i=1}^n\sum_{k=1}^d \mu_i^k(l_i^k-x_i^k)
+\sum_{i=1}^n\sum_{k=1}^d \nu_i^k(x_i^k-u_i^k).
\]
A primal--dual quadruple $(x^\star,\lambda^\star,\mu^\star,\nu^\star)$ is optimal if and only if:
\begin{align}
&\sum_{i=1}^n x_i^\star = Q,
\qquad l_i^k \le x_i^{k,\star}\le u_i^k \quad \forall i,k, \label{eq:kkt_primal_feas}\\
&\mu_i^{k,\star}\ge 0,\qquad \nu_i^{k,\star}\ge 0 \quad \forall i,k, \label{eq:kkt_dual_feas}\\
&\mu_i^{k,\star}\,(l_i^k-x_i^{k,\star})=0,\qquad
\nu_i^{k,\star}\,(x_i^{k,\star}-u_i^k)=0 \quad \forall i,k, \label{eq:kkt_comp}\\
&0 \in \partial_{x_i} f_i(x_i^\star) + \lambda^\star - \mu_i^\star + \nu_i^\star
\quad \forall i. \label{eq:kkt_stationarity_general}
\end{align}

Assume now the single-factor model of Assumption~\ref{ass:single_factor}.
As in \eqref{eq:equity.simplified}, under $p_T=p_\tau+\epsilon v$ one has
\[
\sigma_i(x_i,p_T)=E_i(\epsilon\,\ell_i^{(v)}(x_i)-1)_+,
\]
so we may write
\begin{equation}\label{eq:f_i_factor_form_kkt}
f_i(x_i)=E_i\,\psi\!\big(\ell_i^{(v)}(x_i)\big).
\end{equation}
Moreover, $\psi$ is continuously differentiable with
\[
\psi'(z)=\E\!\left[\epsilon\,\mathbf 1_{\{\epsilon z>1\}}\right],\qquad z\in\R,
\]
and $\psi'$ is strictly increasing (see Proposition~\ref{prop:psi_strictly_increasing}). Using $\frac{\partial}{\partial x_i^k}\ell_i^{(v)}(x_i)=-v^k/E_i$ and \eqref{eq:f_i_factor_form_kkt}, we obtain
\begin{equation}\label{eq:f_subdiff_factor}
\frac{\partial}{\partial x_i^k} f(x)
=\frac{\partial}{\partial x_i^k} f_i(x_i)
= -\,v^k\,\psi'\!\big(\ell_i^{(v)}(x_i)\big),
\qquad \forall i,k.
\end{equation}
Substituting \eqref{eq:f_subdiff_factor} into \eqref{eq:kkt_stationarity_general} yields the single-factor stationarity condition
\begin{equation}\label{eq:kkt_stationarity_factor}
0 = -v^k \psi'\!\big(\ell^{(v)}_i(x_i^\star)\big) + \lambda^{k,\star} - \mu_i^{k,\star} + \nu_i^{k,\star},
\qquad \forall i,k.
\end{equation}
In particular, whenever $v^k\neq 0$, by complementary slackness
\begin{equation}\label{eq:ratio_identification}
l_i^k < x_i^{k,\star} < u_i^k
\quad\Longrightarrow\quad
\psi'\!\big(\ell^{(v)}_i(x_i^\star)\big) = \frac{\lambda^{k,\star}}{v^k}.
\end{equation}

Let
\[
K\defeq\{k\in\{1,\dots,d\}: Q^k\neq 0,\ v^k\neq 0\}
\]
denote the set of \emph{active} assets whose reductions affect factor exposure. Given an optimal solution~$x^\star$,
define the interior coordinates of account $i$ by
\[
F_i\defeq\{k\in K:\ l_i^k < x_i^{k,\star} < u_i^k\},\qquad i=1,\dots,n,
\]
and let $K_{\mathrm{int}}(x^\star)\defeq \cup_{i=1}^n F_i$ be the set of active coordinates that are strictly interior for at
least one account.
\begin{definition}[Interior-coverage graph]\label{def:interior_coverage_graph_active_main}
The \emph{interior-coverage graph} is the simple undirected graph
\[
G(x^\star)\defeq\bigl(K_{\mathrm{int}}(x^\star),E\bigr),
\]
where
\[
E \defeq \Big\{\{k,k'\}\subseteq K_{\mathrm{int}}(x^\star): k\neq k'\ \text{and}\ \exists\, i\ \text{s.t.}\ \{k,k'\}\subseteq F_i\Big\}.
\]
Equivalently, there is an edge between $k$ and $k'$ if and only if there exists an account whose reduction is strictly interior in
both coordinates.
\end{definition}

\begin{proposition}\label{prop:lambda_parallel_v_connected}
Let $(x^\star,\lambda^\star,\mu^\star,\nu^\star)$ satisfy the
KKT system \eqref{eq:kkt_primal_feas}--\eqref{eq:kkt_stationarity_general}. If $K_{\mathrm{int}}(x^\star)\neq\emptyset$
and the interior-coverage graph $G(x^\star)$ is connected\footnote{If $K_{\mathrm{int}}(x^\star)$ is a singleton it is trivially connected.}, then there exists $\eta^\star\in\R$
such that
\[
\lambda^{k,\star}=\eta^\star v^k,\qquad \forall k\in K_{\mathrm{int}}(x^\star).
\]
Moreover, for any $i$ and any $k\in F_i$, one has $\eta^\star=\psi'\!\big(\ell_i^{(v)}(x_i^\star)\big)$.
\end{proposition}

\begin{proof}
Fix any $i$ and any $k\in F_i$. Since $k\in K$, we have $v^k\neq 0$, and \eqref{eq:ratio_identification} gives
\[
\frac{\lambda^{k,\star}}{v^k}=\psi'\!\big(\ell_i^{(v)}(x_i^\star)\big).
\]
If $k,k'\in F_i$ for the same account $i$, then the same identity holds for $k'$ and hence
\begin{equation}\label{eq:ratio_edge_identity}
\frac{\lambda^{k,\star}}{v^k}=\frac{\lambda^{k',\star}}{v^{k'}}.
\end{equation}
Thus, along any edge $\{k,k'\}\in E$, the ratios $\lambda^{k,\star}/v^k$ and $\lambda^{k',\star}/v^{k'}$ agree.

If $|K_{\mathrm{int}}(x^\star)|=1$, the conclusion follows immediately by setting
$\eta^\star\defeq \lambda^{k,\star}/v^k$ on the unique vertex $k$.
Otherwise, connectivity implies that for any $k\in K_{\mathrm{int}}(x^\star)$ there exists a path from a fixed reference
vertex $k_0$ to $k$, and the edge-by-edge identity \eqref{eq:ratio_edge_identity} propagates along the path to yield
$\lambda^{k,\star}/v^k=\lambda^{k_0,\star}/v^{k_0}$ for all $k$.
Setting $\eta^\star\defeq \lambda^{k_0,\star}/v^{k_0}$ gives $\lambda^{k,\star}=\eta^\star v^k$ on
$K_{\mathrm{int}}(x^\star)$, and the identification
$\eta^\star=\psi'(\ell_i^{(v)}(x_i^\star))$ for $k\in F_i$ follows from the interior stationarity equality above.
\end{proof}

\begin{remark}\label{rem:lambda_parallel_v_practice}
Proposition~\ref{prop:lambda_parallel_v_connected} is an \emph{ex post} structural statement: if the optimal allocation
features enough strictly-interior overlap across active assets so $K=K_{\mathrm{int}}(x^\star)$ and $G(x^\star)$ is connected, then the optimal shadow prices for the active assets
must align with the factor loading. In large venues, such overlap is more plausible when many large accounts hold genuinely
cross-asset portfolios, creating interior reductions in multiple assets simultaneously.
\end{remark}

The alignment $\lambda^\star=\eta^\star v$ on interior active coordinates is the structural condition under which the
expected-loss ADL problem reduces to a one-dimensional (clipped) water-filling rule in factor-leverage space.

\begin{definition}[Connected Partial Deleveraging]\label{def:con.part.del}
We say an optimal allocation $x^\star$ satisfies the \emph{connected partial deleveraging condition} if $K_{\mathrm{int}}(x^\star)=K$ and the interior-coverage graph $G(x^\star)$ is connected.
\end{definition}

\begin{proposition}\label{prop:gen.water.filling.connected} Let
$(x^\star,\lambda^\star,\mu^\star,\nu^\star)$ satisfy the KKT conditions for \eqref{eq:adl_multiasset_exp_loss} and assume $x^\star$ satisfies the connected partial deleveraging condition of Definition~\ref{def:con.part.del}. Suppose further that if $Q^k\neq0$ then $v^k\neq0$.
Then there exists $\eta^\star\in\R$ such that for every account $i$,
\[
\ell_i^{(v)}(x_i^\star)\in \argmin_{z\in[\underline\ell_i,\overline\ell_i]}\ \{\psi(z)-\eta^\star z\},
\]
equivalently,
\[
\ell_i^{(v)}(x_i^\star)=\ell_i^\star(\eta^\star),
\]
where $\ell_i^\star(\eta)$ is the clipped factor water-filling target from
Theorem~\ref{thm:factor_filling_verification}.
\end{proposition}

\begin{proof}
By Proposition~\ref{prop:lambda_parallel_v_connected} and the assumptions $K_{\mathrm{int}}(x^\star)=K$ and connectivity of $G(x^\star)$,
there exists $\eta^\star\in\R$ such that
\begin{equation}\label{eq:lambda_equals_eta_v_on_K_proof}
\lambda^{k,\star}=\eta^\star v^k,\qquad \forall k\in K.
\end{equation}

Fix an account $i$. By Proposition~\ref{prop:market_clearing_kkt},
$x_i^\star$ minimizes the per-account Lagrangian
\begin{equation}\label{eq:per_account_subproblem_connected_proof}
\min_{x_i\in\mathcal Y_i}\ \Big\{f_i(x_i)+(\lambda^\star)^\top x_i\Big\}.
\end{equation}
Under the single-factor model,
\[
f_i(x_i)=E_i\,\psi\!\big(\ell_i^{(v)}(x_i)\big),
\qquad
\ell_i^{(v)}(x_i)=\frac{v^\top(q_i-x_i)}{E_i}.
\]

We partition coordinates into two disjoint sets:
\[
\mathcal Z\defeq\{k:\ Q^k=0\}, \qquad
K=\{k:\ Q^k\neq 0\}.
\]
Note that under the standing assumptions $K$ takes this simplified form.
By the directional bounds \eqref{eq:directional.bds}, $k\in\mathcal Z$ implies $x_i^k=0$ for all $x_i\in\mathcal Y_i$,
so $\sum_{k\in\mathcal Z}\lambda^{k,\star}x_i^k\equiv 0$ and can be dropped from \eqref{eq:per_account_subproblem_connected_proof}.
Thus, up to additive constants independent of $x_i$, the per-account objective in \eqref{eq:per_account_subproblem_connected_proof}
reduces to
\begin{equation}\label{eq:reduced_per_account_obj_proof}
E_i\,\psi\!\big(\ell_i^{(v)}(x_i)\big)+\sum_{k\in K}\lambda^{k,\star}x_i^k.
\end{equation}
Using \eqref{eq:lambda_equals_eta_v_on_K_proof}, we have
\[
\sum_{k\in K}\lambda^{k,\star}x_i^k
=
\eta^\star\sum_{k\in K}v^k x_i^k.
\]
Moreover, since $x_i^k=0$ for all $k\in\mathcal Z$ we obtain
\[
\sum_{k\in K}v^k x_i^k = \sum_{k=1}^d v^k x_i^k = v^\top x_i.
\]
Finally, by definition of $\ell_i^{(v)}$,
\[
v^\top x_i=v^\top q_i-E_i\,\ell_i^{(v)}(x_i).
\]
Substituting these identities into \eqref{eq:reduced_per_account_obj_proof} yields, up to the additive constant $\eta^\star v^\top q_i$,
\[
E_i\Big(\psi\!\big(\ell_i^{(v)}(x_i)\big)-\eta^\star \ell_i^{(v)}(x_i)\Big).
\]
Therefore minimizing \eqref{eq:per_account_subproblem_connected_proof} over $x_i\in\mathcal Y_i$ is equivalent to minimizing
\[
\psi(z)-\eta^\star z\qquad \text{over}\qquad z=\ell_i^{(v)}(x_i)\in[\underline\ell_i,\overline\ell_i].
\]
Hence
\[
\ell_i^{(v)}(x_i^\star)\in \argmin_{z\in[\underline\ell_i,\overline\ell_i]}\ \{\psi(z)-\eta^\star z\}.
\]
By uniqueness of $\ell_i^\star(\eta^\star)$ (Theorem~\ref{thm:factor_filling_verification}), we conclude that
$\ell_i^{(v)}(x_i^\star)=\ell_i^\star(\eta^\star)$ for all $i$.
\end{proof}

The need for verification disappears when ADL is required in only one
asset, because the vector clearing constraint reduces to a single scalar equation.

\begin{lemma}\label{lem:one_asset_bijection}
Assume $Q^k=0$ for $k\neq k_0$ and $v^{k_0}\neq0$, so $x_i^k=0$
for all $k\neq k_0$. Then for each account $i$,
\[
\ell_i^{(v)}(x_i)
=\frac{1}{E_i}\Big(\sum_{k\neq k_0} v^k q_i^k\Big)+\frac{v^{k_0}}{E_i}\bigl(q_i^{k_0}-x_i^{k_0}\bigr),
\]
and the map $x_i^{k_0}\mapsto \ell_i^{(v)}(x_i)$ is an affine bijection from $[l_i^{k_0},u_i^{k_0}]$ onto
$[\underline\ell_i,\overline\ell_i]$. Its inverse is
\begin{equation}\label{eq:inverse_map_one_asset}
x_i^{k_0}
=
q_i^{k_0}-\frac{E_i\,\ell-\sum_{k\neq k_0} v^k q_i^k}{v^{k_0}},
\qquad \ell\in[\underline\ell_i,\overline\ell_i].
\end{equation}
\end{lemma}

\begin{proof}
With $x_i^k=0$ for $k\neq k_0$, the displayed affine form follows directly from the definition
$\ell_i^{(v)}(x_i)=v^\top(q_i-x_i)/E_i$. Since $v^{k_0}\neq 0$, the coefficient of $x_i^{k_0}$ is nonzero, so the map is
injective and maps the compact interval $[l_i^{k_0},u_i^{k_0}]$ onto a compact interval, which must equal
$[\underline\ell_i,\overline\ell_i]$ by definition. Solving the affine relation for $x_i^{k_0}$ yields \eqref{eq:inverse_map_one_asset}.
\end{proof}

The following special case corresponds to \cref{thm:factor_filling_one_asset} in the main text.

\begin{corollary}\label{cor:one_asset_factor_filling_optimal}
Assume $Q^k=0$ for $k\neq k_0$ and $v^{k_0}\neq0$, so $x_i^k=0$
for all $k\neq k_0$.
 Let $\eta^\star$ satisfy the budget equation \eqref{eq:eta_budget_main}, and define targets
$\ell_i^\star(\eta^\star)$ by \eqref{eq:factor_fill_target_main}. Define an allocation $x^\star$ by setting
$x_i^{\star, k}=0$ for $k\neq k_0$ and
\[
x_i^{\star, k_0}
\;\defeq\;
q_i^{k_0}-\frac{E_i\,\ell_i^\star(\eta^\star)-\sum_{k\neq k_0} v^k q_i^k}{v^{k_0}},
\qquad i=1,\dots,n.
\]
Then $x^\star\in\mathcal X$ and $x^\star$ is optimal for \eqref{eq:adl_multiasset_exp_loss}.
\end{corollary}

\begin{proof}
By Lemma~\ref{lem:one_asset_bijection}, each target $\ell_i^\star(\eta^\star)\in[\underline\ell_i,\overline\ell_i]$ corresponds to a
unique $x_i^{\star, k_0}\in[l_i^{k_0},u_i^{k_0}]$, so the constructed $x^\star$ satisfies the per-account constraints.

It remains to verify the clearing constraint $\sum_i x_i^{\star, k_0}=Q^{k_0}$. Using the definition of $\ell_i^{(v)}$ and that
$x_i^{\star,k}=0$ for $k\neq k_0$, we compute
\[
\sum_{i=1}^n E_i\,\ell_i^{(v)}(x_i^\star)
=
\sum_{i=1}^n\sum_{k\neq k_0} v^k q_i^k
+
v^{k_0}\Big(\sum_{i=1}^n q_i^{k_0}-\sum_{i=1}^n x_i^{\star, k_0}\Big).
\]
On the other hand, since $Q^k=0$ for $k\neq k_0$,
\[
v^\top\Big(\sum_{i=1}^n q_i - Q\Big)
=
\sum_{i=1}^n\sum_{k\neq k_0} v^k q_i^k
+
v^{k_0}\Big(\sum_{i=1}^n q_i^{k_0}-Q^{k_0}\Big).
\]
Subtracting these equalities and using \eqref{eq:eta_budget_main} gives
$v^{k_0}\big(\sum_i x_i^{\star,k_0}-Q^{k_0}\big)=0$, hence $\sum_i x_i^{\star, k_0}=Q^{k_0}$ because $v^{k_0}\neq 0$.
Thus $x^\star\in\mathcal X$.

Finally, by construction $\ell_i^{(v)}(x_i^\star)=\ell_i^\star(\eta^\star)$ for all $i$, so Theorem~\ref{thm:factor_filling_verification}
implies that $x^\star$ is globally optimal.
\end{proof}

\begin{remark}
Corollary~\ref{cor:one_asset_factor_filling_optimal} formalizes the sense in which the single-asset case admits a true
water-filling solution in factor-leverage space: the scalar budget equation \eqref{eq:eta_budget_main} is then
equivalent to the (scalar) clearing constraint, so the water-filling targets are automatically implementable.
In contrast, when $|K|>1$, implementability can fail because a multi-dimensional clearing constraint cannot generally be
enforced by matching only one scalar target.
\end{remark}

\subsection{Single-Factor Analysis for More General Risk Measures}\label{sec:Cvar_one_factor}

This section briefly sketches how the single-factor analysis for the expected loss extends to more general risk measures. The crucial observation is that, in the single-factor case, comonotonicity leads to separability across accounts, even for risk measures such as $\mathrm{CVaR}_\beta$.

Assume the single-factor model $p_T=p_\tau+\epsilon v$ and recall that
\[
e_i(x_i,p_T)=E_i\bigl(1-\epsilon\,\ell_i^{(v)}(x_i)\bigr),
\qquad
\ell_i^{(v)}(x_i)\defeq \frac{v^\top(q_i-x_i)}{E_i}.
\]
We restrict attention to accounts with nonnegative factor exposure.
Namely, we impose $\ell_i^{(v)}(x_i)\ge 0$ as an additional feasibility constraint for allocations $x$ and restrict
to accounts with nonnegative factor exposure before the ADL event. We assume that, within this subset of accounts, the modified feasible set
\[
\mathcal X \,\cap\, \bigl\{x:\ \ell_i^{(v)}(x_i)\ge 0,\ \forall i\bigr\}
\]
is nonempty.
Then $\sigma_i(x_i,p_T)$ is a nondecreasing function
of the common scalar factor $\epsilon$ for each $i$, and hence the collection $\{\sigma_i(x_i,p_T)\}_{i=1}^n$ is \emph{comonotone}.

For any comonotone-additive risk functional $\rho$, including spectral risk measures such as
$\rho(\cdot)=\E[\cdot]$ or $\rho(\cdot)=\mathrm{CVaR}_\beta(\cdot)$, we then have 
\[
\rho\!\left(\sum_{i=1}^n \sigma_i(x_i,p_T)\right)=\sum_{i=1}^n \rho\!\left(\sigma_i(x_i,p_T)\right).
\]
Defining the one-dimensional factor-risk function
\[
\psi_{i,\rho}(\ell)\defeq \rho\!\left(E_i(\epsilon \ell-1)_+\right),\qquad \ell\ge 0,
\]
the exchange objective becomes separable across accounts, 
\[
\rho\!\left(\Lscr(x,p_T)\right)=\sum_{i=1}^n \psi_{i,\rho}\bigl(\ell_i^{(v)}(x_i)\bigr).
\]
This is analogous to~\eqref{eq:separable.across.accounts}, but now driven by comonotonicity rather than linearity of the expectation.

After this observation, the analysis becomes analogous to the expected loss in the single-factor model. As in~\eqref{eq:sum_factor_fixed}, the clearing constraint implies a fixed equity-weighted aggregate factor exposure,
\[
\sum_{i=1}^n E_i\,\ell_i^{(v)}(x_i)
=
\sum_{i=1}^n v^\top(q_i-x_i)
=
v^\top\Big(\sum_{i=1}^n q_i-Q\Big),\qquad \forall x\in\mathcal X.
\]
If the KKT multiplier satisfies $\lambda^\star=\eta^\star v$
(as implied, for instance, by an analogous interior-coverage\footnote{In the present setting, an ``interior'' condition must also account for the additional constraint $\ell_i^{(v)}(x_i)\ge 0$.} connectivity condition; cf.~Proposition \ref{prop:lambda_parallel_v_connected}), the corresponding per-account Lagrangian subproblem (cf.~\eqref{eq:account-primal} for the expected loss) depends
on $x_i$ only through $\ell_i^{(v)}(x_i)$ and reduces to the one-dimensional program
\[
\min_{\ell\in[\underline \ell_i,\overline \ell_i]\cap\R_+}\ \bigl\{\psi_{i,\rho}(\ell)-\eta^\star\,E_i\,\ell\bigr\},
\]
where $[\underline \ell_i,\overline \ell_i]=\{\ell_i^{(v)}(x_i):x_i\in\mathcal Y_i\}$ is the feasible factor-leverage interval.
If $\psi_{i,\rho}$ is differentiable and strictly convex on $\R_+$, the induced optimizer satisfies the clipped
water-filling rule
\[
\ell_i^\star(\eta^\star)=
\begin{cases}
\underline \ell_i^+, & \eta^\star\le \psi_{i,\rho}'(\underline \ell_i^+)/E_i,\\
(\psi_{i,\rho}')^{-1}(\eta^\star E_i), & \psi_{i,\rho}'(\underline \ell_i^+)/E_i<\eta^\star<\psi_{i,\rho}'(\overline \ell_i)/E_i,\\
\overline \ell_i, & \eta^\star\ge \psi_{i,\rho}'(\overline \ell_i)/E_i,
\end{cases}
\]
where $\underline \ell_i^+\defeq \max\{\underline \ell_i,0\}$. As was the case for the expected loss, even when the targets $\{\ell_i^\star(\eta)\}$ satisfy the scalar budget
\[
\sum_{i=1}^n E_i\,\ell_i^\star(\eta)=v^\top\Big(\sum_{i=1}^n q_i-Q\Big),
\]
a joint selection of allocations $x_i\in\mathcal Y_i$ achieving these targets and satisfying
$\sum_{i=1}^n x_i=Q$ need not exist in general because the clearing constraint $\sum_{i=1}^n x_i=Q$ is vector-valued. Nonetheless, as in Theorem~\ref{thm:factor_filling_one_asset}, when only a single asset is deleveraged, the same argument shows that the ``clipped water-filling'' solution in factor-leverage space is attainable by admissible reductions satisfying the clearing constraint and is therefore optimal.

\subsection{Numerical Example}\label{sec:appendix_numerical_example_crossmargin}

This appendix details the one-factor approximation of the bivariate GBM model in~\cref{sec:numerical_example_crossmargin}.  Note that we take $v$ from the price increment $\Delta p \defeq p_T-p_\tau$ of the GBM model, rather than from log returns, so that $v$ has the correct dollar units in the additive model $p_T=p_\tau+\epsilon v$. Under the model stated in~\cref{sec:numerical_example_crossmargin}, the covariance matrix of $\Delta p$ is

\begin{align*}
\Sigma_{\Delta p}
&=
\begin{pmatrix}
(p_\tau^{\mathrm{BTC}})^2\!\left(e^{(\sigma^{\mathrm{BTC}}_{\mathrm{ann}})^2\Delta}-1\right)
&
p_\tau^{\mathrm{BTC}}p_\tau^{\mathrm{ETH}}\!\left(e^{\rho\,\sigma^{\mathrm{BTC}}_{\mathrm{ann}}\sigma^{\mathrm{ETH}}_{\mathrm{ann}}\Delta}-1\right)
\\[0.4em]
p_\tau^{\mathrm{BTC}}p_\tau^{\mathrm{ETH}}\!\left(e^{\rho\,\sigma^{\mathrm{BTC}}_{\mathrm{ann}}\sigma^{\mathrm{ETH}}_{\mathrm{ann}}\Delta}-1\right)
&
(p_\tau^{\mathrm{ETH}})^2\!\left(e^{(\sigma^{\mathrm{ETH}}_{\mathrm{ann}})^2\Delta}-1\right)
\end{pmatrix}
\\
&\approx
\begin{pmatrix}
44{,}494{,}130.91 & 1{,}341{,}048.70\\
1{,}341{,}048.70 & 56{,}064.46
\end{pmatrix}.
\end{align*}
We then compute its principal eigenpair numerically via the symmetric eigendecomposition of $\Sigma_{\Delta p}$, normalizing the eigenvector to unit Euclidean norm and choosing the sign so that its BTC loading is positive. This gives
\[
\lambda_1 \approx 44{,}534{,}564.19,
\qquad
u_1 \approx
\begin{pmatrix}
0.99954578\\
0.03013680
\end{pmatrix},
\qquad
v \defeq \sqrt{\lambda_1}\,u_1
\approx
\begin{pmatrix}
6670.3910\\
201.1156
\end{pmatrix}.
\]

\begin{figure}[tbh]
\centering
\includegraphics[width=0.55\textwidth]{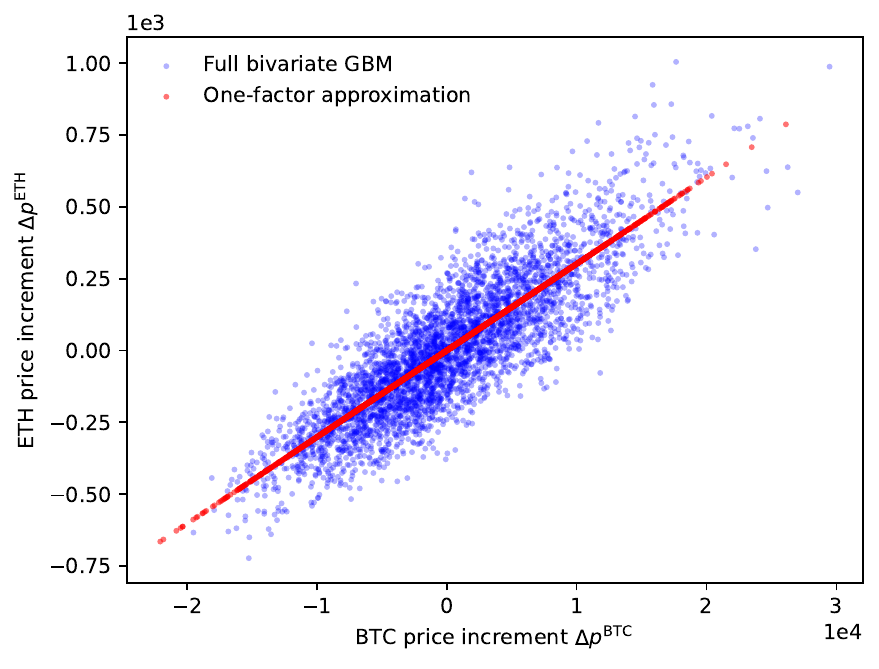}
\caption{Samples of price increments under the bivariate GBM model (blue) and the one-factor approximation collapsing the bivariate distribution onto the dominant covariance direction (red).}
\label{fig:crossmargin_increment_scatter}
\end{figure}


\section{Details for the Exchange-Scale Problem}
\subsection{ADMM: Algorithm and Convergence}\label{app:algorithm-admm}
Algorithm~\ref{alg:admm} summarizes the abstract two-block ADMM iteration.  Section~\ref{sec:admm-impl} explains how this abstract scheme is instantiated on the exchange-scale instance. The box is intentionally stated at the theory level: $\varepsilon_{\mathrm{stop}}$ denotes a generic residual tolerance, while $\rho$ and the iteration cap $K$ are algorithmic inputs whose concrete values depend on the application. In the exchange-scale implementation we use more detailed empirical calibration for the penalty parameter and stopping rule; see Appendix~\ref{app:exchange-scale-admm}.

\begin{algorithm}[H]
\caption{Abstract two-block ADMM with Sample-Average Approximation}
\label{alg:admm}
\begin{algorithmic}[1]
\Require
  aggregate target $Q\in\R^d$;
  per-account boxes $\{\Yscr_i\}_{i=1}^n$;
  SAA scenarios $\{\Delta p^{(s)}\}_{s=1}^S$;
  penalty $\rho>0$; tolerance $\varepsilon_{\mathrm{stop}}>0$;
  iteration cap $K$.
\Ensure accountwise iterate $x^K\in(\R^d)^n$,
  projected market-clearing iterate $z^K\in(\R^d)^n$, and
  primal/dual residuals $(r^K, s^K)$.
\medskip
\State Initialize $z^0$ (zero initialization or warm start), and set $u^0=0$.
\For{$k = 0, 1, \dots, K-1$}
  \For{$i=1,\dots,n$} \textbf{in parallel}
    \State $c_i^k \leftarrow z_i^k - u_i^k$;
    \State Solve per-account proximal subproblem~\eqref{eq:admm-x}.
  \EndFor
  \State $z^{k+1} \leftarrow \Pi_{\mathcal C}(x^{k+1}+u^k)$ \Comment{project onto the clearing set}
  \State $u^{k+1} \leftarrow u^k + x^{k+1} - z^{k+1}$
  \State $r^{k+1} \leftarrow x^{k+1}-z^{k+1}$ \Comment{primal residual}
  \State $s^{k+1} \leftarrow \rho(z^{k+1}-z^k)$ \Comment{dual residual}
  \If{$\norm{r^{k+1}}\le\varepsilon_{\mathrm{stop}}$
      \textbf{and} $\norm{s^{k+1}}\le\varepsilon_{\mathrm{stop}}$}
    \State \textbf{break}.
  \EndIf
\EndFor
\State \Return $x^\star \leftarrow x^{k+1}$, $z^\star \leftarrow z^{k+1}$
\Statex \hspace{\algorithmicindent} ($x^\star$ satisfies the per-account box constraints,
         $z^\star$ satisfies the market-clearing constraint $\sum_i z_i^\star = Q$ by construction, and primal
         feasibility is recovered up to the residual $\norm{x^\star-z^\star}$.)
\end{algorithmic}
\end{algorithm}

\begin{proposition}\label{prop:admm-convergence}
Fix $\rho>0$ and suppose the feasible set is nonempty. Then the
iterates generated by \eqref{eq:admm-x}--\eqref{eq:admm-u} satisfy
\[
\norm{x^k-z^k}\to 0,
\qquad
\norm{\rho(z^k-z^{k-1})}\to 0,
\]
and
\[
\sum_{i=1}^n f_i(x_i^k)+I_{\mathcal C}(z^k)\to L_S^\star,
\]
where $L_S^\star$ denotes the optimal value of the SAA problem. Moreover,
$\rho u^k$ converges to a dual optimal point.
\end{proposition}

\begin{proof}
Writing
\[
F(x)\defeq \sum_{i=1}^n f_i(x_i),
\qquad
G(z)\defeq I_{\mathcal C}(z),
\]
the problem takes the standard two-block ADMM form
\[
\min_{x,z}\;F(x)+G(z)
\qquad\text{subject to}\qquad x-z=0.
\]
Each $f_i$ is closed, proper, and convex, since it is the sum of a
finite convex piecewise-linear loss term and the indicator of the
corresponding box. Hence $F$ is closed, proper, and convex. The set
$\mathcal C$ is a nonempty closed affine set, so $G=I_{\mathcal C}$ is
closed, proper, and convex as well.

The feasible set is contained in the product of the per-account boxes
and is therefore compact. Since the objective is lower
semicontinuous, the primal optimum is attained. Because the problem is
a finite-dimensional convex program with polyhedral objective and affine
constraints, a saddle point of the unaugmented Lagrangian exists.
Therefore the assumptions of the standard two-block ADMM convergence
theorem in \citet[Section~3.2]{boyd2011distributed} are satisfied, and
the stated residual, objective, and dual convergence follow.
\end{proof}

\subsection{Data reconstruction}
\label{app:state-reconstruction}
For each wave time $\tau_j$, the empirical input is constructed in four
steps.
The cleaned reconstructed wave-level states used in the numerical
analysis are available in the accompanying GitHub repository.

First, we define the raw ADL wave directly from the Hyperliquid event
tape by restricting to trades whose direction is
Auto-Deleveraging in the corresponding event window. Throughout the
empirical analysis, we work on a 75-coin core universe: all ADL coins
are retained, and non-ADL coins are included only if their aggregate
gross notional exceeds USD $100{,}000$, so that the covariance
calibration and scenario generation remain supported on names with
sufficient trading activity. This determines
the raw realized ADL users, the raw set of affected coins, and the raw
wave notional.

Second, we reconstruct the pre-event account state strictly before
$\tau_j$ from Allium account snapshots. For each retained user--coin
pair, we take the most recent snapshot strictly before $\tau_j$ as the
baseline and roll it forward using the full trade tape.
Realized PnL accrued during this roll-forward is added back to the
account margin variable. User--coin pairs without a strict pre-event
baseline are excluded from the reconstructed state.

Third, we restrict the raw ADL flow in the wave window to the retained
reconstructed user--coin pairs and aggregate the retained fills. This
yields the realized Hyperliquid allocation
$x^{\mathrm{Hyperliquid},(j)}$ and the corresponding target vector
\[
Q^{(j)} \defeq \sum_i x_i^{\mathrm{Hyperliquid},(j)}.
\]
This is the step at which the optimization target becomes smaller than
the raw event total.

Fourth, we restrict the optimization problem to the active asset set
$\supp(Q^{(j)})$ and apply the directional-box reduction. Accounts whose
feasible interval has zero width on every active asset are removed
before solving. The resulting active-account instance is the final
optimization problem used in the empirical comparison. Concretely, the
solver is run on a $11{,}026\times 37$ instance at $\tau_1$ and on a
$7{,}185\times 4$ instance at $\tau_3$, rather than on the full
reconstructed ambient state.

Table~\ref{tab:data-reconstruction} summarizes these reductions for the
first wave $\tau_1$ and for the third wave $\tau_3$.

\begin{table}[t]
\centering
\caption{Data reconstruction summary for the first wave $\tau_1$ and
the third wave $\tau_3$. The reconstructed pre-event state is built on
the full 75-asset modeling universe, while the wave-specific retained
target $Q^{(j)}$ is supported only on the coins that remain active after
matching the raw ADL flow to reconstructable user--coin pairs.}
\label{tab:data-reconstruction}
\begin{tabular}{lcc}
\toprule
Quantity & $\tau_1$ & $\tau_3$ \\
\midrule
Raw ADL users in wave & $8{,}048$ & $2{,}384$ \\
Raw ADL coins in wave & $37$ & $4$ \\
Raw ADL notional & $\$1.7686\times 10^8$ & $\$3.8661\times 10^8$ \\
ADL users with reconstructable pre-event state & $6{,}956$ & $1{,}961$ \\
Accounts in reconstructed pre-event state & $60{,}324$ & $59{,}978$ \\
Users with positive retained ADL in $x^{\mathrm{Hyperliquid},(j)}$
& $6{,}820$ & $1{,}814$ \\
Retained target coins $\lvert \supp(Q^{(j)}) \rvert$ & $37$ & $4$ \\
Retained target notional $\sum_k p_{\tau,k}\lvert Q_k^{(j)}\rvert$
& $\$1.5390\times 10^8$ & $\$3.4606\times 10^8$ \\
Active optimization accounts & $11{,}026$ & $7{,}185$ \\
\bottomrule
\end{tabular}
\end{table}

Two distinctions are worth emphasizing. First, the row
``ADL users with reconstructable pre-event state'' is a user-level
coverage statistic, whereas the row
``Users with positive retained ADL in
$x^{\mathrm{Hyperliquid},(j)}$'' reflects the stricter pair-level
filtering applied when constructing the retained realized allocation.
Second, the row ``Active optimization accounts'' refers to the final
directional-box reduction used before solving. It contains all accounts
with positive retained ADL in $x^{\mathrm{Hyperliquid},(j)}$, together
with additional reconstructable accounts that were not delevered in the
realized event but still hold positions on $\supp(Q^{(j)})$ that can be
delevered under the directional box constraints. This is why the active
optimization set is larger than the support of
$x^{\mathrm{Hyperliquid},(j)}$, while still being smaller than the full
reconstructed pre-event state.


\subsection{Covariance Estimation}
\label{sec:covariance}

ADL events occur in precisely the market conditions where volatility clustering is most severe: sharp directional moves concentrated in short windows. Crypto assets are known to exhibit pronounced time-varying volatility, and conditioning on the current volatility level at the trigger time $\tau$ is therefore essential for a meaningful risk assessment. An unconditional covariance estimate would understate risk in stressed regimes and overstate it in calm ones. Figure~\ref{fig:volatility} illustrates this clustering for the assets in the ADL universe: rolling volatility estimates vary substantially over time, with spikes that coincide with periods of market stress. This motivates the Exponentially Weighted Moving Average (EWMA)-standardized decomposition below.

\begin{figure}[t]
\centering
\includegraphics[width=0.95\textwidth]{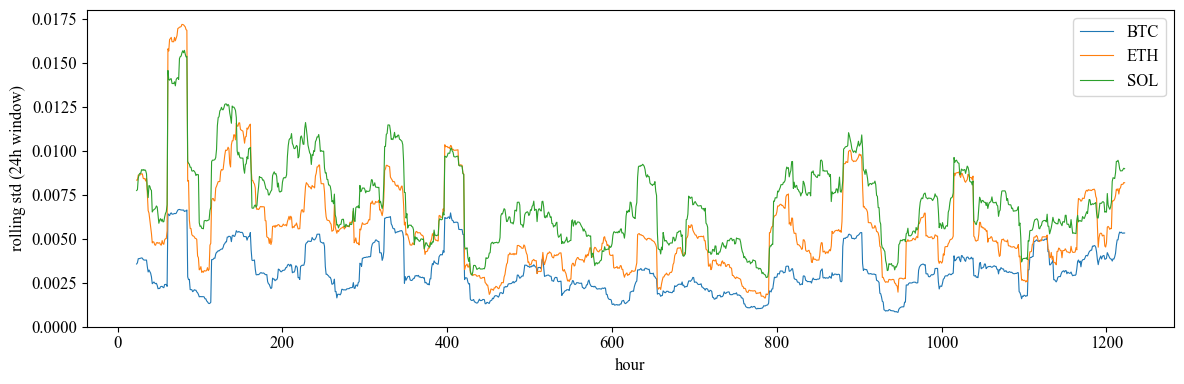}
\caption{Rolling-window volatility estimates for selected assets.
  Pronounced clustering motivates conditioning on the volatility
  level at the ADL trigger time $\tau$ rather than using an
  unconditional estimate.}
\label{fig:volatility}
\end{figure}

\paragraph{EWMA-standardized covariance model.}
We decompose hourly log-returns as $r_t = D_t \varepsilon_t$, where $D_t$ is a diagonal matrix of time-varying asset volatilities and $\varepsilon_t$ is the standardized residual vector, whose covariance $\Sigma_\varepsilon = \mathrm{Cov}(\varepsilon_t)$ is approximately stationary over time. The conditional
1-hour log-return covariance at the ADL trigger time $\tau$ is then 
\begin{equation}\label{eq:ewma-cov-main}
\Sigma_{\tau,1\mathrm{h}}^{\log} = D_\tau \Sigma_\varepsilon D_\tau,
\end{equation}
where $D_\tau$ is the EWMA volatility diagonal at time $\tau$ \citep{engle1982autoregressive}. This separates the two empirical ingredients relevant for ADL risk: the slowly varying cross-sectional dependence structure $\Sigma_\varepsilon$, and the rapidly changing volatility level $D_\tau$ that must be captured at the trigger time. The matrix $\Sigma_\varepsilon$ is approximated by a diagonal-plus-rank-one structure, as described in the empirical support paragraph below and used in the price model inputs of \Cref{sec:pm_inputs}.

Since hourly return autocorrelation is empirically weak, we scale the
hourly covariance linearly across horizons. Expressed in annualized GBM
units, this gives
\[
\Sigma_{\tau,\Delta}^{\log}
=
\Delta\,(24\times 365)\,\Sigma_{\tau,1\mathrm{h}}^{\log},
\]
which is consistent with the empirical absence of meaningful serial
correlation in hourly returns once the covariance matrix has been
annualized. Throughout the exchange-scale experiments, we set
$
\Delta = 10/365,
$
corresponding to a ten-day liquidation horizon under annualized inputs.

\paragraph{Price simulation models.}
To evaluate the SAA objective~\eqref{eq:saa-fi}, we sample $S$
terminal price vectors $p_T^{(s)}$ from a multivariate geometric
Brownian motion calibrated to the horizon-$\Delta$ log-return
covariance matrix $\Sigma_{\tau,\Delta}^{\log}$. Writing
\[
X \defeq \log(p_T/p_\tau),
\]
we model the horizon-$\Delta$ log-return vector as
\[
X \sim \mathcal N\!\left(
 -\tfrac12 \operatorname{diag}(\Sigma_{\tau,\Delta}^{\log}),
 \ \Sigma_{\tau,\Delta}^{\log}
\right),
\]
where $\operatorname{diag}(\Sigma_{\tau,\Delta}^{\log})\in\R^d$
denotes the vector of diagonal entries of
$\Sigma_{\tau,\Delta}^{\log}$. Equivalently, if $L_\tau$ is any
matrix square root satisfying
\[
L_\tau L_\tau^\top = \Sigma_{\tau,\Delta}^{\log},
\]
and $Z\sim\mathcal N(0,I_d)$, then
\begin{equation}\label{eq:gbm-main}
X = -\tfrac12 \operatorname{diag}(\Sigma_{\tau,\Delta}^{\log}) + L_\tau Z,
\qquad
p_T = p_\tau \odot \exp(X),
\end{equation}
where $\odot$ denotes componentwise multiplication. This
representation makes the cross-asset covariance structure explicit.

For implementation, the same covariance matrix may equivalently be
written in volatility-correlation form as
\[
\Sigma_{\tau,\Delta}^{\log}
=
\Delta\,\operatorname{diag}(\sigma_\tau)\,R_\tau\,\operatorname{diag}(\sigma_\tau),
\]
where $\sigma_\tau\in\R_+^d$ is the vector of annualized conditional
volatilities at time $\tau$ and $R_\tau$ is the corresponding
correlation matrix.

The water-filling algorithm additionally requires a factor direction
$v\in\R^d$ defining the dominant covariance mode. This is extracted
from the same GBM calibration. Defining the horizon-$\Delta$ dollar price drop by $\Delta p \defeq p_\tau - p_T,$
we have the first-order approximation
\[
\Delta p \approx -\operatorname{diag}(p_\tau)\,X.
\]
This yields the dollar price-drop covariance
\begin{equation}\label{eq:sigma-dp-main}
\Sigma_{\tau,\Delta}^{\Delta p}
\approx
\operatorname{diag}(p_\tau)\,\Sigma_{\tau,\Delta}^{\log}\,\operatorname{diag}(p_\tau),
\end{equation}
and the leading eigenpair $(\lambda_1, u_1)$ of
$\Sigma_{\tau,\Delta}^{\Delta p}$ gives
$v \defeq \sqrt{\lambda_1}\,u_1$.
This vector enters the water-filling solver as the one-factor
direction defining factor leverage
(cf.\ \Cref{def:factor_leverage}).

\paragraph{Empirical support for the one-factor approximation.}
Figure~\ref{fig:cov-stat} shows the stationary covariance matrix
$\Sigma_\varepsilon$ and its diagonal-plus-rank-one approximation
\[
\Sigma_\varepsilon^{(1)}
=
\lambda_1^{\varepsilon} u_1^{\varepsilon} (u_1^{\varepsilon})^\top
+
\operatorname{diag}\!\bigl(\Sigma_\varepsilon - \lambda_1^{\varepsilon} u_1^{\varepsilon} (u_1^{\varepsilon})^\top\bigr),
\]
where $(\lambda_1^{\varepsilon}, u_1^{\varepsilon})$ is the leading eigenpair of
$\Sigma_\varepsilon$. The leading principal component explains
approximately $56\%$ of the total variance of $\Sigma_\varepsilon$.
Despite this moderate explained-variance share, the rank-one
approximation reproduces the main cross-asset dependence patterns
accurately: the relative Frobenius error is approximately $5\%$ in
covariance space and $4\%$ in correlation space. The residual
variation is predominantly idiosyncratic, concentrated on the
diagonal rather than in additional common factors. This supports the
use of the one-factor model as a reduced-form benchmark and justifies
the rank-one plus idiosyncratic structure used for $\Sigma_\varepsilon$.

\begin{figure}[t]
\centering
\includegraphics[width=\textwidth]{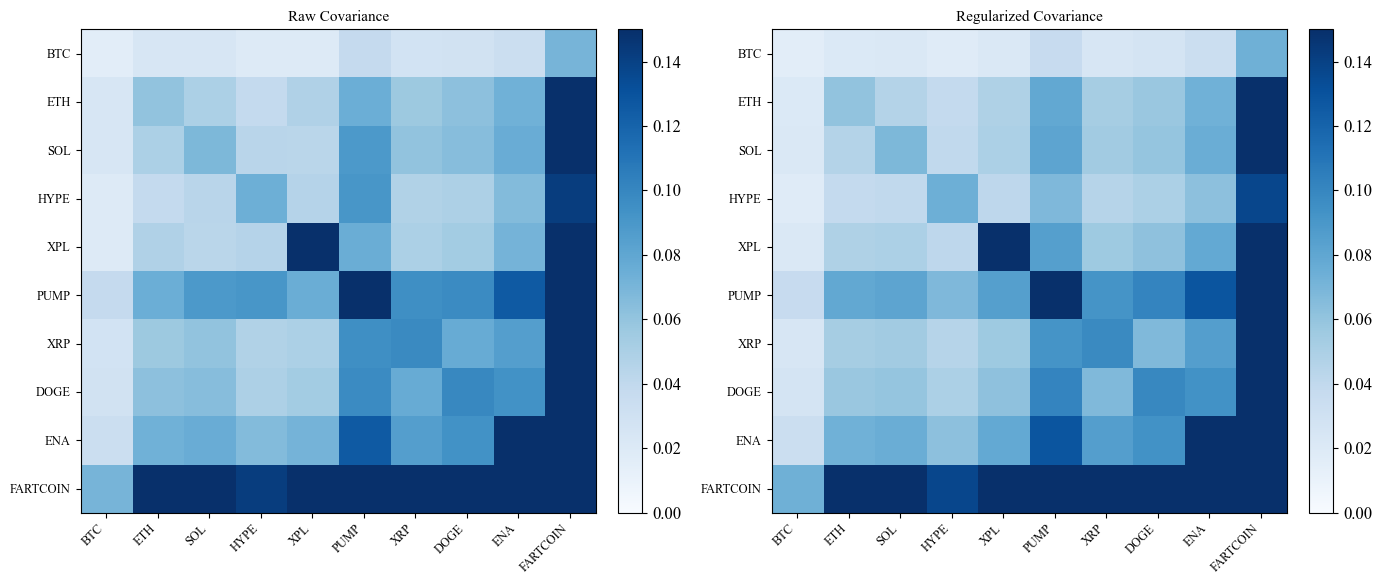}
\caption{Left: stationary covariance matrix $\Sigma_\varepsilon$ of
  EWMA-standardized residuals. Right: diagonal-plus-rank-one
  approximation $\Sigma_\varepsilon^{(1)}$, preserving the original
  diagonal and retaining only the leading principal component in the
  off-diagonal structure. The relative Frobenius error is
  approximately $5\%$, with residual variation predominantly
  idiosyncratic.}
\label{fig:cov-stat}
\end{figure}

\subsection{Exchange-Scale ADMM Implementation Details}
\label{app:exchange-scale-admm}

\paragraph{Active-set restriction.}
The ADMM iteration is never run on the full reconstructed ambient array.
For a given wave, we first restrict the optimization to the active
asset set $\supp(Q)$ and then remove every account whose directional
box has zero width on that set. This reduction is exact at the level of
the decision variables: accounts that cannot contribute in any active
asset are fixed at zero in every feasible allocation. The
scenario-dependent shortfall intercept in \eqref{eq:dollar-reparam},
however, still incorporates the PnL of the retained non-active
positions through the term $-q_i^\top\Delta p^{(s)}$, so the
optimization is carried out on the active subset while preserving the
contribution of the remaining retained positions to the shortfall
threshold.

\paragraph{Dollar reparameterization.}
Different assets have price scales that differ by several orders of magnitude, which causes
numerical ill-conditioning in the raw position-unit formulation. We therefore reparameterize the
decision variable componentwise as
\[
y_i^k \defeq p_\tau^k\,x_i^k,
\]
so that $y_i^k$ is the dollar notional removed from account $i$ in asset $k$. Defining
\[
E_i \defeq m_i + q_i^\top(p_i^{(e)}-p_\tau),
\qquad
a^{(s)} \defeq \Delta p^{(s)} \oslash p_\tau,
\]
where $\oslash$ denotes componentwise division, the SAA program is written equivalently in dollar
units as
\begin{equation}\label{eq:dollar-reparam}
\begin{array}{ll}
\minimize_{y} &
\displaystyle
\frac{1}{S}\sum_{s=1}^S \sum_{i=1}^n
\Bigl(-E_i - q_i^\top\Delta p^{(s)} + (a^{(s)})^\top y_i\Bigr)_+ \\[8pt]
\subjectto &
\displaystyle \sum_{i=1}^n y_i = \operatorname{diag}(p_\tau)\,Q, \\[6pt]
& y_i^k\in
[\,p_\tau^k\,l_i^k,\; p_\tau^k\,u_i^k\,],
\quad i=1,\dots,n,\ k\in\supp(Q).
\end{array}
\end{equation}
This transformation is an exact equivalence: problem~\eqref{eq:dollar-reparam} has the same
optimal value as~\eqref{eq:admm-split}, and its optimal $y^\star$ recovers $x^\star$ by
componentwise division. In dollar units, the decision variables, clearing constraints, and loss
coefficients all lie on a common numerical scale, which substantially improves numerical
conditioning.

\paragraph{Support-aware per-account decomposition.}
Because the $x$-update~\eqref{eq:admm-x} separates exactly across accounts, the dollar-variable
$y$-update does as well, and each per-account subproblem is solved on its own active support ---
the set of assets in which account $i$ has nonzero feasible movement width. The majority of
accounts are movable in exactly one asset and reduce to a scalar proximal problem. The remaining
multi-asset accounts are grouped by their support pattern; within each group the subproblems share
the same constraint structure and differ only in the proximal center and the scenario right-hand
sides. This turns the accountwise ADMM step from one large heterogeneous optimization problem into
a collection of many small low-dimensional subproblems.

\paragraph{Parallel implementation of the accountwise proximal step.}
In each ADMM iteration, the computational core is the per-account proximal $x$-update~\eqref{eq:admm-x},
or equivalently the dollar-variable $y$-update described above. Because this step still separates
fully across accounts after the active-set reduction, the implementation distinguishes only between
two solver types: accounts movable in a single asset are updated by the fast scalar routine,
whereas accounts movable in multiple assets are updated by small exact per-account quadratic
programs. For the multi-asset accounts, constructing a fresh optimization model for every solve
would create substantial overhead, so accounts are grouped by their active support pattern and one
cached quadratic-program model is reused for each pattern. For a given pattern, the quadratic part
of the proximal objective is fixed once, so across ADMM iterations and across accounts sharing that
pattern only the account-specific linear objective coefficients, the scenario right-hand side, and
the box bounds are updated before re-solving. These multi-asset solves are then dispatched
concurrently across worker tasks and warm-started from the previous iterate, while the only global synchronization point in each ADMM iteration is the projection step that restores market clearing.
The long exchange-scale runs themselves are initialized from the saved
directional sequential Water-filling allocation, whereas the one-factor
Water-filling allocation is retained as the main constructive
benchmark in the loss comparisons.

\paragraph{Penalty parameter calibration.}
The augmented-Lagrangian penalty $\rho$ controls the relative weight of the market-clearing
residual against the per-account objectives. Because the SAA objective is measured in dollars and
the clearing constraint involves dollar notionals of order $\sim\!10^8$, the numerically
effective range of $\rho$ lies well below standard defaults. We determine $\rho$ by a broad
logarithmic search; the exchange-scale runs reported in \Cref{sec:oct10-results} use
$\rho=5\times 10^{-6}$ together with absolute and relative stopping
tolerances $\varepsilon_{\mathrm{stop}}=10^{-6}$,
and we evaluate the final ADMM iterate produced at termination.

\subsection{Risk-Reduction Robustness and Convergence}
\label{app:rr-robustness}
The following tables and figures collect the risk-reduction diagnostics
deferred from the main text. They serve three purposes. First, they
report the third wave $\tau_3$ using the same
definitions as the main $\tau_1$ comparison. Second, they compare
one-factor Water-filling with the directional coin-by-coin
benchmark inspired by the isolated-margin ranking. Third, they show both
Monte Carlo robustness and the evolution of the relative
risk-reduction statistic along the fixed-sample ADMM trajectories.

\begin{table}[H]
\centering
\small
\begin{tabular}{p{1.2cm} p{4cm} >{\raggedleft\arraybackslash}p{2.5cm} >{\raggedleft\arraybackslash}p{2.0cm}}
\toprule
\textbf{Wave} & \textbf{Allocation} & \textbf{Sampled Loss $L_S$} & \textbf{Relative RR} \\
\midrule
& No allocation & $5.218\times 10^8$ & $0.000$ \\
& Hyperliquid & $5.175\times 10^8$ & $0.354$ \\
& Isolated Water-Filling & $5.130\times 10^8$ & $0.732$ \\
$\tau_1$ & One-factor Water-filling & $5.114\times 10^8$ & $0.871$ \\
& Risk-optimal & $5.098\times 10^8$ & $1.000$ \\
\midrule
& No allocation & $5.291\times 10^8$ & $0.000$ \\
& Hyperliquid & $5.178\times 10^8$ & $0.568$ \\
& Isolated Water-Filling & $5.115\times 10^8$ & $0.886$ \\
$\tau_3$ & One-factor Water-filling & $5.106\times 10^8$ & $0.931$ \\
& Risk-optimal & $5.092\times 10^8$ & $1.000$ \\
\bottomrule
\end{tabular}
\caption{Sampled-loss comparison on the fixed reference scenario set for $\tau_1$
  and for the third wave $\tau_3$. The first row in
  each panel reports the zero-allocation baseline $L_S(\mathbf 0)$. The
  directional coin-by-coin benchmark is reported only as a robustness
  check. In both waves, one-factor Water-filling dominates the
  directional heuristic under the updated relative-risk-reduction
  metric.}
\label{tab:hyperliquid-wave-robustness}
\end{table}

Table~\ref{tab:hyperliquid-wave-robustness} records the fixed-sample
comparison on the reference scenario set used to compute the risk-optimal allocation.
It confirms that the one-factor sequential benchmark is the stronger
constructive comparison. At $\tau_1$, replacing the isolated margin/directional ranking
by one-factor Water-filling raises the relative risk reduction
from $0.732$ to $0.871$. At $\tau_3$, the same
refinement raises the relative risk reduction from $0.886$ to $0.931$.
The directional heuristic is therefore retained only as a robustness
benchmark.

To gauge the Monte Carlo sensitivity of the reported $\tau_1$ losses, we keep the candidate allocations fixed and re-evaluate them over $50$ independent $S=128$ scenario sets under the same calibrated model, using evaluation seeds $400,\dots,419,421,\dots,450$ so that the optimization seed $420$ is excluded. The corresponding means and standard deviations of the out-of-sample losses are reported in \cref{sec:oct10-results}; relative to those averages, the fixed reference scenario set in Table~\ref{tab:hyperliquid-wave-robustness} yields losses that are uniformly lower by about $2.1\times 10^7$ for all candidate allocations, while leaving the ordering in risk reduction of the allocations unchanged.

\begin{figure}[H]
\centering
\includegraphics[width=0.95\textwidth]{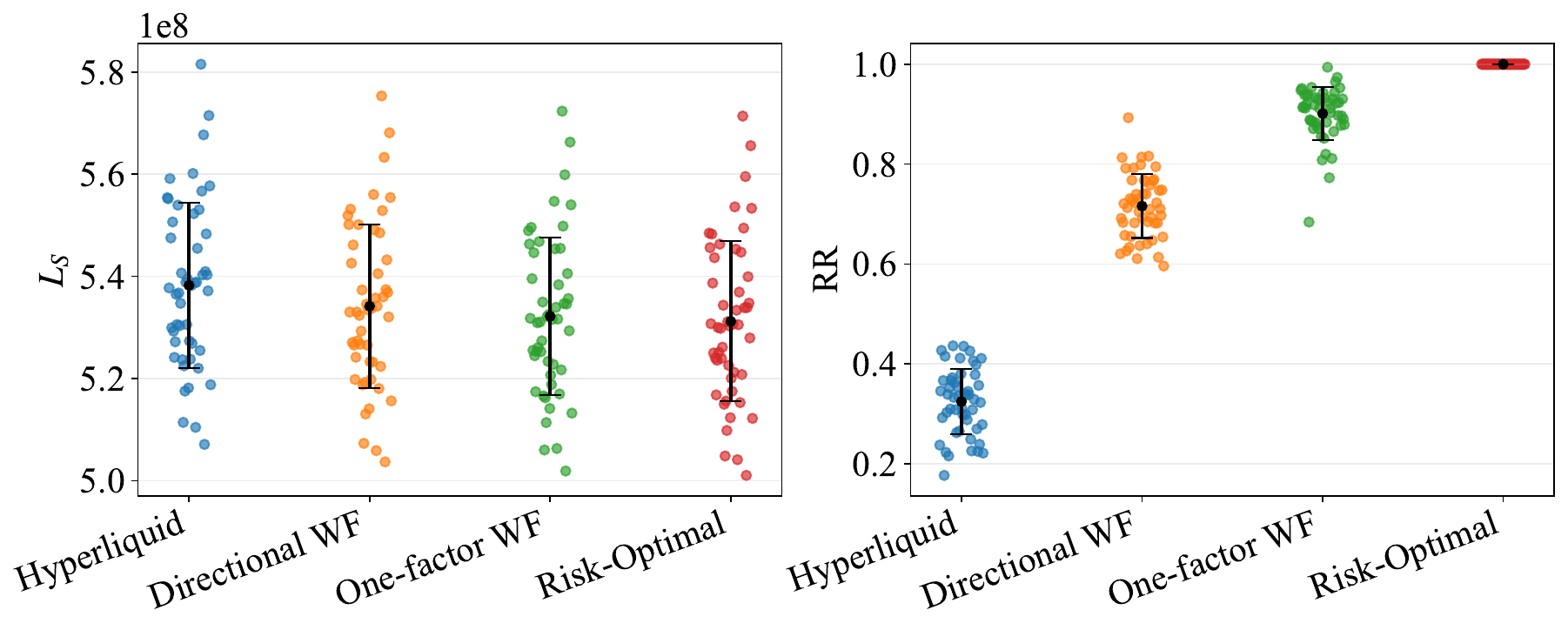}
\caption{Monte Carlo seed-robustness check for the fixed $\tau_1$
  allocations at $S=128$. Each point is a full-state re-evaluation of
  the sampled loss or the relative risk-reduction statistic on an
  independent scenario set under the same calibrated model; the black
  markers report the seed mean with one standard-deviation error bars.
  The spread is economically visible at this small scenario count, but
  the ordering
  $\mathrm{Risk\text{-}optimal} < \mathrm{One\text{-}factor\ WF} <
  \mathrm{Directional\ WF} < \mathrm{Hyperliquid}$ is stable across
  seeds.}
\label{fig:tau1-seed-robustness}
\end{figure}

\paragraph{Convergence.}
\label{app:rr-convergence}

Figure~\ref{fig:rr-convergence} plots the relative risk-reduction
statistic along the ADMM iterate sequence. The initial point in each
curve is the Water-filling warm start. At each iteration,
we evaluate the induced sampled loss of the projected feasible
allocation on the same fixed reference scenario set as in
Table~\ref{tab:hyperliquid-wave-robustness}, and then map it into the
statistic $\mathrm{RR}$ defined in \cref{eq:rr_def}. The horizontal
axis is the ADMM iteration count. The resulting trajectories show that
the relative risk-reduction statistic approaches its terminal level
quickly: at
iteration $500$ it already equals $0.952$ at $\tau_1$ and $0.887$ at
$\tau_3$, and at iteration $1{,}000$ it reaches $0.966$ and $0.917$,
respectively. The $\tau_1$ run reaches the stopping rule before the
iteration cap. The reported $\tau_3$ trajectory ends at the cap
$K=20{,}000$ rather than at the formal tolerance, with terminal
residuals $\|r^K\|_2 \approx 5.19$ and $\|s^K\|_2 \approx 8.0\times
10^{-3}$. Hence most of the reduction in sampled loss is attained
well before termination; the remaining iterations primarily improve the
ADMM residuals rather than materially changing the induced allocation.

\begin{figure}[H]
\centering
\includegraphics[width=0.95\textwidth]{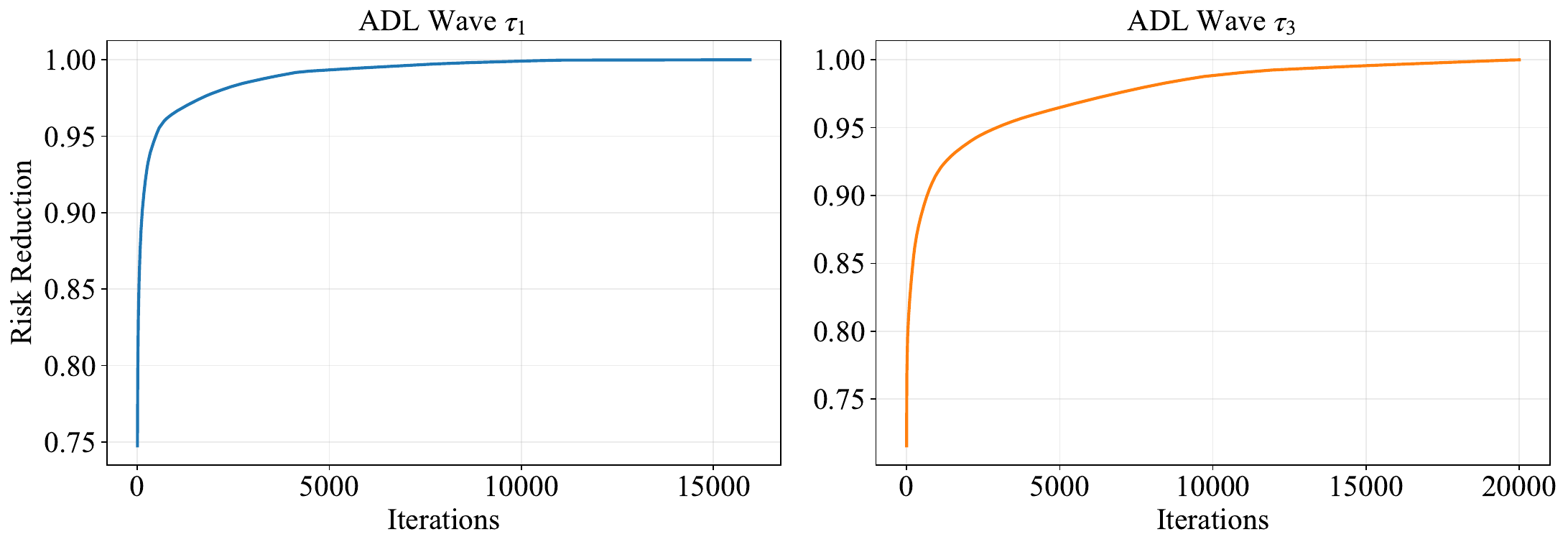}
\caption{Relative-risk-reduction trajectories along the Risk-optimal
  runs on the fixed reference scenario sets for $\tau_1$ and $\tau_3$. The curves show that
  most of the improvement arrives in the first few hundred iterations,
  while the remaining runtime is largely spent reducing ADMM residuals.
  Each curve starts from the isolated water-filling warm start and
  evaluates the projected feasible allocation on the fixed reference scenario
  set for that wave.}
\label{fig:rr-convergence}
\end{figure}

\subsection{User-Level Characteristics}
\label{app:user-level-tables}
The following tables and figures provide the additional user-level
diagnostics deferred from the analysis in \cref{sec:oct10-results}. They report thresholded
user counts for $\tau_1$ and $\tau_3$ and then summarize the analogous
$\tau_3$ four-panel diagnostics using the same definitions as the main
$\tau_1$ comparison.

\begin{table}[H]
\centering
\small
\begin{tabular}{p{1.8cm} >{\raggedleft\arraybackslash}p{2.2cm} >{\raggedleft\arraybackslash}p{3.0cm} >{\raggedleft\arraybackslash}p{2.6cm}}
\toprule
 \textbf{Threshold} &
\textbf{Hyperliquid} &
\textbf{One-factor WF} & \textbf{Risk-optimal} \\
\midrule
 $>0$      & $6{,}820$ & $9{,}480$ & $2{,}794$ \\
 $>\$1$    & $6{,}780$ & $1{,}409$ & $1{,}920$ \\
 $>\$10$   & $6{,}706$ & $1{,}292$ & $1{,}749$ \\
 $>\$100$  & $5{,}988$ & $773$ & $1{,}303$ \\
 $>\$1$k   & $2{,}416$ & $448$ & $874$ \\
 $>\$10$k  & $739$     & $246$ & $436$ \\
 $>\$100$k & $170$     & $92$  & $154$ \\
\bottomrule
\end{tabular}
\caption{Thresholded user-count comparison at $\tau_1$. At low
  thresholds, one-factor Water-filling and the Risk-optimal allocation
  touch many more users than the realized Hyperliquid event, but much
  of this mass disappears once even a one-dollar cutoff is imposed.
  Above $>\$100$k, one-factor Water-filling concentrates more notional
  on fewer users than either Hyperliquid or the Risk-optimal
  allocation.}
\label{tab:hyperliquid-wave-threshold-long}
\end{table}

\begin{table}[H]
\centering
\small
\begin{tabular}{p{1.8cm} >{\raggedleft\arraybackslash}p{2.2cm} >{\raggedleft\arraybackslash}p{3.0cm} >{\raggedleft\arraybackslash}p{2.6cm}}
\toprule
 \textbf{Threshold} &
\textbf{Hyperliquid} &
\textbf{One-factor WF} & \textbf{Risk-optimal} \\
\midrule
 $>0$      & $1{,}814$ & $3{,}913$ & $2{,}805$ \\
 $>\$1$    & $1{,}811$ & $847$ & $1{,}044$ \\
 $>\$10$   & $1{,}765$ & $755$ & $917$ \\
 $>\$100$  & $1{,}344$ & $579$ & $695$ \\
 $>\$1$k   & $650$     & $376$ & $437$ \\
 $>\$10$k  & $237$     & $203$ & $246$ \\
 $>\$100$k & $88$      & $91$  & $98$ \\
\bottomrule
\end{tabular}
\caption{Thresholded user-count comparison at $\tau_3$. The same dust
  pattern is visible at low thresholds: Hyperliquid changes almost not
  at all between $>0$ and $>\$1$, whereas one-factor Water-filling and
  the Risk-optimal allocation lose most of their touched users once a
  one-dollar cutoff is imposed. In the far upper tail, the ordering now
  reverses relative to $\tau_1$: above $>\$100$k, Hyperliquid touches
  slightly fewer users than one-factor Water-filling.}
\label{tab:hyperliquid-wave-threshold-long-t3}
\end{table}

Tables~\ref{tab:hyperliquid-wave-threshold-long}
and~\ref{tab:hyperliquid-wave-threshold-long-t3} make the dust-allocation
pattern explicit. At both $\tau_1$ and $\tau_3$, the realized
Hyperliquid allocation changes very little between the $>0$ and
$>\$1$ thresholds, whereas one-factor Water-filling and the
Risk-optimal allocation touch many more users in total but much more
lightly. At $\tau_1$, one-factor Water-filling also carries more dollar
notional on fewer large users, as seen most clearly in the $>\$100$k
row. At $\tau_3$, this upper-tail ordering no longer holds: the roles
of Hyperliquid and one-factor Water-filling are essentially reversed.

Figure~\ref{fig:comparison-stats-t3} reports the analogous four-panel
diagnostics for $\tau_3$. The upper-left and upper-right panels recover
the same concentration facts in graphical form. At $\tau_3$,
Hyperliquid touches only $1{,}814$ users overall,
compared with $3{,}913$ under one-factor Water-filling and $2{,}805$
under the Risk-optimal allocation, but its largest-user deleveraging
notional rises to $\$1.7\times 10^8$, compared with
$\$1.0\times 10^8$ for the Risk-optimal allocation and
$\$6.9\times 10^7$ for one-factor Water-filling. Its top ten users
already account for $87.5\%$ of total deleveraging, compared with
$72.3\%$ for the Risk-optimal allocation and $76.0\%$ for one-factor
Water-filling. Thus the upper-left and upper-right panels show that, at
$\tau_3$, the realized Hyperliquid allocation is the most concentrated
of the three in the upper tail.

The lower-left panel looks qualitatively similar to $\tau_1$:
Hyperliquid places more mass near full closeout, whereas one-factor
Water-filling and the Risk-optimal allocation usually remove more
moderate fractions of the deleverageable portfolio. The lower-right
panel is more ordered for Hyperliquid than at $\tau_1$, while
one-factor Water-filling and the Risk-optimal allocation retain
approximately the same monotone factor-leverage targeting pattern.

\begin{figure}[H]
\centering
\includegraphics[width=0.95\textwidth]{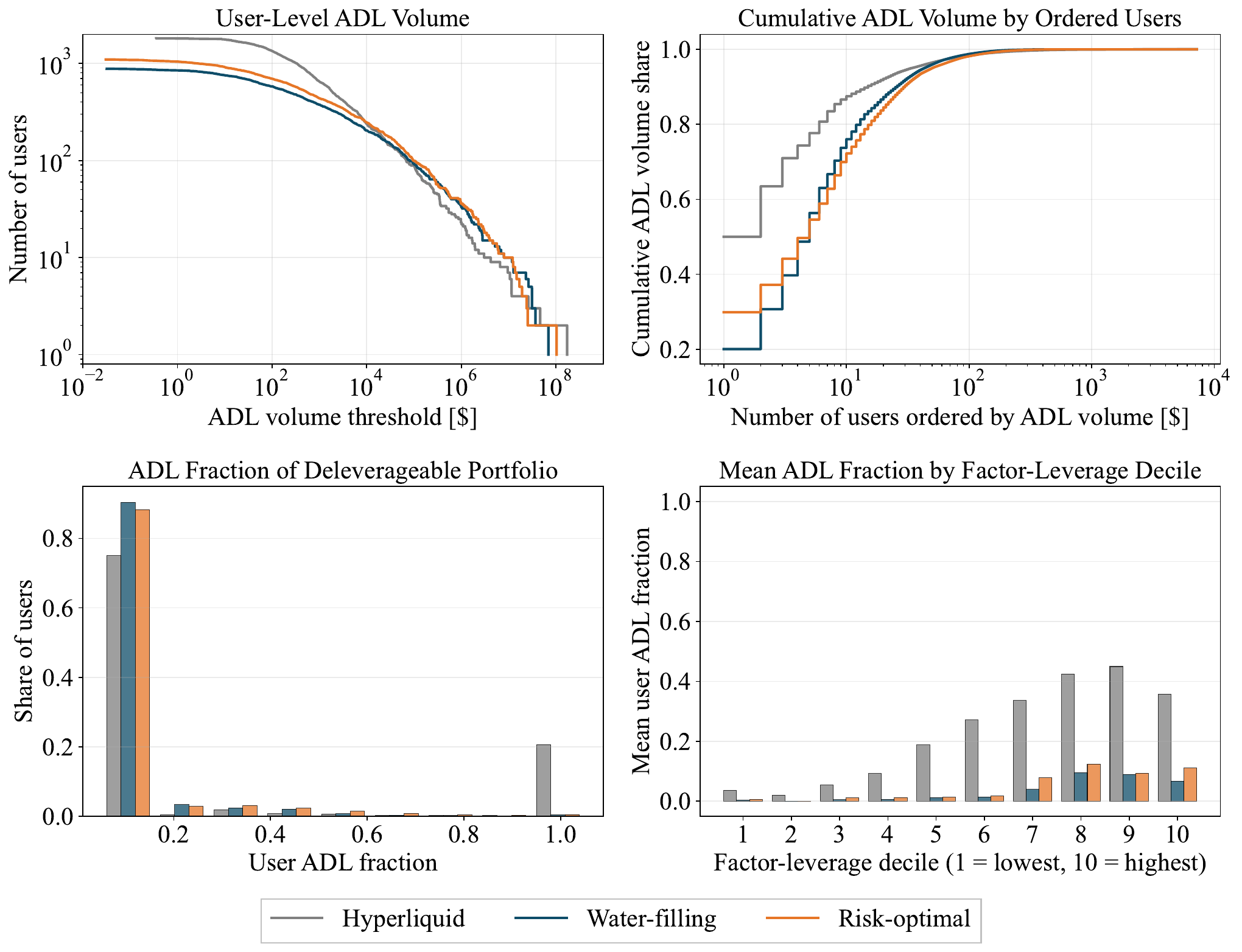}
\caption{User-level diagnostics for $\tau_3$: user count above
  deleveraging-notional threshold, cumulative top-share concentration
  curve, histogram of deleveraging fractions, and factor-leverage-decile
  targeting plot.}
\label{fig:comparison-stats-t3}
\end{figure}


\end{document}